%% file: dh_ebn.tex
\documentclass
{article}

\newif\ifamsart
\amsartfalse 
\newif\ifllncs
\llncsfalse 
\newif\ifpubrel
\pubrelfalse            

\usepackage{smallsubsub}

\input{macros}

\ifpubrel
\usepackage{fancyhdr}
\fancypagestyle{title}{%
  \fancyhf{}%
  \fancyhead[L]{Approved for Public Release; Distribution Unlimited. Case Number 18-0067.}%
  \fancyfoot[C]{This technical data was produced for the
    U.S. Government\\ under Contract. No. W15P7T-13-C-A802, and is\\
    subject to the Rights in Technical Data-Noncommercial Items clause\\
    at DFARS 252.227-7013 (FEB 2012)\\
    \copyright 2018 The MITRE Corporation. ALL RIGHTS RESERVED.}%
}
\else 
\fi

\title{Enrich-by-need Protocol Analysis for Diffie-Hellman\\(Extended version)}

\author{Moses D. Liskov\quad\and Joshua D. Guttman\quad\and John
  D. Ramsdell \quad\and Paul D. Rowe\quad\and F. Javier
  Thayer\thanks{Email:  \texttt{mliskov, guttman, ramsdell,
      prowe@mitre.org}}}

\ifllncs\institute{The MITRE Corporation}\fi

\begin{document}
\maketitle

\ifpubrel
\thispagestyle{title}
\fi
\pagestyle{plain}

\begin{abstract}
  Enrich-by-need protocol analysis is a style of symbolic protocol
  analysis that characterizes all executions of a protocol that extend
  a given scenario.  In effect, it computes a strongest security goal
  the protocol achieves in that scenario.  {\cpsa}, a Cryptographic
  Protocol Shapes Analyzer, implements enrich-by-need protocol
  analysis.

  In this paper, we describe how to analyze protocols using the
  Diffie-Hellman mechanism for key agreement ({\DHKA}) in the
  enrich-by-need style.  {\DHKA}, while widespread, has been
  challenging for protocol analysis because of its algebraic
  structure.  {\DHKA} essentially involves fields and cyclic groups,
  which do not fit the standard foundational framework of symbolic
  protocol analysis.  By contrast, we justify our analysis via an 
  algebraically natural model.

  This foundation makes the extended {\cpsa} implementation reliable.
  Moreover, it provides informative and efficient results.

  An appendix, written by John Ramsdell in 2014, explains how
  unification is efficiently done in our framework.  

\end{abstract}

\ifllncs{}\else%
\newpage
{\footnotesize
  \tableofcontents
}
\newpage
\fi
\section{Introduction}
\label{sec:intro}
\input{intro}

\section{An example:  Unified Model}
\label{sec:xmpls}
\input{example}

\input{eg_how}
\input{results}

\section{Algebraic and logical context
}
\label{sec:algebra}
\input{algebra}

\section{Strands, bundles, and protocols}
\label{sec:strands}
\input{strands}

\section{Protocol languages, skeletons and cohorts}
\label{sec:skeletons}
\input{skeletons}

\section{Tests and Cohorts}
\subsection{Solving tests}
\label{sec:tests}
\input{tests}

\subsection{Cohorts enrich skeletons}
\label{sec:cohorts}
\input{cohorts}

\bibliography{../../inputs/secureprotocols}
\bibliographystyle{plain}

\appendix
\section{Appendix:  Algebraic independence}
\label{sec:ind}
\input{ind}

\section{Efficient Unification}
\label{sec:unification}
\paragraph{Author of this appendix:  John D. Ramsdell.}  Date:
February, 2014.

\input{unification}

\end{document}

%% file: macros.tex
\usepackage{latexsym}
\usepackage{epic}
\usepackage{epsfig}
\usepackage{amsmath}
\usepackage{amssymb}
\usepackage{amscd}
\usepackage{url}
\usepackage{color}

\usepackage{graphics}
\usepackage{hyperref}

\input{xy}
\xyoption{all}

\newcounter{running}[section]
{\setcounter{running}{\value{enumi}}\end{enumerate}}

\newcommand{\DHKA}{{DH}}
\newcommand{\DHCR}{\ensuremath{\mathsf{DHCR}}}
\newcommand{\cpsa}{\textsc{cpsa}}
\newcommand{\kind}[1]{\ensuremath{\mathsf{#1}}}

\newcommand{\Algebra}{\kind{M}}
\newcommand{\Algebrabasic}{\Algebra_b}
\newcommand{\Algebraconstr}{\Algebra_c}

\newcommand{\Sigmabasic}{\Sigma_b}
\newcommand{\Sigmaconstr}{\Sigma_c}
\newcommand{\Sigmamult}{\Sigma_m}

\newcommand{\dom}{\kind{dom}}

\newcommand{\node}{\kind{nodes}}

\newcommand{\esc}{\kind{esc}}

\newcommand{\sent}{\kind{sent}}
\newcommand{\form}{\kind{form}}
\newcommand{\term}{\kind{term}}
\newcommand{\syntax}[1]{\textcolor{blue}{#1}}

\newcommand{\RPPred}[2]{\kind{RP}_{#1,#2}}
\newcommand{\PPred}[2]{\kind{P}_{#1,#2}}
\newcommand{\nf}[2]{\RPPred{#1}{#2}}
\newcommand{\pf}[2]{\PPred{#1}{#2}}

\newcommand{\fn}[1]{\ensuremath{\operatorname{\mathit{#1}}}}

\newcommand{\mt}[1]{\ensuremath{\mathtt{#1}}}

\newcommand{\field}{\mbox{\textsc{fld}}}
\newcommand{\zero}{\ensuremath{\mathbf{0}}}
\newcommand{\one}{\ensuremath{\mathbf{1}}}
\newcommand{\plus}{\ensuremath{{+}}}
\newcommand{\minus}{\ensuremath{{-}}}
\newcommand{\ftimes}{\ensuremath{{\cdot}}} 
\newcommand{\fdivide}{\ensuremath{{/}}} 

\newcommand{\group}{\mbox{\textsc{grp}}}
\newcommand{\gen}{\ensuremath{\mathbf{g}}}
\newcommand{\expt}{\ensuremath{\mathbf{exp}}}

\newcommand{\ltx}[1]{\kind{ltx}(#1)}
\newcommand{\sctrsc}{\mbox{\textsc{trsc}}}
\newcommand{\scname}{\mbox{\textsc{name}}}
\newcommand{\sctext}{\mbox{\textsc{text}}}
\newcommand{\scskey}{\mbox{\textsc{skey}}}
\newcommand{\scakey}{\mbox{\textsc{akey}}}
\newcommand{\sccreate}{\mbox{\textsc{create}}}
\newcommand{\scsbasic}{\mbox{\textsc{basic}}}
\newcommand{\scbasic}{\mbox{\textsc{basic}}}
\newcommand{\scmesg}{\mbox{\textsc{mesg}}}
\newcommand{\scnode}{\mbox{\textsc{node}}}
\newcommand{\scfld}{\field}
\newcommand{\scgrp}{\group}

\newcommand{\skel}{\ensuremath{\mathbb{A}}}
\newcommand{\skelB}{\ensuremath{\mathbb{B}}}
\newcommand{\skelC}{\ensuremath{\mathbb{C}}}

\ifamsart{}\else{

}\fi

\newcommand{\qdot}{\,\mathbf{.}\;}

\newcommand{\limp}{\Longrightarrow}
\newcommand{\entails}{\Vdash}

\newcommand{\bnd}{\mathcal{B}}

\newcommand{\lang}[1]{\mathcal{L}_{#1}}

\newcommand{\thefield}{\mathcal{F}}
\newcommand{\thegroup}{\mathcal{C}}
\newcommand{\mdlA}{\mathcal{A}}
\newcommand{\mdlB}{\mathcal{B}}

\newcommand{\seq}[1]{\langle #1 \rangle}

\newcommand{\strandnode}[2]{ {#1}\downarrow{#2} }

\newcommand{\ltk}{\ensuremath{\kind{ltk}}}
\newcommand{\kinv}{\ensuremath{\kind{inv}}}
\newcommand{\pk}{\ensuremath{\kind{pk}}}
\newcommand{\sk}{\ensuremath{\kind{sk}}}

\newcommand{\length}[1]{{| #1 |}}

\newcommand{\enc}[2]{\{\!|#1|\!\}_{#2}}
\newcommand{\tagged}[2]{[\![\,#1\,]\!]_{#2}}
\newcommand{\hash}[1]{\#
  (#1)}

\newcommand{\encs}[2]{\ensuremath{\{\!|#1|\!\}^{\mathrm{s}}_{#2}}}
\newcommand{\enca}[2]{\ensuremath{\{\!|#1|\!\}^{\mathrm{a}}_{#2}}}

\newcommand{\foundwithin}[3]{\ensuremath{#1 \mathrel{\odot^{#2}} #3}}

\newcommand{\restricted}{\mathbin{\mid\!\hspace{-1.2pt}\grave{}}}
\newcommand{\restrict}[2]{#1 \restricted #2}

\newcommand{\ingredient}{\sqsubseteq}

\newcommand{\msg}{\kind{msg\/}}

\newcommand{\direction}{\kind{dir\/}}

\newcommand{\assume}{\kind{assume\/}}

\newcommand{\strands}{\fn{Str}}

\newcommand{\tr}{\fn{tr}}

\newcommand{\s}{\kind{s}}

\newcommand{\skelax}{\kind{Ax}_{\kind{sk}}}

\newcommand{\trsc}{\ensuremath{\mathsf{trsc}}}
\newcommand{\ccolon}{\ensuremath{\mathbin{\mathbb{::}}}}
\newcommand{\cccolon}{\ensuremath{\mathbin{\mathbb{:::}}}}
\newcommand{\fv}{\kind{fv}}
\newcommand{\append}{\mathbin{{}^\smallfrown}}
\newcommand{\termat}{\mathbin{@}}

\ifllncs
{\bfseries}{\itshape}
\newtheorem{principle}{Principle}{\bfseries}{\itshape}
\newtheorem{mylemma}[prop]{Lemma}{\bfseries}{\itshape}
\newtheorem{thm}{Theorem}{\bfseries}{\itshape}
\newtheorem{eg}{Example}{\bfseries}{\upshape}
\newtheorem{cor}[prop]{Corollary}{\bfseries}{\itshape}
\spnewtheorem*{proofsketch}{Proof Sketch}{\itshape}{\rm}
\spnewtheorem*{thmredux}{Theorem}{\bfseries}{\itshape}
\spnewtheorem*{propredux}{Proposition}{\bfseries}{\itshape}

\else
\newtheorem{theorem}{Theorem}

\newtheorem{definition}{Definition}

\newtheorem{lemma}{Lemma}

\newtheorem{Assum}{Assumption}

\ifamsart{}\else
\def\squareforqed{\hbox{\rlap{$\sqcap$}$\sqcup$}}
\def\qed{\ifmmode\squareforqed\else{\unskip\nobreak\hfil
\penalty50\hskip1em\null\nobreak\hfil\squareforqed
\parfillskip=0pt\finalhyphendemerits=0\endgraf}\fi}

\newenvironment{proof}{\par\smallskip{\noindent\rm\bfseries
    Proof:}}{
\par\medskip\noindent}
\fi 
\fi



%% file: intro.tex

Diffie-Hellman key agreement ({\DHKA})~\cite{DiffieHellman76}, while
widely used, has been challenging for mechanized security protocol
analysis.  Some techniques,
e.g.~\cite{DBLP:conf/fosad/EscobarMM07,cremers2012operational,MeierSCB13,kuesters2009using},
have produced informative results, but focus only on proving or
disproving individual protocol security goals.  A protocol and a
specific protocol goal are given as inputs.  If the tool terminates,
it either proves that this goal is achieved, or else provides a
counterexample.  However, constructing the \emph{right} security goals
for a protocol requires a high level of expertise.

By contrast, the \emph{enrich-by-need} approach starts from a protocol
and some scenario of interest.  For instance, if the initiator has
had a local session, with a peer whose long-term secret is
uncompromised, what other local sessions have occurred?  What session
parameters must they agree on?  Must they have happened recently, or
could they be stale?

Enrich-by-need protocol analysis identifies all essentially different
smallest executions compatible with the scenario of interest.  While
there are infinitely many possible executions---since we put no bound
on the number of local sessions---often surprisingly few of them are
really different.  The Cryptographic Protocol Shapes Analyzer
(\cpsa)~\cite{cpsa09}, a symbolic protocol analysis tool based on
strand spaces~\cite{ThayerHerzogGuttman99,Guttman10,cpsatheory11},
efficiently enumerates these minimal, essentially different
executions.  We call the minimal, essentially different executions the
protocol's \emph{shapes} for the given scenario.

Knowing the shapes tells us a \emph{strongest} relevant security goal,
i.e.~a formula that expresses authentication and confidentiality trace
properties, and is at least as strong as any one the protocol achieves
in that scenario~\cite{Guttman14,Ramsdell12}.  Using these shapes, one
can resolve specific protocol goals.  The hypothesis of a proposed
security goal tells us what scenario to consider, after which we can
simply check whether the conclusion holds in each resulting shape
(see~\cite{RoweEtAl2016} for a precise logical treatment).

Moreover, enrich-by-need has key advantages.  Because it can also
compute strongest goal formulas directly for different protocols, it
allows comparing the strength of different
protocols~\cite{RoweEtAl2016}, for instance during standardization.
Moreover, the shapes provide the designer with \emph{visualizations}
of exactly what the protocol may do in the presence of an adversary.
Thus, they make protocol analysis more widely accessible, being
informative even for those whose expertise is not mechanized protocol
analysis.

In this paper, we show how we strengthened {\cpsa}'s enrich-by-need
analysis to handle {\DHKA}\@.  It is efficient, and has a flexible
adversary model with corruptions.

Foundational issues needed to be resolved.  For one thing, finding
solutions to equations in the natural underlying theories is
undecidable in general, which means that mechanized techniques must be
carefully circumscribed.  For another, these theories, which include
fields, are different from many others in security protocol analysis.
The field axioms are not (conditional) equations, meaning that they do
not have a simplest (or ``initial'') model for the analysis to work
within.  Much work on mechanized protocol analysis, even for {\DHKA},
which relies on equational theories and their initial models,
e.g.~\cite{DBLP:conf/fosad/EscobarMM07,cremers2012operational,MeierSCB13,kuesters2009using}.
However, one would like an analysis method to have an explicit theory
justifying it relative to the standard mathematical structures such as
fields.  {\cpsa} now has a transparent foundation in the algebraic
properties of the fields that {\DHKA} manipulates.  This development
extends our earlier work~\cite{DoughertyGuttman2014,LiskovThayer14}.

\paragraph{Contributions.}  In this paper, we describe how we extended
{\cpsa} to analyze {\DHKA} protocols.  Our method is currently
restricted to protocols that do not use addition in the exponents,
which is a large class.  We call these \emph{multiplicative}
protocols.  The method is also restricted to protocols that disclose
randomly chosen exponents one-by-one.  This allows modeling a wide
range of possible types of corruption; however, it excludes a few
protocols in which products of exponents are disclosed.  We say that a
protocol \emph{separates disclosures} if it satisfies our condition.

{\cpsa}'s analysis is based on two
principles---Lemmas~\ref{lemma:simple:visible}
and~\ref{lemma:group:visible}---that are valid for all executions of a
protocol $\Pi$, if $\Pi$ is multiplicative and separates disclosures.
They characterize how the adversary can obtain an exponent value such
as $xy$ or an exponentiated {\DHKA} value such as $g^{xy}$, resp.
They tell {\cpsa} what information to add to a partial analysis to get
one or more possible enrichments that describe the different
executions in which the adversary obtains these values.
Lemmas~\ref{lemma:cohort:group}--\ref{lemma:cohort:field} show how
{\cpsa} represents the enrichments.

To formulate and prove these lemmas, we had to clarify the algebraic
structures we work with.  In our semantics, the messages in protocol
executions contain mathematical objects such as elements of fields and
cyclic groups.  We regard the the random choices of the compliant
protocol participants serve as ``transcendentals,'' i.e.~primitive
elements added to a base field that have no algebraic relationship to
members of the base field.  The adversary cannot rely on any
particular algebraic relations with known values before the compliant
participant has made the random choice.  Even after the choice of a
value $x$, if the adversary sees only the exponentiated value $g^x$,
the adversary still does not know any algebraic (polynomial)
constraints on $x$.  Instead, as in the Generic Group
Model~\cite{Shoup97,Maurer05}, the adversary can apply algebraic
operations to such values, but cannot benefit from the bitstrings by
means of which they are represented.  

{\cpsa}'s analysis operates within a language of first order predicate
logic.  The analysis relates to real protocol executions via
Tarski-style satisfaction.  Protocol executions (using group and field
elements as messages, in the presence of an adversary) are the
\emph{models} of the \emph{theories} used in the analysis process.
This provides a familiar foundational setting.

{\cpsa} is very efficient.  When executed on a rich set of variants of
Internet Key Exchange (IKE) versions 1 and 2, our analysis required
less than 30
seconds on a laptop.  By contrast, Scyther analyzed this same set of
IKE variants, requiring about a day of work on a computing
cluster~\cite{Cremers11keyexchange}.  This is about 3.4 orders of
magnitude more time, quite apart from the difference between the quite
powerful cluster (circa 2010) and the laptop (built circa 2015).
Another benchmark~\cite{BasinCM13} with little use of {\DHKA} led to
similar results.  This suggests that {\cpsa} is broadly efficient,
with or without {\DHKA}.  Performance data for Tamarin and Maude-NPA
is less available.

%% file: example.tex

Many variants of the original Diffie-Hellman
idea~\cite{DiffieHellman76} exchange both certified, long-term values
$g^a$ and $g^b$, and also one-time, ephemeral values $g^x$ and $g^y$.
The peers agree compute session keys using \emph{key computation
  functions} $\mathsf{KCF}^{(\cdot)}$, where in successful sessions,
$\mathsf{KCF^A}(a, x, g^b, g^y) = \mathsf{KCF^B}(b, y, g^a, g^x)$.
The ephemeral values ensure the two parties agree on a different key
in each session.  The long-term values are intended for
authentication, namely that any party that obtains the session key
must be one of the intended peers.
%
%
Different functions $\mathsf{KCF}^{(\cdot)}$ yield different security
properties.  Protocol analysis tools for {\DHKA} key exchanges must be
able to use algebraic properties to identify these security
consequences.

\subsection{DH Challenge-Response with Unified Model keying}
\label{sec:xmpls:dhcr:um}

We consider here a simple {\DHKA} Challenge-Response protocol
${\DHCR}$ in which nonces from each party form a challenge and
response protected by a shared, derived {\DHKA} key (see
Fig.~\ref{fig:dhcr}).
\begin{figure}[tb]\small
  \centering
  \begin{eqnarray*}
    \xymatrix@R=2mm{\bullet\ar@{=>}[r] & \bullet\ar@{=>}[r]\ar[d] & \bullet\ar@{=>}[r] & \bullet\ar[d]
    \\   
    \mathtt{certs}\ar[u] & na, A, B, g^x  & g^\beta, \enc{na,nb}K\ar[u] &{nb}
                                                                                          } \\ 
    \xymatrix@R=2mm{    \mathtt{certs}\ar[d] & na, A, B, g^\alpha\ar[d]  & g^y, \enc{na,nb}K & {nb}\ar[d]
    \\   \bullet\ar@{=>}[r] & \bullet\ar@{=>}[r] & \bullet\ar@{=>}[r]\ar[u] & \bullet
                                                                              }
    \\
    \xymatrix@R=3mm{  
    \bullet\ar[d]\ar@{=>}[r] &\cdots  &&& \\
    \tagged{g^{\ltx P}}{\sk(P)}
                                           } 
  \end{eqnarray*}
  \begin{tabbing}
    This is  
    \= where \= \texttt{init} computes $K$ as\quad \= $\mathsf{KCF^B}({b}, y,g^{{a}},g^\alpha)$
    and \kill 
    \> where \> \texttt{init} computes $K$ as \> $\mathsf{KCF^A}({a}, x, g^{b}, g^\beta)$
    \\
    \>\> \texttt{resp} computes $K$ as \> $\mathsf{KCF^B}({b}, y,g^{{a}},g^\alpha)$ 
  \end{tabbing}
  \caption{Protocol ${\DHCR}$:  Initiator, responder, and
    $\ltx{}$ self-certifying roles.  The $\mathtt{certs}$ are
    $\tagged{g^{\ltx A}}{\sk(P)}$ and $\tagged{g^{\ltx B}}{\sk(P)}$
    for each role.}
  \label{fig:dhcr}
\end{figure}

Each participant obtains two self-signed certificates, covering the
long term public values of his peer and himself.  $A$ and $B$'s
long-term {\DHKA} exponents are values $\ltx A, \ltx B$, which we will
mainly write as a lower case $a,b$.  We assume that the principals are
in bijection with distinct long-term $\ltx{\cdot}$ values.

Each participant furnishes an ephemeral public {\DHKA} value $g^x,g^y$
in cleartext, and also a nonce.  Each participant receives a value
which may be the peer's ephemeral, or may instead be some value
selected by an active adversary.  Neither participant can determine
the exponent for the ephemeral value he receives, but since the value
is a group element, it must be some value of the form
$g^\alpha,g^\beta$.

Each participant computes a session key using his
$\mathsf{KCF}^{(\cdot)}$.  The responder uses the key to encrypt the
nonce received together with his own nonce.  The initiator uses the
key to decrypt this package, and to retransmit the responder's nonce
in plaintext as a confirmation.

The \emph{registration} role allows any principal $P$ to emit its
long-term public group value $g^{\ltx P}$ under its own digital
signature.  In full-scale protocols, a certifying authority's
signature would be used, but in this example omit the CA so as not to
distract from the core {\DHKA} issues.

We will assume that each instance of a role chooses its values for
certain parameters freshly.  For instance, each instance of the
registration role makes a fresh choice of ${\ltx P}$.  Each instance
of the initiator role chooses $x$ and $na$ freshly, and each responder
instance chooses $y$ and $nb$ freshly.

This protocol is parameterized by the key computations.  We consider
three key computations from the Unified
Model~\cite{ankney1995unified}, with a hash function~$\hash\cdot$
standing for key derivation.  The shared keys---when each participant
receives the ephemeral value $g^\alpha=g^x$ or $g^\beta=g^y$ that the
peer sent---are:

\begin{tabbing}\textbf{Three-component UM3:} \quad \= \kill 
  \textbf{Plain UM:\quad} \>  $\hash{g^{{a}{b}}, g^{xy}}$; \\ 
  \textbf{Criss-cross UMX:\quad} \> $\hash{g^{{a}y},
    g^{{b}x}}$; \\ 
  \textbf{Three-component UM3:\quad} \> $\hash{g^{{a}y}, g^{{b}x}, g^{xy}}$
\end{tabbing}
%
We discuss three security properties:
\begin{description}
  
  \item[Authentication:]  If either the initiator or responder
  completes the protocol, and \emph{both} principals' private
  long-term exponents are secret, then the intended peer must have
  participated in a matching conversation.
\end{description}
All three key computations UM, UMX, and UM3 enforce the authentication
goal.
\begin{description}
  \item[Impersonation resistance:]  If either the initiator or
  responder completes the protocol, and the intended peer's private
  long-term exponent is secret, then the intended peer must have
  participated in a matching conversation.
\end{description}
Here we do not assume that \emph{ones own} private long-term exponent
is secret.  Can the adversary impersonate the intended peer if ones
own key is compromised?
  
${\DHCR}$ with the plain UM KCF is susceptible to an
impersonation attack:  An attacker who knows Alice's own long-term
exponent can impersonate any partner to Alice.  The adversary can
calculate $g^{{a}{b}}$ from ${a}$ and $g^{{b}}$, and can calculate the
$g^{xy}$ value from $g^x$ and a $y$ it chooses itself.

On the other hand, the UMX and UM3 KCFs resist the impersonation
attack.
\begin{description}
  \item[Forward secrecy:]  If the intended peers complete a protocol
  session and then the private, long-term exponent of each party is
  exposed \emph{subsequently}, then the adversary still cannot derive
  the session key.

  This is sometimes called \emph{weak} forward secrecy.
\end{description}
To express forward secrecy, we allow the registration role to continue
and subsequently disclose the long term secret $\ltx P$ as in
Fig.~\ref{fig:dhcr:registration}.  If we assume in an analysis that
$\ltx P$ is uncompromised, that implies that this role does not
complete.
\begin{figure}\small
  \centering
  \[  \xymatrix@R=3mm{  
      \bullet\ar[d]\ar@{=>}[r] & \bullet\ar@{=>}[r] & \bullet\ar[d] \\
      \tagged{g^{\ltx P}}{\sk(P)} & \mathtt{dummy}\ar[u]& {\ltx P}}
  \]
  \caption{The registration role in full}
  \label{fig:dhcr:registration}
\end{figure}
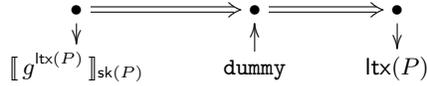
The dummy second node allows specifying that the $\ltx P$ release node
occurs after some other event, generally the completion of a normal
session.

The Normal UM KCF guarantees forward secrecy, but UMX does not.  In
UMX, if an adversary records both $g^x$ and $g^y$ during the protocol
and learns ${a}$ and ${b}$ later, the adversary can compute the key by
exponentiating $g^x$ to the power $b$ and $g^y$ to the power $a$.  

UM3 restores forward secrecy.  It meets all three goals.

\subsection{Strand terminology}
\label{sec:xmpls:terms}

In discussing how {\cpsa} carries out its analysis, we will use a
small amount of the strand space terminology; see
Section~\ref{sec:strands} for more detail.

We call a sequence of transmission and reception events as illustrated
in Figs.~\ref{fig:dhcr}--\ref{fig:dhcr:registration} a \emph{strand}.
We call each send-event or receive-event on a strand a \emph{node}.
We draw strands either horizontally as in
Figs.~\ref{fig:dhcr}--\ref{fig:dhcr:registration} or vertically, as in
diagrams generated by {\cpsa} itself.

A \emph{protocol} consists of a finite set of these strands, which we
call the \emph{roles} of the protocol.  The roles contain variables,
called the \emph{parameters} of the roles, and by plugging in values
for the parameters, we obtain a set of strands called the
\emph{instances} of the roles.  We also call a strand a \emph{regular
  strand} when it is an instance of a role, because it then complies
with the rules.  \emph{Regular nodes} are the nodes that lie on
regular strands.  We speak of a \emph{regular principal} associated
with a secret if that secret is used only in accordance with the
protocol, i.e.~only in regular strands.

In an execution, events are (at least) partially ordered, and values
sent on earlier transmissions are available to the adversary, who
would like to provide the messages expected by the regular
participants on later transmission nodes.  The adversary can also
generate primitive values on his own.  We will make assumptions
restricting which values the adversary does generate to express
various scenarios and security goals.
%


%% file: eg_how.tex
\subsection{How {CPSA} works, I:  UMX initiator}
\label{sec:xmpls:how}

Here we will illustrate the main steps that {\cpsa} takes when
analyzing {\DHCR}.  For this illustration, we will focus on the
impersonation resistance of the UMX key computation, in the case where
the initiator role runs, aiming to ensure that the responder has also
taken at least the first three steps of a matching run.  The last step
of the responder is a reception, so the initiator can never infer that
it has occurred.  We choose this case because it is typical, yet quite
compact.

\paragraph{Starting point.}  We start {\cpsa} on the problem shown in
Fig.~\ref{fig:dhcr:umx:init:1},
\begin{figure}[tb]
  \centering
  \begin{minipage}[c]{.3\linewidth}
    \centering
    \includegraphics[width=\linewidth]{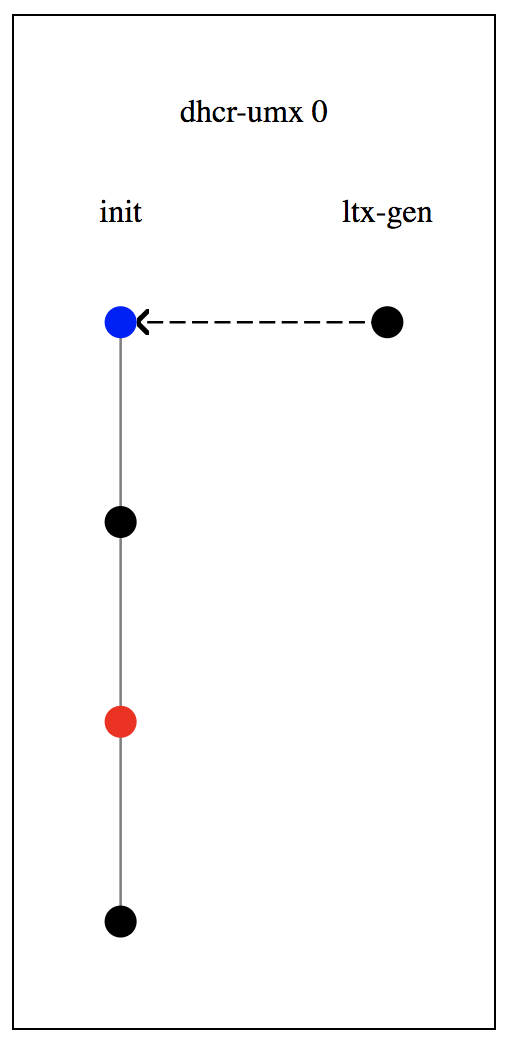} \\
  \end{minipage}\qquad
  \begin{minipage}[c]{.37\linewidth}
    \centering
    \includegraphics[width=\linewidth]{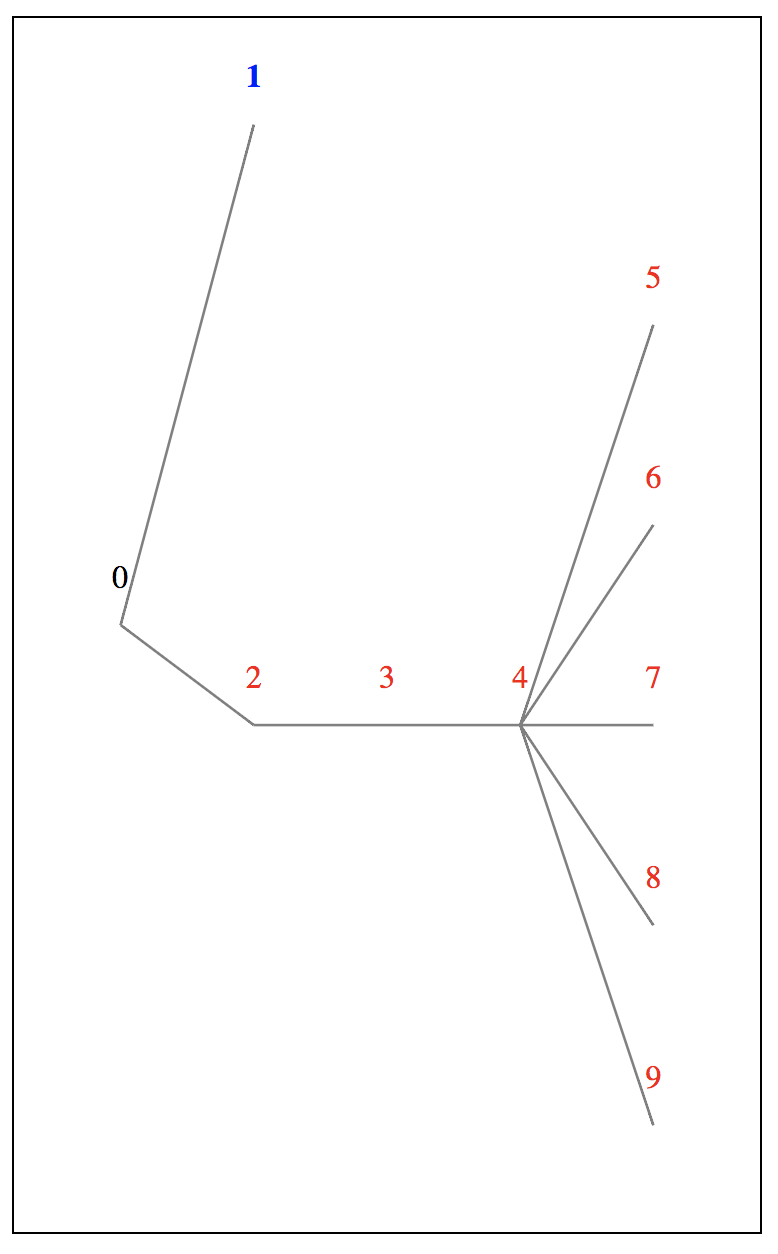}
  \end{minipage}
  \caption{Skeleton 0:  The initial scenario, assuming {$\ltx B$
      non-compromised and $A\not=B$}.  \quad {\cpsa{} exploration tree}}
  \label{fig:dhcr:umx:init:1}
  \label{fig:dhcr:um:tree}
\end{figure}
%
%
in which $A$, playing the initiator role, has made a full local run of
the protocol, and received the long term public value of $B$ from a
genuine run of the ltx-gen role.  These are shown as the vertical
column on the left and the single transmission node at the top to its
right.  We will assume that $B$'s private value $\ltx B$ is
non-compromised and freshly generated, so that the public value
$g^{\ltx B}$ originates only at this point.  In particular, this run
definitely does not progress to expose the secret as in the third node
of Fig.~\ref{fig:dhcr:registration}.  The fresh selection of $\ltx B$
must certainly precede the reception of $g^{\ltx B}$ at the beginning
of the initiator's run.  This is the meaning of the dashed arrow
between them.  We \emph{do not} assume that $A$'s long term secret
${\ltx A}$ is uncompromised, although we will assume that the
ephemeral $x$ value is freshly generated by the initiator run, and not
available to the adversary.
We assume $A\not=B$, which is the case of most interest.

\paragraph{Exploration tree.}  In Fig.~\ref{fig:dhcr:um:tree}, we also
show the exploration tree that {\cpsa} generates.  Each item in the
tree---we will call each item a \emph{skeleton}---is a scenario
describing some behavior of the regular protocol participants, as well
as some assumptions.  For instance, skeleton~0 contains the
assumptions about ${\ltx B}$ and $A$'s ephemeral value $x$ mentioned
before.  The exploration tree contains one blue, bold face entry,
skeleton~1 (shown in Fig.~\ref{fig:umx:init:skel:1}), as well as a
subtree starting from 2 that is all red.  The bold blue skeleton~1 is
a \emph{shape}, meaning it describes a simplest possible execution
that satisfies the starting skeleton~0.  The red skeletons are
\emph{dead skeletons}, meaning possibilities that the search has
excluded; no executions can occur that satisfy these skeletons.  Thus,
skeleton~1 is the only shape, and {\cpsa} has concluded that all
executions that satisfy skeleton~0 in fact also satisfy skeleton~1.

In other examples, there may be several shapes identified by the
analysis, or in fact zero shapes.  The latter means that the initial
scenario cannot occur in any execution.  This may be the desired
outcome, for instance when the initial scenario exhibits some
disclosure that the protocol designer would like to ensure is
prevented.

\begin{figure}
  \centering
  \begin{minipage}[c]{.36\linewidth}
    \includegraphics[width=\linewidth]{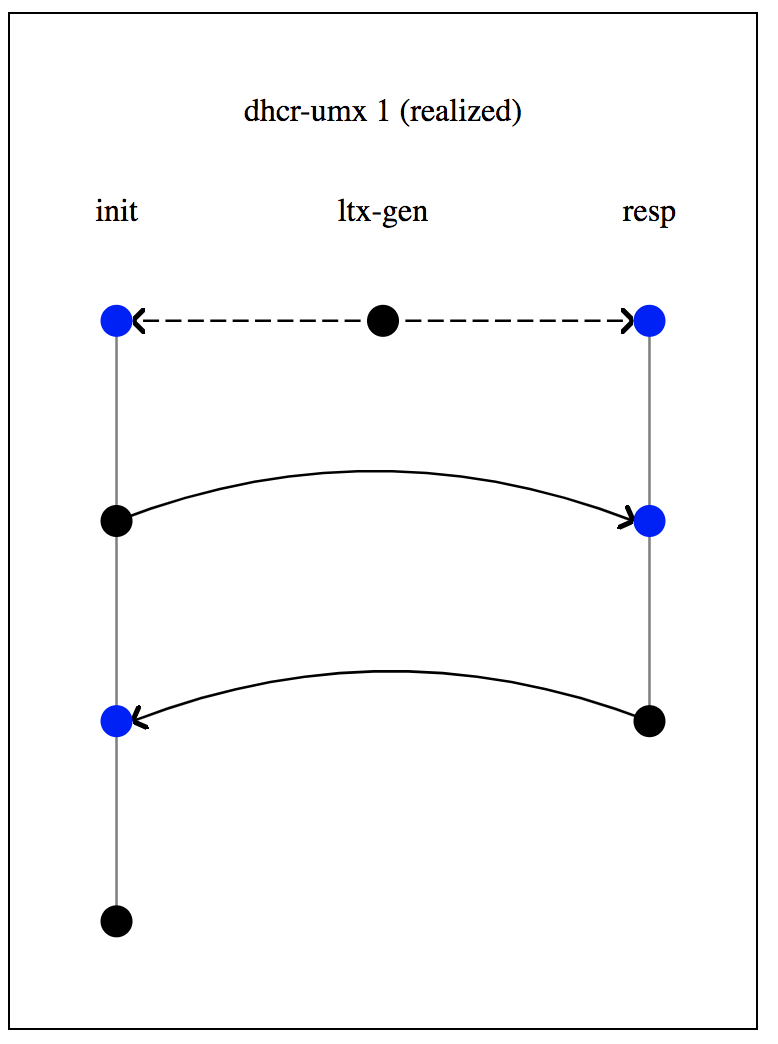}
  \end{minipage}\qquad\quad
  \begin{minipage}[c]{.36\linewidth}
    \centering
    \includegraphics[width=\linewidth]{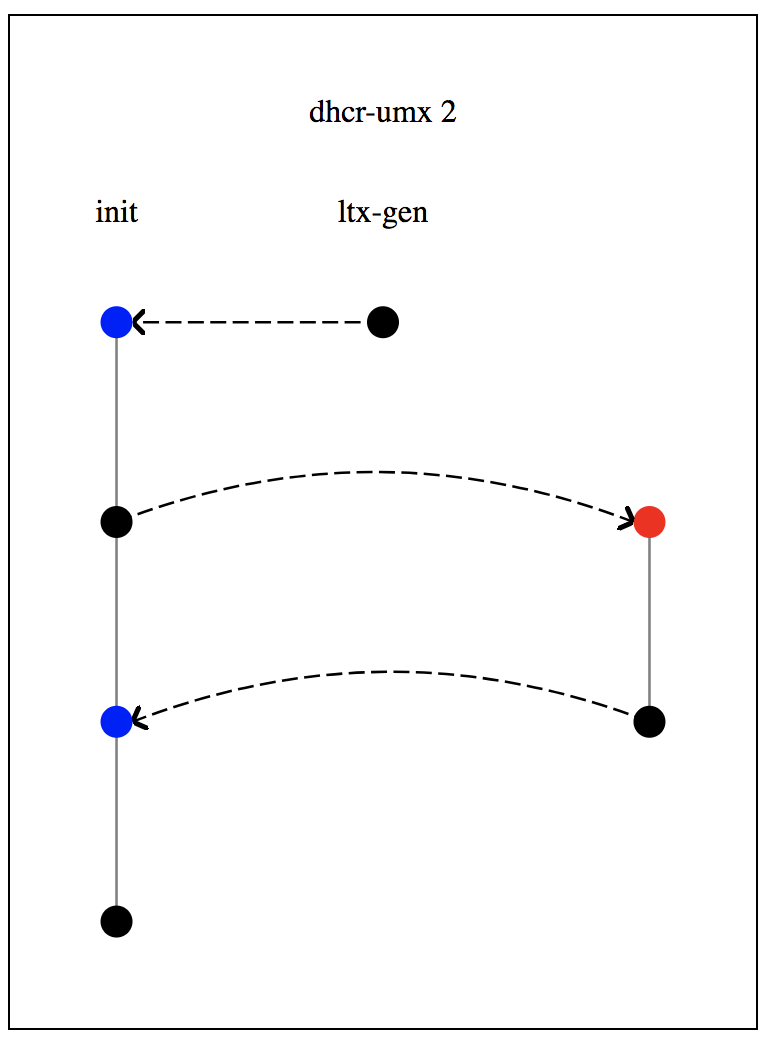}
  \end{minipage}
  \caption{Skeleton~1, the sole resulting shape. \quad {Skeleton~2:  Can the
      UMX encryption key $K$ be exposed, on the rightmost strand?}}
  \label{fig:umx:init:skel:1}
  \label{fig:umx:init:skel:2}
\end{figure}

\paragraph{First step.}  {\cpsa}, starting with skeleton~0 in
Fig.~\ref{fig:dhcr:umx:init:1}, identifies the third node of the
initiator strand, which is shown in red, as unexplained.  This is the
initiator receiving the {\DHKA} ephemeral public value $g^y$ and the
encryption $\enc{na,nb}K$, where $K$ is the session key $A$ computes
using $g^y$ and the other parameters.  The node is red because the
adversary cannot supply this message on his own, given the materials
we already know that the regular, compliant principals have
transmitted.  Thus, {\cpsa} is looking for additional information,
including other transmissions of regular participants, that could
explain it.  Two possibilities are relevant here, and they lead to
skeletons~1 and~2 (see
Fig.~\ref{fig:umx:init:skel:1}
).

In skeleton~1, a regular protocol participant executing the responder
role transmits the message $g^y,\enc{na,nb}K$.  Given the values in
this message---including those used to compute $K$ using the UMX
function---all of the parameters in the responder role are determined.
It is executed by the intended peer $B$, who is preparing the message
for $A$, with matching values for the nonces and ephemeral exponents.
These are matching conversations~\cite{BellareRogaway93,Lowe97}.  A
solid arrow from one node to another means that the former precedes
the latter, and moreover it transmits the same message that the latter
receives.

Skeleton~2 considers whether the key $K=\hash{g^{ay}, g^{bx}}$,
computed by the initiator, might be compromised.  $K$ is the value
received on the rightmost strand.  The reception node is called a
\emph{listener node}, because it witnesses for the availability of $K$
to the adversary.  The ``heard'' value $K$ is then retransmitted so
that {\cpsa} can register that this event must occur before the
initiator's third node.  This listener node is red because {\cpsa}
cannot yet explain how $K$ would become available.  However, if
additional information, such as more actions of the regular
participants, would explain it, then the adversary could use $K$ to
encrypt $na,nb$ and forge the value $A$ receives.  Thus, skeleton~2
identifies this listener node for further exploration.

\paragraph{Step 2.}  Proceeding from skeleton~2, {\cpsa} performs a
simplification on $K=\hash{g^{ay}, g^{bx}}$.  The value $y$ is
available to the adversary, as is $a$, since we have not assumed them
uncompromised.  Thus, $g^{ay}$ is available.  The adversary will be
able to compute $K$ if he can obtain $g^{bx}$.  Skeleton~3 (not shown)
is similar to skeleton~2 but has a red node asking {\cpsa} to explain
how to obtain $g^{bx}$.

This requires a step which is distinctive to {\DHKA} protocols.
{\cpsa} adds Skeleton~4, which has a new rightmost strand with a red
node, receiving the pair $g^{bx/w}, w$.  To resolve this, {\cpsa} must
meet two constraints.  First, it must choose an exponent $w$ that can
be exposed and available to the adversary.  Second, for this value of
$w$, either $w=bx$ or else the ``leftover'' {\DHKA} value $g^{bx/w}$
must be transmitted by a regular participant and extracted by the
adversary.
\begin{figure}
  \centering
  \includegraphics[height=.18\textheight]{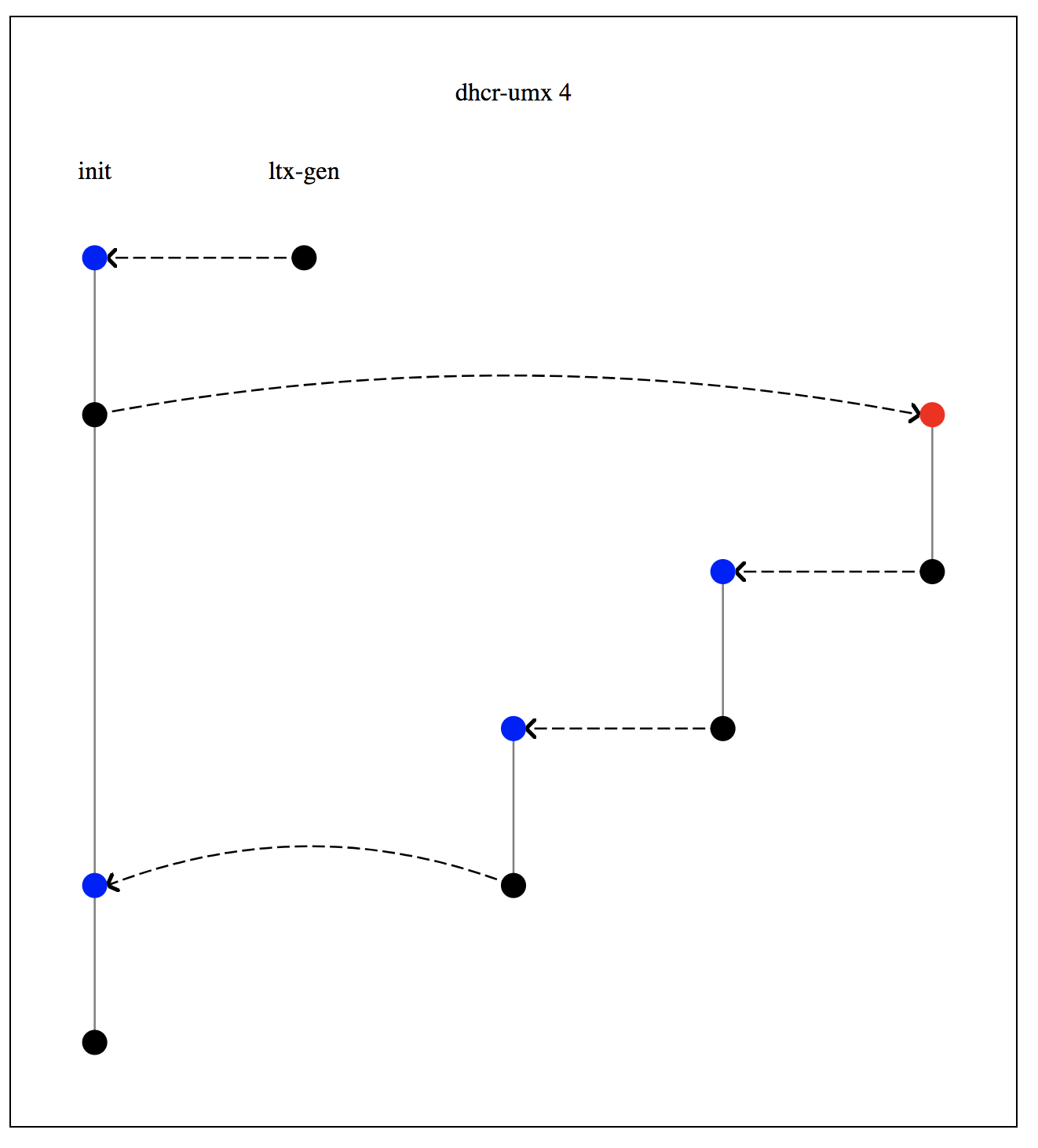}
  \caption{Skeleton~4:  $g^{bx/w}, w$ is on the rightmost strand.  Is
    there an exposed exponent $w$ where either $w=bx$ or else
    $g^{bx/w}$ was sent by a regular participant?}
  \label{fig:umx:init:skel:4}
\end{figure}
\label{step:2:group}

One of our key lemmas, Lemma~\ref{lemma:group:visible}, justifies this
step.

\paragraph{Step 3, clean-up.}  From skeleton~4, {\cpsa} considers the
remaining possibilities in this branch of its analysis.  First, it
immediately eliminates the possibility $w=bx$, since the protocol
offers no way for the adversary to obtain $b$ and $x$.
\label{step:3}

In fact, because $b$ and $x$ are random values, independently chosen
by different principals, the adversary cannot obtain their product
without obtaining the values themselves.

{\cpsa} then considers each protocol role in turn, namely the
initiator, responder, and registration roles.  Can any role transmit a
{\DHKA} value of the form $g^{bx/w}$, where the resulting inferred
value for $w$ would be available to the adversary?

In skeleton~5, it considers the case in which the initiator strand is
the original starting strand, which transmits $g^x$.  Thus, $x=bx/w$,
which is to say $w=b$.  However, since $b$ is assumed uncompromised,
the adversary cannot obtain it, and this branch is dead.  Skeleton~6
explores the case where a different initiator strand sends $g^z$, so
$z=bx/w$, i.e.~$w=bx/z$.  However, this is unobtainable, since it too
is compounded from independent, uncompromised values $b,x,z$.

Skeleton~7 considers the responder case, and skeletons~8 and~9
consider a registration strand which is either identical with the
initial one (skeleton~8) or not (skeleton~9).  They are eliminated for
corresponding reasons.

The entire analysis takes about 0.2 second.

\paragraph{{CPSA} overall algorithm.}  In this paper, unlike earlier
work, a \emph{skeleton} for $\Pi$ will be a theory that a set of
executions satisfies.
This may be the empty set of executions, in which case the skeleton is
``dead,'' like skeletons~2--9.

Protocol analysis in {\cpsa} starts with a skeleton, the initial
scenario.  At any step, {\cpsa} has a set $\mathcal{S}$ of skeletons
available.  If $\mathcal{S}$ is empty, the run is complete.

Otherwise, {\cpsa} selects a skeleton $\skel$ from $\mathcal{S}$.  If
$\skel$ is \emph{realized}, meaning that it gives a full description
of some execution, then {\cpsa} records it as a result.  Otherwise,
there is some reception node $n$ within $\skel$ that is not explained.
This $n$ is the \emph{target node}.  That means that {\cpsa} cannot
show how the message received by $n$ could be available, given the
actions the adversary can perform on his own, or using messages
received from earlier transmissions.

{\cpsa} replaces $\skel$ with a ``cohort.''  This is a set
$\skelC_1,\ldots,\skelC_k$ of extensions of the theory $\skel$.  For
every execution satisfying $\skel$, there should be at least one of
the $\skelC_i$ which this execution satisfies.  {\cpsa} must not
``lose'' executions.  When $k=0$ and there are no cohort members,
{\cpsa} has recognized that $\skel$ is dead.  {\cpsa} then repeats
this process starting with
$\mathcal{S}\setminus\{\skel\}\cup\{\skelC_1,\ldots,\skelC_k\}$.

\paragraph{{CPSA} cohort selection.}  {\cpsa} generates its cohorts by
adding one or more facts, or new equalities, to $\skel$, to generate
each $\skelC_i$.  It also does some renaming, so that the theories
$\skelC_i$ result from a theory interpretation from $\skel$ rather
than syntactic extension.

The selection of facts to add is based on a taxonomy of the executions
satisfying $\skel$.  In each one of them, the reception on the target
node $n$ must somehow be explained.  There are only a limited number
of types of explanation, which are summarized in
Fig.~\ref{fig:cohort}.  
\begin{figure}
  \centering
  \hrule~\\[2mm]   %
  \begin{description}\small
    \item[Regular transmission:]  Some principal, acting in accordance
    with the protocol (i.e.~``regularly''), has transmitted a message
    in this execution which is not described in $\skel$.

    Skeleton~1 was introduced in this way in step~1 above.  
    \item[Encryption key available:]  An encrypted value in a
    reception must be explained, and the adversary obtains the
    encryption key in a way that {\cpsa} will subsequently explore.

    Skeleton~2 was introduced in this way in step~1 above.  
    \item[Decryption key available:]  A value was previously
    transmitted in encrypted form, and, in a way that {\cpsa} will
    subsequently explore, the adversary obtains the decryption key to
    extract it.
    \item[Specialization:]  The execution satisfies additional
    equations, not included in $\skel$, and in this special case the
    adversary can obtain the target node message.
    \item[{\DHKA} value computed:] The adversary needs to supply a
    {\DHKA} value $g^\alpha$, and obtains it from $g^{\alpha/w}$ by
    exponentiating with $w$, which must also be available.

    Skeleton~4 was introduced in this way in step~2 above.  
    \item[Exponent value computed:]  The adversary must obtain an
    exponent $xw$.  There are then two subcases:
    \begin{enumerate}
      \item Both $x$ and $w$ will be obtainable in ways that {\cpsa}
      will subsequently explore; or
      \item $w$ will be instantiated as some $v/x$, so that $x$ will
      cancel out.  Thus, $x$ is in fact be absent from the instance of
      $xw$.
    \end{enumerate}
  \end{description}
  \hrule
  
  \caption{Kinds of cohort members.}
  \label{fig:cohort}
\end{figure}
Of these types, the first three are entirely unchanged.  The
\emph{Specialization} clause is conceptually unchanged, but the
unification algorithm that finds the relevant equations has been
updated to reflect the {\DHKA} algebra; it works efficiently in
practice (see Section~\ref{sec:strands:restrictions} below).

The last two clauses are new.  Lemmas~\ref{lemma:group:visible}
and~\ref{lemma:simple:visible} justify them (resp.), and Cohort
Cases~\ref{lemma:cohort:group} and~\ref{lemma:cohort:field} formulate
how the skeletons $\skelC_i$ are defined from $\skel$; in these cases
they are in fact syntactic extensions.  We will push the overall
Theorem~\ref{thm:test} into the Appendix (p.~\pageref{thm:test}) to
concentrate on these central cases.

We now briefly mention the types of step we have not yet illustrated.

%
%
%
%

\paragraph{A decryption step.}  One must also handle a case dual to
the encryption principle we used in step 1 above, leading from
skeleton~0 to skeletons~1,2, but concerning decryption.

Consider the analysis from the \emph{responder's} point of view.  In
this scenario, the responder's last step, in which he receives the
decrypted nonce $nb$, needs explanation.  The responder previously
transmitted it inside the encryption in $g^y,\enc{na,nb}K$.  In this
case, {\cpsa} must explain the how $nb$ can \emph{escape} from the
protection of the encryption.

One possible explanation is that a regular transmission does so.  That
is, in some role of the protocol, a participant accepts messages of
the form $\enc{na,nb}K$, and retransmits the second nonce outside this
form.  This is symmetric to the case with an initiator strand, leading
to a shape very similar to skeleton~1.

The other possible explanation is that the adversary obtains the key
$K=\hash{g^{ay}, g^{bx}}$.  This would allow the adversary to do the
decryption, and free $nb$ from its protection.  The analysis of the
resulting skeleton~12 is almost identical with the analysis of
skeleton~2 in Section~\ref{sec:xmpls:how}.

\paragraph{A step for exponents.}  When {\cpsa} needs to explore the
availability of some exponent value, such as $bx/z$ in an example
above, it may be able to resolve the question directly.  When it
cannot, it takes a step based on our other key lemma,
Lemma~\ref{lemma:simple:visible}.  This involves splitting the
exploration tree, i.e.~distinguishing possible cases.  In one branch,
variables will be instantiated so that an element such as $b$ in the
exponent will be canceled out.  The other branch explores whether the
same element $b$ can be available to the adversary.  {\cpsa} resorts
to this step relatively rarely when a protocol uses {\DHKA} in
straightforward ways; the forward secrecy analysis for the plain UM
key computation gives an example however.


%% file: results.tex

\subsection{Performance}
\label{sec:xmpls:perf}

Our implementation of the {\cpsa} tool is highly efficient.  See
Figure~\ref{table:performance} for a list of performance results.  We
ran the tool not only on the Diffie-Hellman challenge-response
protocol described in this section, with each of the key derivation
options, but also on a rich set of variants of Internet Key Exchange
(IKE) versions 1 and 2.  Cremers' Scyther tool was used circa 2010 to
analyze this same set of variants \cite{cremers2012operational}.  The
IKE variants were analyzed for an average of five properties each,
yielding conclusions similar to those drawn using
Scyther~\cite{cremers2012operational}.  Scyther can thus be used as a
basis for performance comparison.  The authors reported that the
analysis of the IKE variants took about a day of work on a computing
cluster.
\begin{figure}[tb]\centering\small
  
  \begin{tabular}{c|c|c} 
    \multicolumn{3}{c}{ DHCR: Example \& Time} \\ \hline
    dhcr-um {\quad} 4.06s & dhcr-umx {\quad} 0.72s & dhcr-um3 {\quad}
                                                     0.47s  
  \end{tabular}
  ~ \\[4ex]
  ~
  \begin{tabular}{l@{\quad}r|l@{\quad}r}
    \multicolumn{4}{c}{ IKEv1:  Example \& Time} \\ \hline
    IKEv1-pk2-a & 1.06s & IKEv1-pk2-a2 & 1.02s \\
    IKEv1-pk2-m & 0.49s & IKEv1-pk2-m2 & 0.58s \\
    IKEv1-pk-a1 & 1.27s & IKEv1-pk-a12 & 1.09s \\
    IKEv1-pk-a2 & 1.00s & IKEv1-pk-a22 & 1.10s \\
    IKEv1-pk-m  & 0.51s & IKEv1-pk-m2  & 0.49s \\
    IKEv1-psk-a & 0.43s & IKEv1-psk-m  & 0.68s \\
    IKEv1-psk-m-perlman & 0.69s & IKEv1-quick & 0.66s \\
    IKEv1-psk-quick-noid & 0.65s & IKEv1-quick-nopfs & 0.09s \\
    IKEv1-sig-a1 & 0.15s & IKEv1-sig-a2 & 0.16s \\
    IKEv1-sig-a-perlman & 0.17s & IKEv1-sig-a-perlman2 & 0.19s \\
    IKEv1-sig-m & 0.21s & IKEv1-sig-m-perlman & 0.19s \\
    \multicolumn{4}{c}{ ~ } \\ 
    \multicolumn{4}{c}{ IKEv2:  Example \& Time } \\
    \hline
    IKEv2-eap & 1.35s & IKEv2-eap2 & 1.36s \\
    IKEv2-mac & 0.76s & IKEv2-mac2 & 0.97s \\
    IKEv2-mac-to-sig{\quad} & 0.83s & IKEv2-mac-to-sig2{\quad} & 0.82s \\
    IKEv2-sig & 0.56s & IKEv2-sig2 & 0.54s \\
    IKEv2-sig-to-mac & 0.70s & IKEv2-sig-to-mac2 & 0.69s
  \end{tabular}
  \caption{{\cpsa} Timings: DHCR-UM*, Internet Key Exchange v.~1 and
    2.}
  \label{table:performance}
\end{figure}

In contrast, {\cpsa} needed no more than 1.36 seconds to analyze any
of the individual variants, and 21.32 seconds to analyze all of them
combined.  The data in Fig.~\ref{table:performance} is from a run of
{\cpsa} on a mid-2015 MacBook Pro with a 4-core 2.2 GHz Intel Core i7
processor, run with up to 8 parallel threads using the Haskell
run-time system.  
    
Our analysis of the Diffie-Hellman challenge response protocol
revealed the properties we expected, as described earlier in this
section.
Each {\cpsa} run checked about five scenarios, considering the
guarantees obtained by initiator and responder each under two sets of
assumptions, as well as a forward secrecy property.  
Our analysis of the IKE variants discovered no novel attacks, but does
sharpen Cremers' analysis, because {\cpsa} reflects the algebraic
properties of Diffie-Hellman natively, while Scyther emulated some
properties of Diffie-Hellman.

We turn now to the task of identifying the foundational ideas that
will justify protocol analysis in the efficient style we have just
illustrated.


%% file: algebra.tex

We first examine the mathematical objects on which Diffie-Hellman
relies (Section~\ref{sec:algebra:math}), and we then consider how to
represent tupling and cryptographic operations above them
(Section~\ref{sec:algebra:free}).

\subsection{The algebraic context}
\label{sec:algebra:math}

Diffie-Hellman operations act in a cyclic subgroup ${\thegroup}$ of a
enclosing group $\mathcal{G}$.  In the original formulation, the
enclosing group $\mathcal{G}$ is the multiplicative group
$\mathbb{Z}^*_p$ modulo a large prime $p$; since 0 does not
participate in the multiplicative group, this has an even number of
members, namely $p-1$.  If $q$ is a prime that divides $p-1$ and $g$
is a typical member of $\mathbb{Z}^*_p$, then $\mathbb{Z}^*_p$
contains a cyclic subgroup ${\thegroup}_q$ of order $q$ generated by
the powers of $g$ to integers mod $q$.  Since $q$ is prime, the
integers mod $q$ form a field ${\thefield}_q$, where the field
operations are addition modulo $q$ and multiplication mod $q$.

Thus, generally, consider an enclosing group $\mathcal{G}$ and a large
prime $q$.  We are interested in the field ${\thefield}_q$ and a
cyclic group ${\thegroup}$ generated by the powers $g^x$ of some
$g\in\mathcal{G}$ exponentiated to field elements $x\in{\thefield}_q$.
We will assume that the $\mathcal{G}$s are chosen (or represented) so
that, given $g^x\in{\thegroup}$, it is algorithmically hard to recover
$x$, and, more specifically, the following problems are hard:
\begin{description}
  \item[Computational DH assumption:] given $g^x$ and $g^y$, generated
  from randomly chosen $x,y\in{\thefield}_q$, to find $g^{xy}$; and
  \item[Decisional DH assumption:] given $g^x$ and $g^y$, generated
  with randomly chosen $x,y\in{\thefield}_q$, to distinguish $g^{xy}$
  from $g^z$, where $z\in{\thefield}_q$ is independently randomly
  chosen.
\end{description}
We will work in a model in which the adversary can apply the group
operation to known group elements; can apply the field operations to
known field elements; and can exponentiate a known group value to a
field value.  As in the generic group model~\cite{Shoup97,Maurer05},
we regard the structures as otherwise opaque to the adversary.  In
\cite{BartheFFMSS14}, Barthe et al.~show that the generic group model,
which is expressed in probabilistic terms, justifies a
non-probabilistic adversary model in which the adversary must solve
equations using only the given algebraic operations.  We will follow
their strategy.

\paragraph{Adversary model.}  The job of the adversary is to construct
counterexamples to security goals of the system.  In the framework we
will use~\cite{Guttman14,RoweEtAl2015}, a security goal is an
implication $\Phi\limp\Psi$, so the adversary, to provide a
counterexample, will offer a structure in which $\Phi$ is satisfied,
but $\Psi$ is not.

For instance, the goal may be an authentication property, in which
case $\Phi$ may say that one party (the initiator, e.g.)~has executed
a run with certain fresh values and uncompromised keys.  We will call
a participant that is following the protocol a \emph{regular}
participant.  The system's goal $\Psi$ may then say that a regular
responder run matches this regular initiator run.  The adversary will
want to exhibit a situation in which there is no such matching run.

If the goal is a non-disclosure goal, then $\Phi$ may say that one
party (the initiator, e.g.)~has executed a run with certain fresh
values and uncompromised keys; has computed a particular session key
$k$; and that same value $k$ has been observed unprotected on the
network.  In this case, any structure that satisfies $\Phi$ is a
counterexample; $\Psi$ doesn't matter, and can be the always-false
formula $\bot$.

Hence, the adversary must ensure that certain equations are satisfied.
For instance, in the non-disclosure goal just mentioned, the adversary
must ensure that the session key $k$ computed by the initiator is
equal to the key observed on the network.  The adversary must also
ensure, for each message received by a regular participant, either
that it is obtained from an earlier transmission, or else that the
adversary can compute it with the help of earlier transmissions.  This
condition also requires solving equations:  The adversary must obtain
or compute values that will equal the messages that the regular
participants are assumed to receive.

Thus, the core of the adversary's job is to solve equations by
transforming the transmissions of the regular participants via a set
of computational abilities.  These include the ability to encrypt
given key and plaintext; to decrypt given decryption key and
ciphertext; and to execute exponentiation and the algebraic operations
of the group and field.  We codify this as a game between the system
and the adversary.

\begin{enumerate}
  \item The system chooses a security goal $\Phi\limp\Psi$, involving
  secrecy, authentication, key compromise, etc., as in
  Section~\ref{sec:xmpls}.  
  %

  \item The adversary proposes a potential counterexample $\skel$
  consisting of local regular runs with equations between values in
  reception of transmission events, e.g.~an equation between session
  keys as computed by two participants, or a regular participant and a
  disclosed value.\label{clause:adv:skel}
  \item For each message reception node in $\skel$, the adversary
  chooses a recipe, intended to produce an acceptable message, using
  the computational abilities.  The adversary may use earlier
  transmission events on regular strands to build messages for
  subsequent reception events.\label{clause:adv:strategy}

  These recipes determine a set of equalities between the values
  computed by the adversary and the values ``expected'' by the
  recipient (i.e.~acceptable to the recipient).  They are the
  \emph{adversary's proposed equations}.
  \item The adversary wins if his proposed equations are valid in
  ${\thefield}_q$, for infinitely many primes
  $q$.\label{clause:adv:winning} 
\end{enumerate}
Concentrating on the field values, if the proposed equations are
valid, then whatever choices the regular participants make for their
random exponents, the adversary's recipes should establish the
equalities.  In effect, the adversary is choosing recipes before the
regular participants choose their exponents.  For this reason, we
regard the choices of the regular participants as \emph{field
  extension elements}.  Thus, we will now introduce field structures,
and mention how they may be extended with new extension elements.

\paragraph{Fields and their cyclic groups.}  We use
$(\field,\zero,\one,\plus,\minus,\ftimes,\fdivide)$ as the signature
for fields; $\field$ is the sole sort.  The \emph{field theory}
contains the axioms:
\begin{enumerate}
  \item $\plus$ and $\ftimes$ are associative and commutative, and
  satisfy the distributive law; 
  \item $\zero$ is an identity element for $\plus$ and $\one$ is an
  identity element for $\ftimes$; 
  \item $\minus$ is inverse to $\plus$; and 
  \item
  $\forall x,y\colon\field\qdot y\not= \zero\limp (x/y)\ftimes y=x$.
\end{enumerate}
A structure for this signature is said to be a \emph{field} iff it
satisfies these axioms.
We augment the field signature to introduce the cyclic group
structure:
\[ \Sigma_G =
  (\field,\group,\zero,\one,\plus,\minus,\ftimes,\fdivide,\gen,\expt),
\] 
where $\gen$ has arity $()\rightarrow\group$ and $\expt$ has arity
$\group\times\field\rightarrow\group$.  Cyclic groups satisfy the
axioms:
\begin{enumerate}
  \item $\forall h\colon\group\qdot\exists x\colon\field\qdot
  \expt(\gen,x)=h$; 
  \item $\forall x,y\colon\field\qdot
  \expt(\gen,x)=\expt(\gen,y)$ implies $x=y$;
  \item $\forall h\colon\group\qdot\expt(h,\one)=h$; and
  \item $\forall h\colon\group,\, x,y\colon\field\qdot
  \expt(\expt(h,x),y)=\expt(h,(x\cdot y))$.
\end{enumerate}
We will use familiar notations, writing e.g.~$(h^x)^y=h^{xy}$ for the
body of the last cyclic group axiom.  The first two axioms ensure that
$\expt(\gen,\cdot)$ is a bijection between the field and the group.
The third axiom fixes how this bijection acts on scalars in the field,
e.g.~rationals in $\mathbb{Q}$.  Since we can always write group
operations by adding the exponents, we have no separate group
operation in the signature.  The group identity element is
$\gen^\zero$.

In this paper, we will in fact consider only protocols $\Pi$ in which
the regular participants use only the sub-signature
$(\field,\one,\ftimes,\fdivide)$.  That is, they never use the
additive structure $\zero,\plus,\minus$ in $\Pi$.  We have proved that
if the regular participants do not use the additive structure, then
the advesary will never need it either.  Every attack the adversary
can achieve against such a protocol $\Pi$, the adversary can achieve
using only the multiplicative structure~\cite{LiskovThayer14}.

However, our conclusions about these protocols are motivated by the
natural underlying mathematical structures, namely the fields and the
cyclic groups on which they act.

\paragraph{Transcendentals.}  We can always extend a given field
${\thefield}$ with new elements $x_1,\ldots,x_n$; the extended field,
written ${\thefield}(x_1,\ldots,x_n)$, is then generated from the
polynomials $P_i$ in $x_1,\ldots,x_n$ with coefficients from
${\thefield}$.  Specifically, the members of
${\thefield}(x_1,\ldots,x_n)$ come from the rational expressions
$P_1/P_2$ where $P_2$ is not the identically $\zero$ polynomial.  Two
rational expressions represent the same field element when the usual
rules for polynomial multiplication (or factoring) and cancellation
imply that they are equal.  The field elements are thus the
equivalence classes of rational expressions partitioned by these
rules.

In algebra, one is concerned with two kinds of field extension
elements.  \emph{Algebraic} extension elements are introduced with a
polynomial of which the new element will be a root.  For instance, the
rationals $\mathbb{Q}$ have no square root of two, $\mathbb{Q}$ has a
proper algebraic extension $\mathbb{Q}(x)$ where $x$ is subjected to
the polynomial $x^2-2$.  Alternatively, a field extension may not be
subjected to any polynomial.  This is of course necessary to introduce
transcendental numbers such as $\pi$ and $e$, since they are not roots
of any polynomial with coefficients from $\mathbb{Q}$.  These
unconstrained field extension elements are called
\emph{transcendentals}.

We will use transcendental field extension elements to represent the
randomly chosen exponents of the regular participants in protocol
runs.  This has the consequence:  If an adversary's proposed equations
include some $P_1=P_2$ involving a transcendental $x$, then it will be
true in a field ${\thefield}$ only if $P_1-P_2$ is identically $\zero$
in ${\thefield}$.  This matches our winning condition
Clause~\ref{clause:adv:winning}, at least for ${\thefield}$.

\paragraph{The rationals $\mathbb{Q}$.}  We will focus on the base
field $\mathbb{Q}$, since a polynomial $P_1$ is identically $\zero$ in
$\mathbb{Q}$ iff there are infinitely many primes $q$ such that $P_1$
is identically $\zero$ in ${\thefield}_q$.  Certainly, if $P_1$ is
$\zero$ in $\mathbb{Q}$, it is $\zero$ in every ${\thefield}_q$, which
can only add equations, not eliminate them.  On the other hand, a 
polynomial of degree $d$ can have at most $d$ zeros in any field.
Thus, if $P_1$ is identically $\zero$ over ${\thefield}_q$ but not over
$\mathbb{Q}$, then $q\le d$.
%
%
%
However, since every polynomial has finite degree, there are only
finitely many such exceptions $q$.

\begin{definition}
  \label{def:alg:trans}
  Fix an infinite set of transcendentals $\trsc$.  Define:
  \begin{description}
    \item[$\Sigma_F$] to be the signature $\Sigma_D$ augmented with a
    sort {\sctrsc}, with the sort inclusion
    {\sctrsc}$\le\field$; 
    \item[${\thefield}$] to be the field $\mathbb{Q}(\trsc)$ of
    rational expressions in $\trsc$;
    \item[${\thegroup}$] to be the cyclic group generated from
    ${\thefield}$ by $\expt$.  
  \end{description}
  ${\thefield}$ and ${\thegroup}$ furnish an algebra of the signature
  $\Sigma_F$.  
\end{definition}

\subsection{Building messages}
\label{sec:algebra:free}

\paragraph{Signatures, algebras, and structures.}  As we have just
illustrated, our messages form certain order sorted
algebras~\cite{GoguenMeseguer92}, although we will not need explicitly
overloaded symbols.  An order sorted signature is a triple
$\Sigma=(S, \leq, \ccolon)$ where:
\begin{description}
  \item[$S$] is a set of sort names; 
  \item[$\leq$] is a partial order on $S$; and 
  \item[$\ccolon$] is a finite map.  Its domain is a finite set of
  function constants, and it returns an \emph{arity}
  $\s_1\times\ldots\times \s_k\rightarrow \s_{0}$ for each of those
  function constants $f$.  We write
  $f\ccolon\s_1\times\ldots\times \s_k\rightarrow \s_{0}$ to assert
  that the arity of $f$ in $\Sigma$ is
  $\s_1\times\ldots\times \s_k\rightarrow \s_{0}$.
\end{description}
The \emph{function symbols} of $\Sigma$ form the domain
$\dom(\ccolon)$; $c\in\dom(\ccolon)$ is an \emph{individual constant}
if it has zero argument sorts $c\ccolon()\rightarrow \s_{0}$, or as we
will write $c\ccolon\s_{0}$.
%
%
%

A structure $\mdlA$ is a $\Sigma$-\emph{algebra} iff
\begin{enumerate}
  \item $\mdlA$ supplies a set $\s(\mdlA)$ for each sort
  $\s$ in $S$, where
  \item $\s_1\leq \s_2$ implies
  $\s_1(\mdlA)\subseteq \s_2(\mdlA)$; 
  \item If $\s_1(\mdlA)\cap \s_2(\mdlA)\not=\emptyset$
  then for some sort $\s\in S$, $\s\le \s_1$ and $\s\le
  \s_2$; \label{clause:no:overlap} and
  \item $\mdlA$ supplies a function
  $f_{\mdlA} \colon \s_1(\mdlA)\times\ldots\times
  \s_k(\mdlA)\rightarrow \s_0(\mdlA)$ for each function
  symbol $f\ccolon s_1\times\ldots\times s_k\rightarrow s_0$.
\end{enumerate}
We assume (Clause~\ref{clause:no:overlap}) that sorts overlap only if
they share a common subsort; when $\s_1\le\s_2$, then $\s_1$ is itself
a common subsort.

Strictly speaking, the algebra is the \emph{map} that associates each
sort $\s$ to its interpretation $\s(\mdlA)$ and each function symbol
$f$ to its interpretation $f_{\mdlA}$.  We will often speak as if the
algebra is the range of the interpretation, but we will make use of
the map wherever needed, e.g.~in Section~\ref{sec:skeletons:lang}.  

A \emph{homomorphism} $H\colon\mdlA\rightarrow\mdlB$ is a
sort respecting map from the domains of $\mdlA$ to the domains
of $\mdlB$ that respects the function symbols:
$f_{\mdlB}(H(v_1),\ldots,H(v_k))=H(f_{\mdlA}(v_1,\ldots,v_k))$.

 We construct our message algebras in two steps.  We first start with
\emph{basic} values that include the field ${\thefield}$ and cyclic
group ${\thegroup}$, as well as other convenient values such as names,
nonces, texts and keys.  We will then freely build messages above
these basic values via tupling and cryptographic operations such as
encryption, hashing, and digital signature.  We will refer to the
algebra generated from basic values by these free constructors as a
\emph{constructed} algebra.

What we claim here is true regardless of the exact choice of
constructors, and of the ``convenient'' basic values we mention.
Moreover, some of our claims are unchanged as the structure of
${\thefield},{\thegroup}$ is extended, as we will mention in
connection with future work.
%
%
In the meantime, we will focus on a particular exemplar.

\paragraph{Basic algebra.}  Let:
\[ S_0=\{{\scskey}, {\scname}, {\sctext}, {\scakey}, {\sctrsc},
  \sccreate, \field, \group, {\scsbasic}\}
\]
be a set of sorts, with ${\sctrsc}\le\field$ and
${\scskey}, {\scname}, {\sctext}, {\scakey}, \sctrsc\le{\sccreate}$.
The adversary may create values of these sorts.  To obtain other
values of field sort (for instance) she uses the constant $\one$ and
the field operations.

The sort $\scsbasic$ is the top sort; all other sorts are below it.
Moreover, the sorts
${\scskey}, {\scname}, {\sctext}, {\scakey}, \field,\group$ are all
flat (hence by Clause~\ref{clause:no:overlap}~disjoint).  We also
require function symbols:
\[  \kinv\ccolon{\scakey}\rightarrow{\scakey} \quad
  \pk\ccolon{\scname}\rightarrow{\scakey} \quad
  \ltk\ccolon{\scname}\times{\scname}\rightarrow {\scskey}
\]
to take the inverse of an asymmetric key (namely the other member of a
public/private key pair); to associate a public key with a name; and
to associate a long term shared symmetric key with a pair of names.
We will refer to this signature as $\Sigmabasic$.  
%


\begin{definition}
  \label{def:alg:basic}
  Fix a $\Sigmabasic$ algebra $\Algebrabasic$ containing
  ${\thefield},{\thegroup}$ satisfying the field and group axioms, and
  infinitely many values of each sort
  ${\scskey}, {\scname}, {\sctext}, {\scakey}$,
  satisfying the three axioms:
  \begin{enumerate}
    \item that inverse satisfies $\kinv(\kinv(k))=k$; 
    \item $\pk,\ltk$ are injective; and 
    \item every member of ${\scsbasic}$ is in one of the subsorts
    ${\scskey}$, ${\scname}$, $ {\sctext}$,
    $ {\scakey}$, $ \field$, $\group$.
  \end{enumerate}
  Strictly speaking, the algebra is the map from $\Sigmabasic$ into
  $\Algebrabasic$.  
\end{definition}

\paragraph{Message algebras.}  Having built the basic part of the
algebra, we augment it by applying free constructors for tupling and
for cryptographic operations, yielding values in a new top sort
{\scmesg}.  The exact set of operations is not crucial; however,
we assume that they are partitioned into:
\begin{description}
  \item[Tupling operations] $\tau(m_1,\ldots,m_i)$ where the $m_j$ are
  drawn from {\scmesg} and the result is in {\scmesg};
  \item[Asymmetric operations] $\enca{m}{K}$ for $m$ in {\scmesg}
  and $K\colon{\scakey}$, yielding a result in {\scmesg}; and
  \item[Symmetric operations] $\encs{m}{K}$ for
  $m,K\colon{\scmesg}$, yielding a result in {\scmesg}.
\end{description} 
We write $\enc{m}{K}$ when we do not care to distinguish $\enca{m}{K}$
from $\encs{m}{K}$.  We call any unit $\enc{m}{K}$ an
\emph{encryption}, even though in practice they can be used as hashes,
digital signatures, and so on.  

Multiple entries in each of these categories are useful, for instance:
Multiple tupling operations can represent distinct formats that cannot
collide~\cite{MoedersheimKatsoris14}; multiple asymmetric operations
can represent encryption vs.~digital signature; and multiple symmetric
operations can represent a cipher vs.~a hash function.

We will illustrate our approach using a single asymmetric operator.
We will use a pair of symmetric operators representing a cipher and a
hash function $\hash{m}$.  We regard $m$ as the \emph{key} argument,
the plaintext being some vacuous value $\zero$;
i.e.~$\hash{m}=\encs{\zero}{m}$.  We use a single tupling operation of
untagged pairing $(m_1,m_2)$.

\begin{definition}
  \label{def:alg:full}
  Fix an extension $\Sigmaconstr$ of $\Sigmabasic$ by augmenting $\Sigmabasic$
  with a new top sort {\scmesg}, and with the free pairing,
  symmetric, and asymmetric operators with result sort {\scmesg},
  as in the previous paragraph.

  The $\Sigmaconstr$-algebra $\Algebraconstr$ is the closure of
  $\Algebrabasic$ under the (free) operators of
  $\Sigmaconstr\setminus\Sigmabasic$.

  Strictly speaking, the algebra is the map from $\Sigmaconstr$ into
  this closure.

\end{definition}

We partition the operators as mentioned because the members of each
partition have corresponding rules for adversary derivability,
displayed in Fig.~\ref{fig:rules}.  These rules are either
introduction rules, named operator-$\uparrow$, or elimination rules,
named operator-$\downarrow$.  The adversary combines these rules to
obtain messages.

\begin{figure}[tb]
  \centering
  \begin{tabular}{l@{\qquad}r@{$\quad\vdash\quad$}l@{\qquad}r} 
    Rule & \small Premises & \small Concl. & \small  \\[1mm]
    $\tau$-$\uparrow$ & $m_1,\ldots,m_i$ & $\tau(m_1,\ldots,m_i)$ &  \\
    $\enca{\cdot}{\cdot}$-$\uparrow$ & $m,K$ & $\enca{m}{K}$ &  \\
    $\encs{\cdot}{\cdot}$-$\uparrow$ & $m,K$ & $\encs{m}{K}$ &  
    \\[1mm]
    $\tau_j$-$\downarrow$ & $\tau(m_1,\ldots,m_i)$ & $m_j$ & for each
                                                             $1\le
                                                             j\le i$\\ 
    $\enca{\cdot}{\cdot}$-$\downarrow$ & $\enca{m}{K}, K^{-1}$ & $m$ & \\
    $\encs{\cdot}{\cdot}$-$\downarrow$& $\encs{m}{K}, K$ & $m$ & 
  \end{tabular}
  \caption{Derivability rules for free constructors}
  \label{fig:rules}
\end{figure}

These have a weak Gentzen-style normalization
property~\cite{Gentzen35,MarreroEtAl97,Paulson98,GuttmanThayer02}.
Suppose that an elimination rule is used immediately after an
introduction rule.  Unless the introduction rule is providing the
decryption key for an application of
$\encs{\cdot}{\cdot}$-$\downarrow$, the operator produced by the
introduction rule must be the same as the operator consumed by the
elimination rule.  In this case, the result of the elimination rules
is equal to one of the inputs to the introduction rule.  Hence, these
pairs of rules provide the adversary nothing new, and we can assume
they do not occur.

The algebras $\Algebrabasic,\Algebraconstr$ are ground algebras in the
sense that they do not contain any variables.  Transcendentals in
particular are not variables, but are particular objects that help to
make up fields of rational expressions.  We will, however, also
consider objects that involve variables, namely linguistic
\emph{terms} and \emph{formulas}.  However, even terms that contain no
variables are not members of $\Algebraconstr$; they are simply
linguistic terms that may refer to members of $\Algebraconstr$.

\input{impl.tex}

\subsection{Logic:  Languages and structures}
\label{sec:algebra:logic}

$\Sigma=(S, \leq, \ccolon, \cccolon)$ is a \emph{first order signature}
iff:
\begin{description}
  \item[$(S, \leq, \ccolon)$] is an order sorted signature, and
  \item[$\cccolon$] is a finite map from relation constants to arities
  $s_1\times\ldots\times s_k$ from $S$.  We assume an equality symbol
  $=^\s$ is in $\dom(\cccolon)$ for each $\s\in S$, with
  ${=^\s}\cccolon\s\times\s$.
\end{description}
A first order language $\lang{}=\lang{}(\Sigma,\kind{Var})$ is
determined from a first order signature and a set of variables.  Given
a supply of sorted variables $\kind{Var}$ with infinitely many
$v_i\colon\s$ for each sort $\s\in S$,
$\lang{}=\lang{}(\Sigma,\kind{Var})$ is the first order language with
terms and formulas defined inductively:
\begin{description}
  \item[Terms of ${\lang{}}$:] If $\s_0\le \s$, then each variable
  $v_i\colon\s_0$ is a
  term of sort $\s$.   \\
  If (i) $t_1\colon \s_1,\ldots, t_k\colon\s_k$ are terms of the sorts
  shown, (ii) $f\ccolon s_1\times\ldots\times s_k\rightarrow s_0$ is a
  function symbol, and (iii) $\s_0\le \s$, then $f(t_1,\ldots, t_k)$
  is a term of sort~$\s$.
  \item[Formulas of ${\lang{}}$:]  If (i)
  $t_1\colon \s_1,\ldots, t_k\colon\s_k$ are terms of sorts shown, and
  (ii) $R$ is a relation symbol
  $R\cccolon s_1\times\ldots\times s_k$, then $R(t_1,\ldots, t_k)$
  is a formula.   \\
  If $\phi$ and $\psi$ are formulas, and $x\colon\s$ is a variable in
  $\kind{Var}({\lang{}})$, then:
  \[ \phi\lor\psi \quad \phi\land\psi \quad \phi\limp\psi \quad
    \lnot\phi \quad \exists x\colon\s\qdot \phi \quad \forall
    x\colon\s\qdot \phi 
  \]
  are also formulas.  
\end{description}
The \emph{occurrences} of variables in terms and formulas are defined
as usual, and occurrences are \emph{bound} if they are within the
scope of a quantifier governing that variable.  We write
$\term({\lang{}})$ for the terms of ${\lang{}}$;
$\form({\lang{}})$ for the formulas of ${\lang{}}$; and
$\sent({\lang{}})$ for the formulas that are \emph{sentences},
namely those in which no variable has a free occurrence.

If we refer to ${\lang{}}$ as a language over an order sorted
signature $(S, \leq, \ccolon)$, we mean that it is a language over the
first order signature $(S, \leq, \ccolon, \cccolon)$ where
$\dom(\cccolon)$ contains only the equality relations.

A structure $\mdlA$ is an ${\lang{}}$-\emph{structure} iff:
\begin{enumerate}
  \item Its restriction $\restrict{\mdlA}{\Sigma}$ to the underlying
  algebraic signature is a $(S, \leq, \Sigma)$-algebra;
  \item For each relation symbol $R$, with
  $R\cccolon s_1\times\ldots\times s_k$, $\mdlA$ supplies a
  relation
  $R_{\mdlA} \subseteq \s_1(\mdlA)\times\ldots\times
  \s_k(\mdlA)$ .
\end{enumerate}
We always interpret $=^\s_{\mdlA}$ as the standard equality on
$\s({\mdlA})$.  As before, the structure is in fact the map from the
signature of $\lang{}$ into its results.  

A \emph{variable assignment} (aka~\emph{environment}) $\eta$ is a map
$\eta\colon\kind{Var}({\lang{}})\rightarrow\mdlA$ such that
for each variable $v\colon \s$, $\eta(v)\in\s(\mdlA)$.  That is,
$\eta$ is sort-respecting.  If $\phi\subseteq\form({\lang{}})$ is a
formula of ${\lang{}}$ and $\eta$ is a variable assignment, then
$\mdlA\models_\eta\phi$ is defined to mean that $\phi$ is
satisfied in $\mdlA$ under $\eta$ in the usual style following
Tarski.  Specifically: 
\begin{description}
  \item[Terms:]  Extend $\eta$ from variables to $\term({\lang{}})$
  using $\mdlA$ by stipulating inductively that
  $\eta_{\mdlA}(f(t_1,\ldots,t_k))$ will be equal to
  $f_{\mdlA}(\eta_{\mdlA}(t_1),\ldots,\eta_{\mdlA}(t_k))$.
  \item[Atomic formulas:]  Stipulate
  $\mdlA\models_\eta R(t_1,\ldots,t_k)$ iff
  $\eta_{\mdlA}(t_1),\ldots,\eta_{\mdlA}(t_k)\in
  R_{\mdlA}$;
  \item[Compound formulas:]   Stipulate inductively that: \\
  $\mdlA\models_\eta\phi\lor\psi$ iff
  $\mdlA\models_\eta\phi$ or else $\mdlA\models_\eta\psi$; \\
  $\mdlA\models_\eta\phi\land\psi$ iff
  $\mdlA\models_\eta\phi$ and also
  $\mdlA\models_\eta\psi$; \\
  $\mdlA\models_\eta\phi\limp\psi$ iff
  $\mdlA\not\models_\eta\phi$ or else
  $\mdlA\models_\eta\psi$; and \\
  $\mdlA\models_\eta\lnot\phi$ iff
  $\mdlA\not\models_\eta\phi$.

  \medskip If $\eta$ is a variable assignment, $v\colon\s$ is a
  variable, and $a\in\s(\mdlA)$ is a member of the domain for
  sort $\s$, then $\eta[v\mapsto a]$ will mean the variable assignment
  that returns $a$ for $v$ and agrees with $\eta$ for all other
  variables.  Using
  this notation:  \\
  $\mdlA\models_\eta\exists v\colon\s\qdot\phi$ iff, for some
  $a\in\s(\mdlA)$,
  $\mdlA\models_{\eta[v\mapsto a]}\phi$; \\
  $\mdlA\models_\eta\forall v\colon\s\qdot\phi$ iff, for every
  $a\in\s(\mdlA)$, $\mdlA\models_{\eta[v\mapsto a]}\phi$.
\end{description}
$\mdlA\models\phi$ means that
$\mdlA\models_\eta\phi$ for all $\eta$.

$T$ is a \emph{theory} of ${\lang{}}$ iff
$T\subseteq\sent({\lang{}})$ is a set of sentences of ${\lang{}}$.
%
%
An ${\lang{}}$-structure $\mdlA$ is a \emph{model} of $T$,
written $\mdlA\models T$, iff $\mdlA\models \phi$ for
every $\phi\in T$.

Suppose that $\Sigma_0=(S, \leq, \ccolon)$ and
$\Sigma=(S, \leq, \ccolon,\cccolon)$.  If $\mdlA$ and
$\mdlB$ are $\Sigma$-structures, then a
$\Sigma$-\emph{homomorphism}
$H\colon\mdlA\rightarrow\mdlB$ is a sort respecting map
from the domains of $\mdlA$ to the domains of $\mdlB$ such
that
\begin{enumerate}
  \item $H$ restricts to an algebra homomorphism
  $(\restrict H{\Sigma_0}) \colon (\restrict{\mdlA}{\Sigma_0})
  \rightarrow (\restrict{\mdlB}{\Sigma_0})$;
  \item For each relation symbol
  $R\cccolon s_1\times\ldots\times s_k$, and each $k$-tuple
  $(v_1,\ldots v_k)\in\s_1(\mdlA)\times\ldots\times
  \s_k(\mdlA)$,
  \[  (v_1,\ldots v_k)\in R_{\mdlA} \mbox{ implies }
    (H(v_1),\ldots H(v_k))\in R_{\mdlB} .  
  \]
\end{enumerate}
Suppose that $\Phi\in\form({\lang{}})$ uses only the logical
connectives conjunction $\land$ and disjunction $\lor$ and the
existential quantifier $\exists$, and has no occurrences of negation
$\lnot$ or implication $\limp$ or the universal quantifier $\forall$.
Then formula $\Phi$ is \emph{positive existential} (PE).  If
$H\colon\mdlA\rightarrow\mdlB$ is an
${\lang{}}$-homomorphism, then if $\Phi$ is PE,
\[ \mdlA\models_\eta\Phi \mbox{ implies }
  \mdlB\models_{H\circ\eta}\Phi .  
\]
That is, satisfaction is preserved for positive existential formulas,
when we extend the variable assignment $\eta$ to $\mdlB$ by
composing with $H$.  In particular, when $\Phi\in\sent({\lang{}})$
has no free occurrences of variables, then $\mdlA\models\Phi$
implies $\mdlB\models\Phi$.

Suppose $\mdlA$ is a $\Sigma=(S, \leq, \ccolon, \cccolon)$-structure,
and ${\lang{}}'$ is a
$\Sigma'=(S', \leq', \ccolon', \cccolon')$-language, where
$S'\subseteq S$, $\leq'$ is $\leq$ restricted to $S'$, and $\ccolon'$
and ${\cccolon}'$ are subfunctions of $\ccolon$ and $\cccolon$
respectively.  Then the notion of satisfaction carries over, as if
${\lang{}}'$ is a sublanguage of a larger $\Sigma$-language.  The
semantic clauses are all local to the sorts and vocabulary that
actually appear in the formula.  In section~\ref{sec:strands} we will
consider restricted languages in relation to richer structures.



%% file: impl.tex


\subsection{Unification and Matching}
\label{sec:algebra:unification}

{\cpsa}'s good performance is due to the use of efficient algorithms
for unification and matching.  Prior versions of {\cpsa} had an
algebra with just one equation---the double inverse of an asymmetric
key is the same as the key.  It is straightforward to modify standard
algorithms for syntactic unification and matching to honor this
equation.

In this implementation, the exponents satisfy $\mathit{AG}$, the
equations for a free Abelian group.  There are efficient algorithms
for unification
modulo~$\mathit{AG}$~\cite[Section~5.1]{BaaderSnyder01}, and matching
is equivalent to unification while treating some terms as constants.
{\cpsa} uses an algorithm that reduces the problem to finding integer
solutions to an inhomogeneous linear equation with integer
coefficients.  The equation solver used is from The Art of Computer
Programming~\cite[Pg.~327]{Knuth81}, and~\cite{agum09} provides an
implementation of $\mathit{AG}$ unification and matching in Haskell.

As expected, the implementation uses two sorts for exponents, with the
sort for transcendentals being a subsort of the one for exponents.
The key observation is that there are no equations between
transcendentals, thus syntactic unification applies.

Section~6.1 of the Handbook~\cite{BaaderSnyder01} describes a general
method for combining unification algorithms.  The method specifies
making many non-deterministic choices that would make unification
expensive.  We take advantage of the fact that both syntactic and
$\mathit{AG}$ unification are relatively simple algorithms, and using
techniques explained in~\cite{KepserRichts99,Liu12}, have a fast
implementation of unification and matching.
Appendix~\ref{sec:unification} presents the algorithm.  For more
detail, see Appendix~\ref{sec:unification}.  


%% file: strands.tex

We will now introduce the main notions of the strand space framework,
which underlies {\cpsa}.  A protocol is a set of roles, together with
some auxiliary information, and each role has instances, which are the
behaviors of individual regular principals on different occasions.
Because they have instances, it is natural that protocols involve
\emph{variables}.  Thus, we will build them from linguistic items,
namely terms $\term({\lang{}})$ in suitable ${\lang{}}$.  On the other
hand, bundles are our execution models, and they involve specific
values.  In particular, in the {\DHKA} context, they may involve field
elements, and they may depend on the properties of these field
elements---such as identities involving polynomials---to be successful
executions.  Thus, we will build bundles from non-linguistic items,
namely the values in our message algebra $\Algebraconstr$
(Def.~\ref{def:alg:full}).

Strands represent local behaviors of participants, or basic adversary
actions.  We use them to define the roles of protocols, and this usage
requires terms containing variables.  Strands are also constituents of
bundles, and that usage requires concrete values and field elements.
Hence, we will allow strands of both kinds, and we will henceforth use
the word \emph{message} to cover either terms in some
$\term({\lang{}})$ and also members of a message algebra, specifically
$\Algebraconstr$.  We will sometimes refer to terms in
$\term({\lang{}})$ as ``formal'' messages, and to values in
$\Algebraconstr$ as ``concrete'' messages.

\subsection{General notions}
\label{sec:strands:generally}

We will now introduce strands, bundles (our notion of execution), and
protocols.  

\paragraph{Strands.}  A strand represents a single local session of a
protocol for a single regular participant, or else a single adversary
action.  Suppose that $\Algebra$ is a set of messages such as
$\Algebraconstr$ or $\term({\lang{}})$, and $+,-$ are two values we
use to represent the direction of messages, representing transmission
and reception, respectively.  When $m\in\Algebra$, we write $+m$ and
$-m$ for short for the pairs $(+,m)$ and $(-,m)$.  We write
$\pm\Algebra$ for the set of all such pairs, and $(\pm\Algebra)^+$ for
the set of non-empty finite sequences of pairs.

By a \emph{strand space} $(\strands,\tr)$ \emph{over} $\Algebra$ we
mean a set of objects $\strands$ equipped with a trace function
$\tr\colon \strands\rightarrow(\pm\Algebra)^+$.  For each
$s\in \strands$, $\tr(s)$ is a finite sequence of transmission and
reception events.  Other types of events have also been used, for
instance to model interaction with long term
state~\cite{Guttman12,GuttmanLRR15}, but only transmission and
reception events will be needed here.  We often do not distinguish
carefully between a strand $s$ and its trace $\tr(s)$.

The length $\length{s}$ of a strand $s$ means the number of entries in
$\tr(s)$.  We will use the same notation $\length\alpha$ for the
length of a sequence $\alpha$ and for the cardinality $\length S$ of a
set $S$.  When $1\le i\le\length{s}$, we regard the pair $(s,i)$ as
representing the $i^{\mathrm{th}}$ event of $\tr(s)$; we call $(s,i)$
a \emph{node} and generally write it $\strandnode s i$.  If $n$ is a
node, we write $\msg(n)$ for the message it sends or receives.  The
\emph{direction} $\direction(n)$ of a node is either $+$ for
transmission or $-$ for reception.  Thus, when $n=\strandnode s i$,
$\tr(s)[i]=(\direction(n), \msg(n))$.

We write $n_1\Rightarrow n_2$ when the node $n_2$ immediately follows
$n_1$ on the same strand, i.e.~$n_1\Rightarrow n_2$ holds iff, for
some strand $s$ and integer $i$, $n_1=\strandnode s i$ and
$n_2=\strandnode s {i+1}$.  We write $n_1\Rightarrow^+ n_2$ for the
transitive closure:  it holds iff, for some strand $s$ and integers
$i,j$, $i<j$, $n_1=\strandnode s i$ and $n_2=\strandnode s j$.  The
reflexive-transitive closure $n_1\Rightarrow^* n_2$ is defined by the
same condition but with $i\le j$.

\paragraph{Bundles.}  Fix a strand space $(\strands,\tr)$ over
$\Algebraconstr$.  A binary relation $\rightarrow$ on nodes is a
\emph{communication relation} iff $n_1\rightarrow n_2$ implies that
$n_1$ is a transmission node, $n_2$ is a reception node,
$\msg(n_1),\msg(n_2)\in\Algebraconstr$, and $\msg(n_1)=\msg(n_2)$.

We require that the range of $\msg(\cdot)$ be in $\Algebraconstr$ so
that this last equality has a definite meaning.  If the messages here
were terms, then the equality would depend on what structure is chosen
to interpret the terms.  We will use the bundle notion only when this
structure is already selected; in this paper, our selection is
$\Algebraconstr$.

\begin{definition}
  \label{def:bundle}
  Let $\bnd=(\mathcal{N},\rightarrow)$ be a set of nodes together with
  a communication relation on $\mathcal{N}$.  $\bnd$ is a
  \emph{bundle} (over $\Algebraconstr$) iff:
  \begin{enumerate}
    \item \label{clause:bundle:backward} $n_2\in\mathcal{N}$ and
    $n_1\Rightarrow n_2$ implies $n_1\in\mathcal{N}$;
    \item $n_2\in\mathcal{N}$ and $n_2$ is a reception node implies
    there exists a unique $n_1\in\mathcal{N}$ such that
    $n_1\rightarrow n_2$; and 
    \item Letting $\Rightarrow_{\bnd}$ be the restriction of
    $\Rightarrow$ to $\mathcal{N}\times\mathcal{N}$, the
    reflexive-transitive closure
    $(\Rightarrow_{\bnd}\cup\rightarrow)^*$ is a well-founded
    relation.  We write $\preceq_{\bnd}$ for this relation.
  \end{enumerate}
  We write $\node(\bnd)$ for $\mathcal{N}$.  We say a strand $s$ is
  \emph{in} $\bnd$ iff there exists an $i$
  s.t.~$\strandnode s i\in\node(\bnd)$ (hence, in particular, its
  first node $\strandnode s 1\in\node(\bnd)$, by
  clause~\ref{clause:bundle:backward}).
\end{definition}
If the set of nodes $\node(\bnd)$ is finite, then the well-foundedness
condition is equivalent to saying that the finite directed graph
$(\mathcal{N},\rightarrow\cup\Rightarrow_{\bnd})$ is acyclic.
Protocol analysis uses the following \emph{bundle induction principle}
incessantly:
\begin{lemma}[see~\cite{ThayerHerzogGuttman99}]
  \label{lemma:bundle:induction}
  If $\bnd$ is a bundle, and $S\subseteq\node(\bnd)$ is non-empty,
  then $S$ contains $\preceq_{\bnd}$-minimal nodes.  
\end{lemma}
For instance, if $S$ is the set of nodes at which something bad is
happening, this principle justifies considering how it could first
start to go wrong, and examining the possible cases for that.  

Bundles furnish our model of execution.  Given a protocol $\Pi$, an
execution of $\Pi$ with an active adversary is a bundle in which every
strand represents either an initial part of a local run of some role
in $\Pi$, or else some adversary activity.  

\paragraph{Protocols.}  A protocol consists of a set of strands over
$\term({\lang{}})$, together with some auxiliary information that
provides assumptions, usually about fresh choices and uncompromised
keys.

Let $\Sigma_0$ be a subsignature of $\Sigmaconstr$ including
{\scmesg}, and let $\Sigma=(S, \leq, \ccolon, \cccolon)$ augment
$\Sigma_0$ with one new incomparable sort {\scnode}.  We will let
$\ccolon$ be identical with the arity function of $\Sigma_0$, and
require that the new relation symbols of $\cccolon$ all involve
{\scnode}.  Thus, letting ${\lang{}}$ be a language over
$\Sigma$, we can interpret ${\lang{}}$ over structures whose
{\scmesg}s belong to the structure $\Algebraconstr$.

We will define an ${\lang{}}$-protocol $\Pi$.  A protocol consists
of a finite set of strands over $\term({\lang{}})$ (the roles of
$\Pi$) together with a function $\assume$ from nodes to formulas in
$\form({\lang{}})$.  For simplicity, we also refer to the set of
roles of $\Pi$ by the symbol $\Pi$.  Thus, we say a protocol $\Pi$
consists of:
\begin{enumerate}
  \item A finite strand space $(\strands,\tr)$ over the terms
  $\term({\lang{}})$.  We call this finite set of strands the
  \emph{roles} of $\Pi$.  We often write $\rho\in\Pi$ to mean that
  $\rho\in\strands$, thereby reducing notation.  We write $\node(\Pi)$
  to mean $\{\strandnode\rho i\colon\rho\in\Pi$ and
  $1\le i\le\length\rho\}$; the image of this set under $\msg(\cdot)$
  consists of terms:  $\msg(\node(\Pi))\subseteq\term({\lang{}})$.

  The \emph{parameters of} $n\in\node(\Pi)$ and $\rho\in\Pi$ are the
  sets:  
  %
  \begin{eqnarray*}
    \kind{params}(n) &=& \{v\in\kind{Var}({\lang{}}) \colon\exists
                         n_0\qdot n_0\Rightarrow^* n\mbox{ and } v\in\fv(\msg(n_0)) \}. \\ 
    \kind{params}(\rho) &=& \{v\in\kind{Var}({\lang{}}) \colon\exists
                            i\qdot v\in\kind{params}(\strandnode\rho i) \}.
  \end{eqnarray*}


  \item A function
  $\assume\colon\node(\Pi)\rightarrow\form({\lang{}})$ such that,
  for a distinguished variable $v_n\colon{\scnode}$, for all
  $n\in\node(\Pi)$, the free variables
  $\fv(\assume(n))\subseteq\{v_n\}\cup\kind{params}(n)$.
\end{enumerate}
Suppose that $\rho\in\Pi$ and $n_1=\strandnode s j\in\node(\bnd)$,
where $\bnd$ is a bundle over $\Algebraconstr$, and $\eta$ is a
variable assignment
$\eta\colon \kind{Var}({\lang{}}) \rightarrow
\node(\bnd)\cup\Algebraconstr$.  Then $n_1$ is an \emph{instance of}
$n_2=\strandnode\rho i$ \emph{under} $\eta$ iff $i=j$ and,
inductively:
\begin{enumerate}
  \item $\eta(\msg(n_2))=\msg(n_1)$;  
  \item $\eta(v_n)=n_1$ and $\bnd\models_{\eta}\assume(n_2)$; and
  \item if $i=k+1$ where $k\ge 1$, then $n_0=\strandnode s k$ is an
  instance of $\strandnode\rho k$ under $\eta[v_n\mapsto n_0]$.
\end{enumerate}
That is, $\eta$ should send the terms of all nodes on the role up to
$n_2$ to the corresponding messages in $\Algebraconstr$, and all the
assumptions should be satisfied.  We will define the \emph{adversary
  strands} momentarily; relative to that notion, we can define when a
bundle is a possible run of a particular protocol.
\begin{definition}
  \label{def:bundle:pi}
  Suppose that $\bnd$ is a bundle (over $\Algebraconstr$) and $\Pi$ is
  a protocol.  $\bnd$ is a \emph{bundle of protocol} $\Pi$ iff, for
  every strand $s$ in $\bnd$, either
  \begin{enumerate}
    \item $s$ is an adversary strand, or else 
    \item letting $n_1=\strandnode s i$ be the last node of $s$ in
    $\node(\bnd)$, there is a role $\rho\in\Pi$ and a variable
    assignment $\eta_s$ s.t.~$n_1$ is an instance of
    $\strandnode\rho i$ under $\eta_s$.
  \end{enumerate}
\end{definition}
%
%
By altering $\eta_s$ at $v_n$, we also obtain assignments
$\eta[v_n\mapsto(\strandnode s j)]$ that witness for earlier nodes
along $s$ being instances of corresponding nodes $\strandnode\rho j$.

There may be several different $\rho_j$ of which a given $n_2$ is an
instance.  This may happen when the different $\rho_j$ represent
branching behaviors that diverge only after the events present in
$\bnd$.

We will make the languages $\lang{}$ more specific in
Section~\ref{sec:skeletons}.  

\subsection{The adversary}
\label{sec:strands:adversary}


The adversary's computational abilities include the remaining
algebraic operations, as well as the derivation rules summarized in
Fig.~\ref{fig:rules}.  The adversary can also select values.  We
represent these abilities as strands.  A computation in which the
adversary generates $f(a,b)$ as a function of potentially known values
$a$ and $b$ may be expressed as a strand
$-a\Rightarrow -b\Rightarrow +f(a,b)$.  If the inputs $a$ and $b$ are
in fact known, then the adversary can deliver them as messages to be
received by this strand.  The strand will then transmit the value
$f(a,b)$.  That in turn can either be delivered to a regular strand,
or else delivered to further adversary strands to compute more
complicated results.  In this way, we build up acyclic graph
structures that ``mimic'' any composite derivations that may be
generated by rules such as those in Fig.~\ref{fig:rules}.

The adversary can also select values of his own choosing.  We will
assume that this means he can originate values of the sorts
${\scskey}, {\scname}, {\sctext}, {\scakey}, {\sctrsc}$ and
$\mathbb Q$.  In particular, if the adversary wants, the adversary can
choose a random exponent, which we represent by a choice in {\sctrsc}.
The adversary can certainly choose the field value $\one$, and by
using addition and division, then adversary can obtain any rational in
$\mathbb Q$.  In all, this gives the adversary the abilities
represented in the strands shown in Fig.~\ref{fig:adversary:strands}.
\begin{figure}[tb]
  \begin{description} 
    \item[Creation:] \quad $+\gen$\quad  $+\one$ 
    \quad
    $+a$ \qquad 
    for \quad
    $a\colon\sccreate
    $
    \smallskip 
    \item[Multiplicative ops:]
    $-w_1\Rightarrow -w_2\Rightarrow +w_1\cdot w_2$ \qquad
    $-w_1\Rightarrow -w_2\Rightarrow +w_1/ w_2$  \\
    \phantom{$-w_1\Rightarrow -w_2\Rightarrow +w_1\cdot w_2$}
    $-h\Rightarrow -w\Rightarrow +\expt(h,w)$ 
    \item[Additive ops:]
    $-w_1\Rightarrow -w_2\Rightarrow +(w_1+_\thefield w_2)$ \qquad
    $-w_1\Rightarrow -w_2\Rightarrow +(w_1-_\thefield w_2)$ \\
    \phantom{$-w_1\Rightarrow -w_2+$}
    $-\expt(h,w_1)\Rightarrow -\expt(h,w_2)\Rightarrow
    +\expt(h,w_1+_\thefield w_2)$ \smallskip
    \item[Construction:] \quad
    $-m_1\Rightarrow -m_2\Rightarrow +(m_1,m_2)$\qquad
    $-m\Rightarrow -K\Rightarrow +\enc m K$
    \item[Destruction:] \quad $-(m_1,m_2)\Rightarrow +m_1$ \qquad
    $-(m_1,m_2)\Rightarrow +m_2$ \\
    \phantom{$-w_1\Rightarrow -w_2\Rightarrow +w_1\cdot w_2$}
    $-\enc m K\Rightarrow -K^{-1}\Rightarrow +m$
  \end{description}
  \begin{center}
    Symbols $+,-$ mean transmission and reception.  \\
    Symbols $+_\thefield,-_\thefield$ mean field addition and
    subtraction.
  \end{center}
  \caption{Adversary strands for $\Algebraconstr$.}
  \label{fig:adversary:strands}
\end{figure}

The first three groups concern basic values.  The first group presents
the creation strands, which allow the adversary to generate its own
values.  The second and third groups contains the algebraic
operations, in which the arguments to the operation are received and
the result is transmitted.  The second group contains the operations
that use the multiplicative structure of the field, while the third
group contains those that use the additive structure.  We separate
them because we will omit them from the powers of the adversary when
we restrict our protocols to those in which the regular participants
use only multiplicative structure (following~\cite{LiskovThayer14}).


The fourth group provides adversary strands that model the
constructive $\uparrow$-rules of Fig.~\ref{fig:rules}.  The fifth
group provides adversary strands to cover the destructive
$\downarrow$-rules of Fig.~\ref{fig:rules}.  We write $K^{-1}$ for
$\kinv(K)$ when $K\colon{\scakey}$; we stipulate that $K^{-1}=K$ for
any key not of that sort.

By ``routing'' the results of some adversary strands as inputs to
others, the adversary can build up finite acyclic graph structures
that do all the work of the inductively defined derivation relation
generated from Fig.~\ref{fig:rules} and corresponding algebraic rules.
Moreover, the adversary strands have an advantage:  They provide
well-localized, \emph{earliest} places where certain kinds of values
become available, as we will illustrate in Section~\ref{sec:tests}.
Hence, defining bundles to include these adversary strands provides a
convenience for reasoning.

We formalize the idea of routing results among adversary strands as
\emph{adversary webs}:
\begin{definition}
  \label{def:adversary:web}
  Suppose that $(W,\rightarrow)$ consists of a finite strand space in
  which every strand is an adversary strand, together with a
  communication relation $\rightarrow$ on the nodes of $W$.

  $(W,\rightarrow)$ is an \emph{adversary web} iff it is acyclic.  We
  write $\preceq_W$ for the well-founded partial order
  $(\Rightarrow_W\cup\rightarrow)^*$.
  A node $n\in\node(W)$ is a \emph{root} iff it is
  $\preceq_W$-maximal.  It is a \emph{leaf} iff it is a
  $\preceq_W$-minimal reception node.  
\end{definition}
We regard an adversary web as a method for deriving its roots,
assuming that its leaves are somehow obtained with the help of regular
protocol participants.  A web may not have any leaves, as happens when
all of its $\preceq_W$-minimal nodes are creation nodes, and thus
transmissions.  Every bundle contains a family of adversary webs,
which show how the adversary has obtained the messages received by the
regular participants, with the help of earlier regular transmissions.
By expanding definitions, we obtain the conclusions summarized in this
lemma:  
\begin{lemma}
  \label{lemma:adversary:web}
  Let $\bnd$ be a $\Pi$-bundle, and let $n\in\node(\bnd)$ be a regular
  reception node.  Define $W$ to contain an adversary strand $s$ iff
  there is a path in $\bnd$ from the last node
  $\strandnode s{\length{s}}$ to $n$ that traverses only adversary
  nodes.  Define $\rightarrow$ to be
  $\rightarrow_\bnd\cap(W\times W)$.  Then:
  \begin{enumerate}
    \item $(W,\rightarrow)$ is an adversary web.
    \item If $W$ is non-empty, then $W$ has a root $n_1$ such that
    $\msg(n_1)=\msg(n)$.  
    \item For every leaf $n_1$ of $(W,\rightarrow)$, then there exists
    a regular transmission node $n_0\in\node(\bnd)$ such that
    $n_0\preceq_\bnd n$ and $\msg(n_1)=\msg(n_0)$.
  \end{enumerate}
\end{lemma}
We call $(W,\rightarrow)$ \emph{the adversary web rooted at} $n$ \emph{in} $\bnd$.

\subsection{Properties of protocols}
\label{sec:strands:properties}
\label{sec:strands:restrictions}

\paragraph{Paths and positions.}
We often want to look inside values built by the free constructors,
identifying the parts by their position.  This is well-defined because
we use these positions only within constructors that are free.

A \emph{position}~$\pi$ is a finite sequence $\pi\in\mathbb{N}^*$ of
natural numbers.  The concatenation of positions $\pi$ and $\pi'$ is
written $\pi \append \pi'$.
A \emph{path} is a pair $p=(m,\pi)$.  The (1-based) submessage of $m$
\emph{at}~$\pi$, written~$m\termat \pi$, is defined recursively:
\begin{equation*}
\begin{array}{rcl}
  m\termat\seq{}&=&m;\\
  \tau(m_1,\ldots,m_i)\termat \seq{j} \append \pi&=&m_j\termat \pi\mbox{
  when $1\le j\le i$};\\
  \enc{m}{k}\termat \seq{1} \append \pi&=&m\termat \pi;\\
  \enc{m}{k}\termat \seq{2} \append \pi&=&k\termat \pi;\\
  \mbox{ otherwise:}\quad\qquad m\termat\pi &=& \perp
\end{array}
\end{equation*}
When $p = (m, \pi)$ is a path, we write $\termat p$ for $m \termat \pi$.
We often focus only on \emph{carried paths} which do not descend
into keys of encryptions.
%
\begin{definition}
  \label{def:paths:1} 
  Suppose $\pi$ is a position and $m$ is a message.  


  Path $p=(m, \pi)$ \emph{terminates} at its endpoint $\termat p$.

  If $\pi=\pi_1\append\pi_2$, then $p$ \emph{visits}
  $m \termat \pi_1$.  If moreover $\pi_2\not=\seq{}$,  $p$
  \emph{traverses} $m \termat \pi_1$.

  If $c\in\Algebraconstr$, then we say that $c$ \emph{is a unit} iff
  $c$ is not a tuple.  Equivalently, $c$ is a unit iff it is either
  basic, $c\in\Algebrabasic$, or else a cryptographic value
  $c=\enc m K$.
\end{definition}
Paths traverse only 
free constructors of $\Sigmaconstr$, and never traverse messages
belonging to the underlying given algebra $\Algebrabasic$.

Importantly, this helps justify adding a proper treatment of DH's
algebraic behavior without disrupting {\cpsa}'s existing approach to
cryptography and freshness.  Equational theories on $\Algebrabasic$
(e.g.~a commutative law) can never disrupt the unambiguous notion of
the submessage at a position $\pi$, whenever it applies.  This also
justifies using the path notion for terms in $\term({\lang{}})$ as
well as members of $\Algebraconstr$.  The latter mimic the structure
of the former unambiguously throughout the free operators.

\begin{definition}
  \label{def:paths}
  \begin{enumerate}
    \item A path $p=(m, \pi)$ is a \emph{carried path} iff $p$ never
    visits the key of an encryption, but only its plaintext.  That is,
    if $\pi=\pi_1\append\pi_2$ and $m\termat\pi_1=\enc{t}K$, and if
    $\pi_2\not=\seq{}$, then $\pi_2=\seq{1,\ldots}$ and not
    $\pi_2=\seq{2,\ldots}$.

    Message $m_0$ is \emph{carried in} $m_1$, written
    $m_0\ingredient m_1$ iff $m_0=m_1\termat \pi$ for some carried
    path $(m_1, \pi)$.

    \item A message $m$ \emph{originates} at node $n=\strandnode s i$
    iff $m\ingredient\msg(n)$, and $n$ is a transmission
    $\direction(n)=+$, and $1\le j<i$ implies
    $m\not\ingredient\msg(\strandnode s j)$.

    \item If $S$ is a set of nodes, then $m$ \emph{originates
      uniquely} in $S$ iff $m$ originates on exactly one $n\in S$.  It
    is \emph{non-originating} in $S$ iff it originates on no $n\in S$.
 
    \item Message $t$ is \emph{visible in} $m$ iff $t=m\termat \pi$
    for some carried path $p=(m, \pi)$ such that $p$ never traverses
    an encryption, but only tuples.
  \end{enumerate}
  %
%
\end{definition}
%
%
We will use the following lemma later, in proving
Lemma~\ref{lemma:simple:visible}; it says a field value originating on
an adversary node is the whole message of that node.
\begin{lemma}
  \label{lemma:adv:visible}
  Suppose that $\bnd$ is a bundle, and $p\colon\field$ originates at
  an adversary node $n\in\node(\bnd)$.  Then $p=\msg(n)$.
\end{lemma}
\begin{proof}
  The adversary node $n$ does not lie on an encryption, decryption,
  tupling, or separation strand, none of which originate any basic
  value.  Thus, $n$ lies on a creation or algebraic strand, and
  $\msg(n)\colon{\scsbasic}$.  Since no nontrivial path exists
  within a basic value, no value other than $\msg(n)$ is carried
  within $\msg(n)$.  \qed
\end{proof}

\paragraph{The ``acquired constraint'' on protocols.}  Since the roles
of a protocol definition form a strand space over $\term(\lang{})$,
their messages contain variables.  Some of these variables
$X\colon{\scmesg}$ can be instantiated by any message, meaning that
their instances do not have a predictable ``shape.''

Two things may go wrong if these variables occur first in a
transmission node.  First, syntactic constraints on the roles of a
protocol will not enforce invariants on the instances of the roles.
Instances may map a variable of sort message to any
$m\in\Algebraconstr$, so that the transmitted message may have any
value as a submessage.  Thus, even values that could never be computed
as a consequence of the workings of the protocol can originate as
arbitrary instances of $X\colon{\scmesg}$.

If variables of sort message may occur first in transmission nodes,
then we could not guarantee a finitely branching {\cpsa} search.  If
$X\colon{\scmesg}$ first appears in a transmission node on
$\rho\in\Pi$, then its instances will include messages of all formats.
Not all of them are instances of any finite set of more specific
forms, if those forms do not use further variables $Y\colon{\scmesg}$
of sort message.  This blows up the search.  If, instead, $X$ comes
from a reception node, it may still take infinitely many formats, but
the relevant differences arise in a finitely branching search, guided
by the forms of earlier transmissions.

A variable $X\colon{\scmesg}$ is \emph{acquired on} node
$n_1=\strandnode\rho i$ iff $X$ is carried in $\msg(n_1)$,
$\direction(n_1)=-$, and for all $n_0\Rightarrow^+n_1$,
$X\notin\fv(\msg(n_1))$.

$\Pi$ \emph{satisfies the acquired constraint} iff, for any
$n_1\in\node(\Pi)$ and any variable $X\colon{\scmesg}$, if
$X\in\fv(\msg(n_2))$, then there is a reception node $n_1$ such that
$n_1\Rightarrow^* n_2$ and $X$ is acquired on $n_1$.

\paragraph{Protocols that separate transcendental variables.}  We will rely in
our analysis on a property that many protocols, though not all,
satisfy.  We will exclude the remaining protocols from our tool.  The
property concerns how the protocol handles field values that are
\emph{not} used for exponentiation but transmitted in carried
position.  These field values contribute to messages as entries in
tuples and in the plaintexts of cryptographic operations.  We do not
constrain how field values feed into exponents at all.

A protocol $\Pi$ that satisfies the acquired constraint
\emph{separates transcendental variables} iff, whenever $\rho\in\Pi$,
$n=\strandnode\rho i$ is a transmission node on $\rho$, and
$p=(\msg(n),\pi)$ is a carried path in the term $\msg(n)$, if
$t={\msg(n)\termat\pi}$ is of sort field $t\colon\scfld$, then $t$ is
a variable of the sort transcendental, $t\in\kind{Var}$ and
$t\colon{\sctrsc}$.

Thus, $\Pi$ never sends out---in carried position---a term that has
the form of a product $xy$ or $3x$.  By contrast, we can use them for
exponentiation and send out $\gen^{xy}$ or $\gen^{3x}$ without any
problem.  

This restriction is mild for us.  We use field values in carried
position mainly to model the compromise of regular choices that were
previously secret, for instance to represent forward secrecy.  We
represent these choices as transcendentals, and therefore the
compromise event transmits one newly compromised values or, if
several, as a tuple but not algebraically combined.

This rules out some protocols.  For instance, the Schnorr signature
protocol transmits a term $k-xe$ where $k\colon\sctrsc$ is an
ephemeral random value; $x\colon\sctrsc$ is the signer's long-term
secret; and $e$ is a term representing the hash of the message to be
signed and $g^k$.  This is an anthology of things we do not represent.
It involves the additive structure of the field, which we will
subsequently assume the protocol definition ignores.  It transmits a
field value that is not simply one of the transcendentals $x,k$.  And,
moreover, it uses hashing to generate a field value.  Our hash
function $\hash m$ generates values of sort {\scmesg}, and our algebra
offers no way to coerce them to the field ${\thefield}$.  This last is
the simplest to remedy, and a version of {\cpsa} in the near future
will support hashing into the exponent.

We can however perfectly well view Schnorr signatures as a
cryptographic primitive $\enc m K$, where $K$ involves the signer's
long-term secret.  We just do not represent how it works.
Computational cryptography backs up the assumption that it in fact
acts as a signature.  

A concrete message is present in another one if it is at the end of
any path, whether carried or not, or if it is a transcendental with
non-zero degree in something at the end of a path:
\begin{definition}
  \label{def:present}
  A concrete message $m_0\in\Algebraconstr$ is \emph{present} in $m_1$
  iff there is a path $p=(m_1,\pi)$ such that either
  $m_0=m_1\termat\pi$, or else $m_0\in\trsc$ and, letting
  $c=m_1\termat\pi$, either
  \begin{enumerate}
    \item $c\in\thefield$ and $m_0$ has non-zero degree in $c$; or 
    \item $c=\gen^p\in\thegroup$, and $m_0$ has non-zero degree in
    $p$.
  \end{enumerate}
  A message $m\in\Algebraconstr$ is \emph{chosen} at node
  $n=\strandnode s i$ iff $m\ingredient\msg(n)$, and $n$ is a
  transmission $\direction(n)=+$, and $1\le j<i$ implies $m$ is not
  present in $\msg(\strandnode s j)$.
\end{definition}
We do not make any corresponding definition for formal messages
$t\in\term(\lang{\Pi})$, because an occurrence of a variable
$v\colon\sctrsc$ in a term $t\colon\scfld$ is not preserved under
interpretation or under substitution.  For instance, when $t$ is
$v\cdot w$, for a variable $w\colon\scfld$, then $\eta_\bnd(t)=\one$
when $\eta(w)=\one/\eta(v)$.  Similarly, $v$ cancels out under the
substitution $w\mapsto w_1/v$.  \emph{Chosen} is to ``present'' as
\emph{originates} is to ``carried.''  

Recall (Def.~\ref{def:paths}) that a message $t$ is \emph{visible} in
a message $m$ iff we can reach $t$ within $m$ by traversing only
tuples.
\begin{lemma}
  \label{lemma:simple:visible}
  Suppose $\Pi$ separates transcendentals, $\bnd$ is a $\Pi$-bundle, and
  $x\in\trsc$ is present in $p\in{\thefield}$.  If $p$ is visible in
  $\msg(n_p)$ for $n_p\in\node(\bnd)$, then $x$ is visible in
  $\msg(n_x)$ for some $n_x\preceq_{\bnd}n_p$.
\end{lemma}
\begin{proof}
  Choose $\bnd$, and if there are any $x,p,n_p$ that furnish a
  counterexample let $n_p\in\node(\bnd)$ be $\preceq_{\bnd}$-minimal
  among counterexamples for any $x,p$.  Observe first that $x\not=p$,
  since if $x=p$ this is not a counterexample to the property.
  
  Since $p$ is carried in $\msg(n_p)$, there exists an
  $n_o\preceq_{\bnd}n_p$ such that $p$ originates on $n_o$.  By the
  definition of originates, $n_o$ is a transmission node.

  First, we show that $n_o$ does not lie on an adversary strand, by
  taking cases on the adversary strands.  The \textbf{creation}
  strands that emit values in $\field$ originate $\one$, members of
  $\mathbb Q$, and transcendentals $y\colon\trsc$.  But $x$ is not
  present in $\one$ or values in $\mathbb Q$, and if $x$ is present in
  $y$, then $x$ and $y$ are identical, which we have excluded.

  If $n_o$ lies on a \textbf{algebraic} strand, then it takes incoming
  field values $p_1,p_2$.  Since $x$ has non-zero degree in $p$ only
  if it has non-zero degree in at least one of the $p_i$, this
  contradicts the $\preceq_{\bnd}$-minimality of the counterexample.

  Node $n_o$ does not lie on an \textbf{encryption},
  \textbf{decryption}, \textbf{tupling}, or \textbf{separation}
  strand, which never originate values in {\scsbasic}.

  Thus, $n_o$ does not lie on an adversary strand.  

  Finally, $n_o$ does not lie on a regular strand of $\Pi$.  By the
  DH-simple assumption, if $n_o$ originates the field value $p$, then $p$ is a
  transcendental.  Thus, if $x$ is present,  $p=x$, which was excluded
  above.  \qed
\end{proof}

\paragraph{Subsignature of $\Sigmaconstr$.}  {\cpsa} uses unification
systematically, and unification in the theory of fields and related
theories is undecidable.  In this work, our approach is to omit the
additive structure $\zero,+,-$.  This restricts the class of protocols
to the fairly large set in which the regular participants do not use
it, but remains faithful:  The adversary does not need to use the
additive structure either, to achieve all possible attacks against
these protocols~\cite{LiskovThayer14}.  But see
also~\cite{DoughertyGuttman12,DoughertyGuttman2014} for alternatives
that avoid unification.

For the remainder of this paper, we will focus on the
``multiplicative-only'' signature
$\Sigmamult=\Sigmaconstr\setminus\{\zero,+,-\}$ for messages.  

We continue to use the message algebra $\Algebraconstr$, containing
the field $\mathbb{Q}(\trsc)$, but we will consider only protocols
that do not mention its addition operations, nor the group operation,
which amounts to adding exponents.
%
%
By~\cite{LiskovThayer14}, we may assume that the adversary never uses
the strands for field addition and subtraction, and for the group
operation (addition in the exponent).

Thus, we will interpret these protocols in structures whose domains
for message sorts are those of $\Algebraconstr$, but whose signature
has ``forgotten'' the additive structure.  Thus, from now on, all of
the polynomials we consider will in fact be monomials.  We will write
$\mu,\nu$, etc.~for monomials in $\mathbb{Q}(\trsc)$, namely
polynomials with a rational coefficient but no additions in a number
of transcendentals $x_1,\ldots x_i\in\trsc$.  The transcendentals may
occur with positive or negative degree.

By a result of Liskov and Thayer~\cite{LiskovThayer14}, when the
protocol uses only the multiplicative structure, the adversary does
not need the additive strands.  That is, every attack that can be
achieved using all of the adversary strands of
Fig.~\ref{fig:adversary:strands} can be achieved without the additive
strands.   Therefore, we do not weaken the adversary if:
\begin{Assum}
  \label{Assum:monomial} 
  We henceforth assume that no bundle $\bnd$ contains any occurrences
  of the additive adversary strands.  
\end{Assum}

\paragraph{Protocols that separate exponents:  Group elements.}
Lemma~\ref{lemma:simple:visible} tells us what must hold if a
polynomial $p$ is disclosed, assuming the protocol separates
transcendentals.  Namely, any transcendental $x\colon\trsc$ with non-0
degree in $p$ is also disclosed.  We now provide a related property
for the group elements.  It holds for protocols that separate
transcendentals and use only the multiplicative signature
$\Sigmamult$.  Thus, the only polynomials of interest are monomials
$\mu$.  

The lemma says that when a group element $\gen^\mu\in\thegroup$ is
disclosed, then either \emph{all} of the transcendentals
$x\colon\trsc$ with non-0 degree in monomial $\mu$ are also disclosed,
or else we can divide the monomial $\mu$ into two parts $\nu$ and
$\mu/\nu$.  All the transcendentals in $\nu$ are disclosed.  And some
regular participant tansmitted $\gen^{\mu/\nu}$ in carried position.
Thus, responsibility for $\nu$ lies with the adversary, and
responsibility for sending an exponentiated version of the quotient
lies with some instance of a role of the protocol.  
\begin{lemma}
  \label{lemma:group:visible}
  Suppose $\Pi$ is a protocol over $\Sigmamult$ that separates
  transcendentals and uses only multiplicative structure; $\bnd$ is a
  $\Pi$-bundle; and $\gen^\mu\in\thegroup$ is carried in $\msg(n_\mu)$
  where node $n_\mu\in\node(\bnd)$.
  Then there is a monomial $\nu\in\thefield$ s.t.:
  \begin{enumerate}
    \item $\nu$ is a product of transcendentals visible before
    $n_\mu$, and 
    \item either (i) $\nu=\mu$ or else \\
    (ii) letting $\xi=\mu/\nu$, there is a regular transmission node
    $n_\xi\in\node(\bnd)$ such that $n_\xi\preceq_\bnd n_\mu$ and the
    value $\gen^{\xi}$ is carried in $\msg(n_\xi)$.  Moreover
    $\gen^{\xi}$ was previously visible.  
  \end {enumerate}
%
\end{lemma}
%
%
\begin{proof}
  Let $\bnd$ be a bundle, let $n_\mu\in\node(\bnd)$, and assume
  inductively that the claim holds for all nodes $n\prec n_\mu$.  If
  $n_\mu$ is a reception node, then the (earlier) paired transmission
  node satisfies the property by the IH.  However, the same $\nu$ and
  $n_\xi$ also satisfy the property for $n_\mu$.  If $n_\mu$ is a
  regular transmission, then the conclusion holds with $\nu=\one$, the
  empty product of transcendentals.

  So suppose $n_\mu$ lies on an adversary strand.  If $n_\mu$
  transmits the group element $\gen$, then let $\nu=\mu=\one$.  The
  constructive strands for tupling or encryption provide no new group
  elements in carried position.  Nor do the destructive strands for
  untupling or decryption.

  Thus, the remaining possibility is that $n_\mu$ is the transmission
  on an exponentiation strand
  $-h\Rightarrow -w\Rightarrow +\expt(h,w)$ where $\expt(h,w)=g^\mu$.
  By the IH, for the node receiving $h\in\thegroup$, the property is
  met.  Thus, $h=\gen^{\mu_0}$, where there exist $\nu_0,\xi_0$
  satisfying the conditions.

  Hence, we may take $\nu=\nu_0 w$ and $\xi=\xi_0$.  By
  Lemma~\ref{lemma:simple:visible}, $w$ is a product of previously
  visible transcendentals, so the requirements are met.  \qed
\end{proof}
\noindent
It also follows that if $g^\mu$ is \emph{visible} at $n_\mu$, then
$g^\xi$ is visible before $n_\mu$.
We will wrap up our three protocol requirements in the word
``compliant.''  
\begin{definition}
  A protocol $\Pi$ is \emph{compliant} if $\Pi$ separates
  transcendentals, uses only multiplicative structure, and satisfies
  the acquired constraint.  
\end{definition}


%% file: skeletons.tex

In this section, we will introduce a language $\lang\Pi$ for each
protocol $\Pi$.  The \emph{skeletons} for $\Pi$ are certain theories
$\skel$ in that language.  We regard a skeleton as describing certain
executions; in our formalization the executions are bundles, and we
will use the usual semantic relation of satisfaction to define which
bundles $\bnd$ a skeleton $\skel$ describes.  Specifically, let
$\mathcal{I}$ be an interpretation from $\lang\Pi$ into $\bnd$;
$\skel$ describes $\bnd$ under this interpretation if
$\mathcal{I}\models\skel$.

Some skeletons ``fully describe'' some bundle $\bnd$ under an
interpretation $\mathcal{I}$, which we define below to mean that
$\mathcal{I}$ is an injective map and is also surjective for nodes
(Def.~\ref{def:cover:realize}).  

\subsection{The protocol languages}
\label{sec:skeletons:lang}

The logical langauge $\lang\Pi$ for the protocol $\Pi$ can be used to
express its protocol goals, but extends the goal language of Guttman,
Rowe, et al.~\cite{Guttman14,RoweEtAl2016}, as it also describes the
behavior in individual executions in more detail.  It uses a first
order signature $\Sigma^*$ that extends $\Sigmaconstr$ (still omitting
the additive operators) with:
\begin{description}
  \item[Sort] {\scnode} for strand nodes;
  \item[Individual constants $\mathcal{C}$,] disjoint from
  $\Sigmaconstr$, in infinite supply at every sort;
  \item[Protocol-independent] vocabulary, which is the same for all
  protocols $\Pi$; and
  \item[Protocol-dependent] vocabulary, which gives a way to refer to
  the specific kinds of nodes on the roles of $\Pi$, and the values
  their parameters take.
\end{description}
We provide more detail about the last two categories below.

If $\Sigma^*$ is any structure that extends $\Sigmaconstr$, we will
define a $\Algebraconstr$-\emph{interpretation} $\mathcal{I}$ to be a
$\Sigma^*$-structure whose restriction to $\Sigmaconstr$ agrees with
the map $\Algebraconstr$.  Since $\Sigma^*$ may involve entirely new
sorts (such as {\scnode}), the target of $\mathcal{I}$ may involve
entirely new domains.

\paragraph{Protocol-independent vocabulary} $\lang{}$ includes
equality at all sorts and:
{\small
\begin{eqnarray*}
  & \mathtt{msg}\ccolon {\scnode}\rightarrow {\scmesg}
  & \\
  \mathtt{Prec}\cccolon{\scnode}\times{\scnode} \quad
  & \mathtt{Coll}\cccolon{\scnode}\times{\scnode}
  & \mathtt{DerBy}\cccolon{\scnode}\times{\scmesg} \\
  \mathtt{Absent}\cccolon{\sctrsc}\times{\scfld} \quad
  &
    \mathtt{GenAt}\cccolon{\scnode}\times{\scnode}\times{\scbasic}
  & \quad \mathtt{Non}\cccolon{\scnode}\times{\scmesg}
\end{eqnarray*}
}
We will use typewriter font for function and relation symbols in
$\lang\Pi$.

The function symbol $\mathtt{msg}$ returns the message transmitted or
received on a node.  The relation symbols
$\mathtt{Prec}, \mathtt{Coll}$ are satisfied by a pair of nodes
(resp.)~iff the first precedes the second and iff both lie on the same
strand.  We sometimes write $\mathtt{StrPrec}(v_0,v_1)$ to mean
$\mathtt{Prec}(v_0,v_1)\land\mathtt{Coll}(v_0,v_1)$.

$\mathtt{DerBy}$ says that its message argument is derivable from
messages visible on nodes preceding its node argument.
$\mathtt{Absent}(v,w)$ says that the transcendental $v$ has degree
zero in the field element $w$.  $\mathtt{Non}(v)$ says that $v$
originates nowhere, and $\mathtt{GenAt}(n_0,n_1,v)$ says that $v$ is
chosen only on nodes $n_0,n_1$ (Def.~\ref{def:paths}).

We write $\mathtt{UnqGen}(n,v)$ to mean $\mathtt{GenAt}(n,n,v)$, where
$v$ can only be chosen at $n$.  {\cpsa} currently implements
$\mathtt{UnqGen}(n,v)$ rather than the full
$\mathtt{GenAt}(n_0,n_1,v)$.  The latter would be convenient for
expressing compromise assumptions (e.g.~in forward secrecy assertions)
in a compact, uniform way.

\begin{definition}
  \label{def:prot:ind:semantics}
  Let $\bnd$ be any bundle, and let $\mathcal{I}$ be any
  $\Algebraconstr$-interpretation for $\lang{}$ such that
  $\scnode(\mathcal{I})\subseteq\node(\bnd)$, i.e.~the sort $\scnode$
  is interpreted by nodes in $\bnd$.

  The clauses in the upper block of Fig.~\ref{fig:prot:semantics} give
  the semantics of the protocol independent vocabulary of $\lang{}$.
  We write syntactic text within $\lang{\Pi}$ in blue to distinguish
  it from our informal metatheory for properties of bundles.
\end{definition}
%
%
\begin{figure}[tb]
  \centering
  \begin{tabular}{l@{\qquad}l}
    Term & Semantics \\
    $\eta_{\mathcal{I}}(\syntax{\mathtt{msg}(t)})$
         & $\msg(\eta_{\mathcal{I}}(\syntax{t}))$
    \\[2mm]
    Formula & Semantics \\
    $\mathcal{I}\models_\eta\syntax{\mathtt{Prec}(t_1,t_2)}$
         & iff
           $\eta_{\mathcal{I}}(\syntax{t_1})\prec_\bnd\eta_{\mathcal{I}}(\syntax{t_2})$
    \\
    $\mathcal{I}\models_\eta\syntax{\mathtt{Coll}(t_1,t_2)}$
         & iff there exist $s,i,j$ such that 
           $\eta_{\mathcal{I}}(\syntax{t_1})=\strandnode s i$, \\
         & \quad  
           $\eta_{\mathcal{I}}(\syntax{t_2})=\strandnode s j$, and
           $(\strandnode s i),(\strandnode s j)\in\node(\bnd)$ \\
    $\mathcal{I}\models_\eta\syntax{\mathtt{Absent}(t_1,t_2)}$
         & iff
           $\eta_{\mathcal{I}}(\syntax{t_1})$ is not present in $\eta_{\mathcal{I}}(\syntax{t_2})$
    \\
    $\mathcal{I}\models_\eta\syntax{\mathtt{Non}(t_1)}$
         & iff $\eta_{\mathcal{I}}(\syntax{t_1})$ does not originate
           at any $n\in\node(\bnd)$ 
    \\
    $\mathcal{I}\models_\eta\syntax{\mathtt{GenAt}(t_1,t_2,t_3)}$
         & iff for all $n\in\node(\bnd)$, 
           if $\eta_{\mathcal{I}}(\syntax{t_3})$  is chosen at $n$, \\
         & \quad then
           $n=\eta_{\mathcal{I}}(\syntax{t_1})$ or $n=\eta_{\mathcal{I}}(\syntax{t_2})$
    \\
    $\mathcal{I}\models_\eta\syntax{\mathtt{DerBy}(t_1,t_2)}$
         & iff there is an adversary web with root
           $\eta_{\mathcal{I}}(\syntax{t_2})$   \\
         & \quad in which, for every leaf $\ell$,  \\
         & \quad
           either $\ell\in\node(\bnd)$ and $\ell\preceq_\bnd\eta_{\mathcal{I}}(\syntax{t_1})$, or else \\
         & \quad
           $\direction(\ell)=-$ and there is an
           $n\in\node(\bnd)$ s.t. \\ 
         & \quad
           $\direction(n)=+$, $\msg(\ell)=\msg(n)$, and
           $n\preceq_\bnd\eta_{\mathcal{I}}(\syntax{t_1})$
    \\ ~ & ~
    \\
    $\mathcal{I}\models_\eta\syntax{\mathtt{\RPPred\rho i}(t_1)}$
         & iff
           $\eta_{\mathcal{I}}(\syntax{t_1})$ is an instance of $\strandnode\rho i$
           under some $\theta$
    \\
    $\mathcal{I}\models_\eta\syntax{\mathtt{\PPred\rho v}(t_1,t_2)}$
         & iff, for some $\rho\in\Pi, i$
           s.t.~$v\in\kind{params}(\strandnode\rho i)$, \\
         & \quad
           for some assignment
           $\theta\colon\kind{Var}(\lang{\Pi})\rightarrow\bnd$,  \\
         & \quad
           $\eta_{\mathcal{I}}(\syntax{t_1})$ is an instance of
           $\strandnode\rho i$ under $\theta$, \\ 
         & \quad
           and $\theta(v)=\eta_{\mathcal{I}}(\syntax{t_2})$
  \end{tabular}
  %
%

  \caption{Semantics of protocol independent and dependent vocabulary}
  \label{fig:prot:semantics}
\end{figure}

By an \emph{adversary web} in the clause for
$\mathtt{DerBy}(t_1,t_2)$, we mean a directed acyclic graph built from
adversary strands (Def.~\ref{def:adversary:web}).  Its \emph{leaves}
are any earliest reception nodes, which are required to draw values
from the bundle $\bnd$.  Its \emph{roots} are all latest nodes; the
messages on these nodes are the values derived.

\paragraph{Protocol-dependent vocabulary of } $\lang\Pi$ consists of
\emph{role position} predicates and \emph{parameter} predicates.  A
role position predicate is a one-place predicate of nodes.  There is
one role position predicate for each pair $\rho,i$ such that
$\rho\in\Pi$ and $1\le i\le\length\rho$, and it is true of a node $n$
iff that node is an instance of $\strandnode\rho i$.

Any family of distinct predicate symbols may be chosen for the
\emph{role position} predicates.  We will refer to them as
$\RPPred\rho i$; i.e.~$\RPPred\cdot \cdot$ is a two dimensional table
with rows labeled by roles of $\Pi$ and columns indexed by the
positions along them.  The entries in the $\RPPred\cdot \cdot$ table
are the role position predicates.

A parameter predicate is a two-place predicate relating a node to a
message value.  There is one role position predicate for each pair
$\rho,v$ such that $\rho\in\Pi$ and $v\in\kind{params}(\rho)$.  It is
true of a node $n$ and a message value $t$ iff, for some $i,\eta$,
$v\in\kind{fv}(\msg(\strandnode\rho i))$ and $n$ is an instance of
$\strandnode\rho i$ under $\eta$ where $\eta(v)=t$.  Thus, $n$ is not
related to any message value $t$ if it is an instance of
$\strandnode\rho i$, but $v$ is a parameter that first occurs in
$\strandnode\rho j$ for $j>i$.  There is, however, at most one
possible $v$; i.e.~the parameter value is a partial function of the
node.  

The predicate symbols for \emph{parameter predicates} need not be
entirely distinct.  For instance, both an initiator role and a
responder role may have $\mathtt{Peer}$ and $\mathtt{Self}$ parameter
predicates.  Choices like this are desirable when the parameters have
corresponding significance for different roles.  If $n_i$ and $n_r$
lie on instances of the initiator and responder roles, and are high
enough that the parameters are defined, then these are \emph{matching
  sessions} only if $\mathtt{Self}(n_i,A)\equiv\mathtt{Peer}(n_r,A)$
and $\mathtt{Peer}(n_i,B)\equiv\mathtt{Self}(n_r,B)$, etc.  Reusing
parameter predicates in different roles is safe when the roles have no
instances in common, for instance when one begins with a transmission
and the other begins with a reception.

We will refer to the parameter predicates as $\PPred\rho v$;
i.e.~$\PPred\cdot \cdot$ is a two dimensional table with rows labeled
by roles of $\Pi$ and columns indexed by their parameters.  The
entries in the $\PPred\cdot \cdot$ table are the parameter predicates.

The (straightforward) semantics for the protocol-dependent vocabulary
of $\lang\Pi$ is given in the lower block of
Fig.~\ref{fig:prot:semantics}.  In this case, we are only interested
in the case in which $\bnd$ is in fact a $\Pi$-bundle.

\paragraph{Some relevant axioms.}  There are a number of axioms
expressed in terms of these predicates that are satisfied in all
$\Pi$-bundles.  
We gather them in Fig.~\ref{fig:skelax}.  
\begin{figure}[tb]\small
  \centering
  \begin{description}
    \item[Prec]  
    is transitive and anti-reflexive, i.e.~a strict order.
    \item[Coll] 
    is reflexive, symmetric, and transitive, i.e.~an equivalence
    relation.
    %
    %
    %
    \item[Precedence for $\nf{\cdot}{\cdot}$:]  For each $\rho\in\Pi$
    and $j$ where $1\le j<\length{\rho}$:
    \[
      \forall n\colon\scnode\qdot \nf{\rho}{j+1}(n)\limp\exists
      m\qdot\nf{\rho}{j}(m)\land \mt{StrPrec}(m, n) .
    \]
    \item[Strand uniqueness of $\nf{\cdot}{\cdot}$:]  Only one node on a
    strand satisfies any role position predicate.  That is, for each
    $\rho\in\Pi$ and $j$ where $1\le j<\length{\rho}$:
    %
    \[
      \forall m,n\colon\scnode\qdot \nf{\rho}{j}(n)\land
      \nf{\rho}{j}(m)\land\mt{Coll}(m,n) \limp m=n .
    \]
    \item[Existence for $\pf{\cdot}{\cdot}$:]  For each $\rho\in\Pi$,
    $j$, and $v$ such that $v\in\kind{params}(\strandnode\rho j)$,
    \[ 
      \forall n\colon\scnode\qdot \nf{\rho}{j}(n) \limp \exists v_1\qdot
      \pf{\rho}{v}(n,v_1) .
    \]
    \item[Preservation for $\pf{\cdot}{\cdot}$:]  For each $\rho\in\Pi$
    and each $v\in\kind{params}(\rho)$:
    \[ \forall m,n\colon\scnode,v_1\colon\scmesg\qdot
      \mt{StrPrec}(m,n)\land \pf{\rho}{v}(m,v_1)\limp
      \pf{\rho}{v}(n,v_1). 
    \] 
    \item[Uniqueness for $\pf{\cdot}{\cdot}$:]  For each $\rho\in\Pi$
    and each $v\in\kind{params}(\rho)$:
    \[ 
      \forall n\colon\scnode,v_1,v_2\colon\scmesg\qdot
      \pf{\rho}{v}(n,v_1)\land\pf{\rho}{v}(n,v_2) \limp v_1=v_2 .
    \]
  \end{description}
  \caption{$\skelax$ axiom groups}
  \label{fig:skelax}
\end{figure}
We will refer to these axioms as the \emph{skeleton axioms}
$\skelax$.  

%
%
%
%
%
%
%
%

A theory $T$ is a set of sentences $T\subseteq\sent(\lang\Pi)$.  We
write $T\entails\phi$ when $\phi$ is a logical consequence of the
sentences in $T$, avoiding the standard symbol $\vdash$ which is also
used for adversary derivability as in Fig.~\ref{fig:rules}.  Let
$\alpha[\vec t/\vec v]$ be the result of replacing the variables in
$\vec v$ by the corresponding terms $\vec t$ throughout~$\alpha$.

\begin{definition}
  \label{def:expanded}
  A theory $T\subseteq\sent(\lang{\Pi})$ is \emph{in expanded form}
  iff, for any tuple of constants $\vec c$ and any axiom
  $\forall \vec v\qdot \alpha\limp\exists\vec w\qdot\beta$ in
  $\skelax$, where $\vec w$ may be empty:


  If $T\entails\alpha[\vec c]$, then there are constants $\vec d$ such
  that $T\entails\beta[\vec c/\vec v;\vec d/\vec w]$.
\end{definition}
Expanding a theory may make it larger, but not too much.  We use
$\length\cdot$ to express the cardinality $\length S$ of a set as well
as the length $\length s$ of a strand.  If $T$ is a set of formulas,
we write $\mathcal{C}(T)$ for the set of constants occurring in $T$.

\begin{lemma}%
  \label{lemma:expanded} Assume $\Pi$ is a protocol whose roles are of
  length less than $\ell$:  $\rho\in\Pi$ implies $\length\rho<\ell$.
  Suppose that each role $\rho\in\Pi$ has less than $j$ parameters:
  $\length{\kind{params}(\rho)}<j$.  And suppose that the assumptions
  of nodes of $\Pi$ are atomic formulas, and each node has less than
  $k$ assumptions:  $\rho\in\Pi$ implies
  $\length{\assume(\strandnode\rho i)}<k$.

  If $T$ is any theory in $\lang\Pi$, let
  $N(T)\subseteq\mathcal{C}(T)$ be the set of constants of sort
  $\scnode$ occurring in $T$.
  
  Suppose that $T\subseteq\sent(\lang\Pi)$ is a finite theory
  consisting of atomic sentences that may contain constants in
  $\mathcal{C}$.%
  %
  Then there exists a theory $T_e\supseteq T$ in expanded form such
  that, letting $b=\ell \length{N(T)}$:
  \begin{enumerate}
    \item $\length{N(T_e)}\le b$;
    \item $\length{T_e}\le (2j+k)b+3b^2+\length{T}$;
    \item $T_e\cup\skelax$ is conservative over $T\cup\skelax$ for
    sentences in $\mathcal{C}(T)$.

    I.e., suppose that $\phi\in\sent(\lang\Pi)$ and
    $\mathcal{C}(\phi)\subseteq\mathcal{C}(T)$.  Then if
    $T_e\cup\skelax\entails\phi$, then $T\cup\skelax\entails\phi$.
  \end{enumerate}  
\end{lemma}
From this it follows that it is decidable whether an atomic sentence
in the constants $\mathcal{C}(T)$ is a consequence of $T\cup\skelax$.
\begin{proof}
  Repeatedly instantiate the existential quantifiers of role
  predecessor axioms with new node constants, whenever they are not
  yet met in the theory.  Because $\ell$ is a length bound on roles,
  we can add at most $\ell-1$ new node constants for each node
  constant appearing in a role position sentence in $T$.  There are at
  most $\length{T}$ node constants occurring in these sentences, since
  each role position predicate is unary.

  For each one of these, we now obtain $T_e$ by saturating the theory
  by adding the corresponding assumptions.  Each one of them can
  contribute at most $j$ parameter predicate assertions, and at most
  $k$ assumptions.  For each node and each parameter assertion about
  it, there may be an equation introduced by a strand uniqueness
  axiom.  

  There are less than $b^2$ order assertions, and no more than $b^2$
  of collinearity.  There can be no more than $b^2$ equations
  introduced by strand uniqueness axioms.  There could be $\length{T}$
  other assertions inherited from $T$.

  $T_e$ is conservative:  Suppose any sentence $\phi\in T_e$ is in the
  minimal such $T_e$, and contains the new constants $\vec c$.
  Selecting new variables $\vec v$, the existential closure
  $\exists \vec v\qdot\phi[\vec v/\vec c]$ is already a consequence of
  $T\cup\skelax$.  This holds invariantly, as it is preserved when we
  introduce a new constant to witness for an exitential conclusion,
  and also under deduction.  If $\phi$ contains no new constants, then
  the existential closure is vacuous, and
  $\exists \vec v\qdot\phi[\vec v/\vec c]$ is identical to $\phi$.
  \qed
\end{proof}
\subsection{Skeletons}
\label{sec:skeletons:skeletons}

For the remainder of this section, consider a fixed protocol $\Pi$.
Thus, by a \emph{role}, we will mean a role $\rho\in\Pi$; by a
\emph{bundle}, we will mean a $\Pi$-bundle; by $\lang{}$, we will mean
the language $\lang{\Pi}$ of $\Pi$; etc.

In the order-sorted context, we define a role-specific theory
(cp.~\cite[Def.~5.8]{Guttman14}) as follows.  A skeleton is a
\emph{theory} which is role-specific and in expanded form; the latter
means that it is closed under the skeleton axioms
(Def.~\ref{def:expanded}).
\begin{definition}
  \begin{enumerate}
    \item A set $T\subseteq\sent(\lang{\Pi})$ of (well-sorted)
    sentences is a \emph{role-specific} theory iff, for every
    individual constant $c\colon{\scnode}$, if $c\in\mathcal{C}(T)$,
    then for some atomic formula $\RPPred{\rho}{i}(c)\in T$, $c$ is
    the argument of a role position predicate $\RPPred{\rho}{i}(c)$.
    %
    \item A \emph{skeleton} $\skel$ is a role-specific theory
    $\skel\subseteq\sent(\lang{\Pi})$ in expanded form.

    \item Let $H$ be a sort-respecting map from constants
    $c\in\mathcal{C}$ to $\term(\lang\Pi)$; extend $H$ from constants
    to terms and formulas homomorphically.

    $H$ is a \emph{skeleton homomorphism}
    $H\colon\skel\rightarrow\skelB$ iff it is a theory interpretation
    for $\Pi$-bundles, i.e.~for all $\Pi$-bundles $\bnd$ and
    interpretations $\mathcal{I}$ into $\bnd$, let $\mathcal{J}$ be
    the interpretation that sends $c$ to the value to which
    $\mathcal{I}$ sends $H(c)$.  Then $\mathcal{I}\models\skelB$
    implies $\mathcal{J}\models H(\skel)$.
  \end{enumerate}
\end{definition}
By a theory interpretation we mean the semantic notion, i.e.~it may
restrict but does not add models.  We use this rather than a notion
defined in terms of deduction because we have not fully axiomatized
the set of $\Pi$-bundles here.  We will avoid doing so, because the
notions of origination and chosen values (Defs.~\ref{def:paths},
\ref{def:present}) require a somewhat cumbersome inductive
definition.

Moreover, we follow the {\cpsa} tradition of using skeletons as data
structures that we inspect and build in computations, rather than as
explicit theories in which we do deduction.  {\cpsa}'s skeletons
maintain some consequences of unique and non-origination assumptions,
which are not captured in $\skelax$.  Apart from this, however, they
maintain the same information that is expressed in role-specific
theories in expanded form; essentially, this is the content of
Thms.~4.13--4.15 of~\cite{Guttman14}.  Indeed, a role-specific theory
in expanded form determines formal messages sent and received, so that
we can use matching and unification on these terms in just the way
that {\cpsa} previously solved its
problems~\cite{Guttman12,cpsatheory11}.

In particular, {\cpsa} looks for reception nodes $n$ in a skeleton
$\skel$ that receive formal messages that the adversary cannot supply
given previous transmissions in $\skel$.  Each such ``non-derivable''
node indicates that the skeleton $\skel$ is incomplete, i.e.~it may
describe some bundles, but it is not a full description of any bundle.
\begin{definition}
  \label{def:cover:realize}
  Let $\bnd$ be a $\Pi$-bundle,
  $\mathcal{I}\colon\mathcal{C}\rightarrow\bnd$, and $\skel$ be a
  skeleton.
  \begin{enumerate}   
    \item $\skel$ $\mathcal{I}$-\emph{covers} $\bnd$ iff $\mathcal{I}$
    is an interpretation into $\bnd$ and $\mathcal{I}\models\skel$.
    \item $\bnd$ $\mathcal{I}$-\emph{realizes} $\skel$ iff $\skel$
    $\mathcal{I}$-covers $\bnd$ and:  
    \begin{enumerate}
      \item \label{clause:realize:alg:ind} for all terms
      $t_1,t_2\in\term(\lang\Pi)$, if
      $\mathcal{I}(t_1)=\mathcal{I}(t_2)$ then
      $\skel\entails t_1=t_2$;
      \item $\mathcal{I}$ is surjective on regular nodes; i.e., for
      every regular $n\in\node(\bnd)$, there is a constant
      $c\colon\scnode$ in $\mathcal{C}(\skel)$ such that
      $\mathcal{I}(c)=n$.
    \end{enumerate}
    \item A formula $\Psi$ is $\Pi$-\emph{entailed} by $\skel$ iff,
    whenever $\skel$ $\mathcal{I}$-covers $\bnd$, then
    $\mathcal{I}\models\Psi$.
    \item If $\skel$ is a skeleton, then the \emph{nodes of} $\skel$,
    written $\node(\skel)$, is the set of constants $c\colon{\scnode}$
    in $\mathcal{C}(\skel)$.

    Node $c_0$ \emph{precedes} node $c_1$ in $\skel$ iff
    $\skel\entails\mathtt{Prec}(c_0,c_1)$.
  \end{enumerate}
  $\skel$ \emph{covers} $\bnd$ iff, for some $\mathcal{I}$, $\skel$
  $\mathcal{I}$-covers $\bnd$.  $\bnd$ \emph{realizes} $\skel$ iff,
  for some $\mathcal{I}$, $\skel$ $\mathcal{I}$-realizes $\bnd$.
  $\skel$ is \emph{realized} iff some $\bnd$ realizes it.  
\end{definition}
When $\bnd$ realizes $\skel$, $\skel$ describes all of the regular
events in $\bnd$.  Moreover, $\skel$ is no more generic then $\bnd$,
in the sense that it reflects all of the equations that $\bnd$ forces
to be true.  Indeed, when $\bnd$ realizes $\skel$, there may be ``more
specific'' bundles $\bnd'$ that force additional equations to be true.
However, when $\skel$ is realized, it is already specific enough to
explain what happens in some possible execution, factoring out the
details of how the adversary chooses to derive messages that will be
derivable.  

\paragraph{Derivable nodes.}  The adversary, when acting on the formal
messages of a skeleton $\skel$, can use instances of the adversary
strands of Fig.~\ref{fig:adversary:strands}, regarded as strands over
formal terms.  Moreover, the adversary can feed a term produced by a
transmission node to a reception node when the two nodes agree on the
term in question.  In establishing this equality, we can use the
axioms for the field and group operations as well as the axioms of
Def.~\ref{def:alg:basic}.

Since the division axiom has the premise that the divisor is non-zero,
we also need a source of assertions of that form.  Since every
transcendental is different from zero, we always have the formula
$v\not=\zero$ for variables $v\colon\sctrsc$.  By the field axioms,
$t_1\not=\zero$ and $t_2\not=\zero$ iff $t_1\cdot t_2\not=\zero$.

Since we are ignoring the additive structure, each $t\colon\scfld$ is
a quotient of monomials.  These are easy to simplify:  One simply
cancels common factors, resulting in a monomial $\mu$ with a rational
coefficient and a sequence of distinct variables $x\colon\sctrsc$ each
raised to a positive or negative power.  We will say that
$x\colon\sctrsc$ occurs in $t$ if it occurs in this reduced form
$\mu$.

Using these methods for proving equalities, we can build up
\emph{formal adversary webs} (cp.~Def.~\ref{def:adversary:web}).
These are directed acyclic graphs in which the nodes lie on adversary
strands, and in which the communication edge $+t_1\rightarrow -t_2$ is
permitted if there is a proof that $t_1=t_2$.

In formal adversary webs, we allow creation nodes in which we create a
constant $c\colon\scfld$.  These are unproblematic, because we can
always interpret them by sending $c$ to an otherwise unused
transcendental, or indeed to $\one$.

Likewise, adapting Def.~\ref{def:present}, we can speak of a node with
message in $\term(\lang\Pi)$ choosing a subterm.  Let us say that a
term $t_0$ \emph{occurs in} a term $t_1$ if, having simplified
monomials in $t_0,t_1$ as just described, $t_0$ occurs syntactically
in $t_1$.  For instance $x\colon\sctrsc$ occurs in $xy/y$ and $x^2y/x$
but not in $xy/x$.

Now we will say that a term $t_0$ is \emph{chosen} at node
$n=\strandnode s i$ iff $t_0$ occurs in $\msg(n)$, and $n$ is a
transmission $\direction(n)=+$, and $1\le j<i$ implies $t_0$ does not
occur in $\msg(\strandnode s j)$.

For a skeleton $\skel$, an earliest creation node $+c$ is permitted if
$c$ is \emph{unrestricted} in $\skel$.  A constant $c$ is restricted
if $\mathtt{Non}(c)\in\skel$---so that $c$ must be non-originating in
any bundle to be covered---or if
$\mathtt{GenAt}(t_1,t_2,c)\in\skel$---so that $c$ can be chosen only
on the regular nodes interpreting $t_1,t_2$.  If there is no such
formula in $\skel$, then $c$ can be freely created by the adversary.
\begin{definition}
  \label{def:derivable}
  Let $\skel$ be a skeleton.  A message $t\in\term(\lang\Pi)$ is
  \emph{derivable by} node $n_1\in\node(\skel)$ iff there is a formal
  adversary web with root $t$ such that (i) every creation node has
  unrestricted value, and (ii) every leaf receives its message from a
  transmission node $n_0\in\node(\skel)$ where $n_0$ precedes $n_1$ in
  $\skel$.

  $\skel$ is \emph{derivable} iff:
  \begin{enumerate}
    \item \label{clause:der:webs} for every reception
    $n_1\in\node(\skel)$ in $\skel$, $\msg(n)$ is derivable by $n$ in
    $\skel$;
    \item if $\skel\entails\mathtt{DerBy}(n,t)$, then $t$ is derivable
    by $n$ in $\skel$;
    \item if $\skel\entails\mathtt{Absent}(x,t)$, then $x$ does not
    occur in $t$;
    \item if $\skel\entails\mathtt{GenAt}(n_1,n_2,t)$, then $t$ is
    chosen on $n$ iff $n=n_1$ or $n=n_2$;
    \item if $\skel\entails\mathtt{Non}(t)$, then $t$ originates on no
    node $n\in\node(\skel)$.
  \end{enumerate}
  We say that $n$ is \emph{derivable in} $\skel$ iff $\msg(n)$ and all
  terms $t$ such that $\skel\entails\mathtt{DerBy}(n,t)$ with this $n$
  are derivable by $n$ in $\skel$.  
\end{definition}
Thus, a skeleton is derivable iff all of its nodes are derivable, and
the $\mathtt{Absent}$, $\mathtt{GenAt}$, and $\mathtt{Non}$
constraints are met.  
\begin{lemma}
  \label{lemma:derivable:realized}
  If $\skel$ is derivable, then there is a bundle $\bnd$ such that
  $\bnd$ realizes $\skel$.  
\end{lemma}
\begin{proof}
  Choose distinct values of corresponding sort in $\Algebraconstr$ for
  all message constants in $\mathcal{C}(\skel)$, using transcendentals
  for constants of sort $\scfld$ as well as of sort $\sctrsc$.  Choose
  distinct nodes for each node constant $c_n$, with their directions
  and messages determined by the role definitions determined by their
  node position predicates and parameter predicates.  These nodes are
  the regular nodes of $\bnd$.  These choices determine an
  interpretation map $\mathcal{I}$.

  By clause~\ref{clause:der:webs}, there is a formal web rooted at
  each $c_n$ that derives it formally from earlier transmissions.
  Applying $\mathcal{I}$ turns this into an advesary web over
  $\Algebraconstr$.  Thus, the nodes form a bundle when equipped with
  these adversary webs.  The remaining annotations are satisfied by
  the corresponding clauses of Def.~\ref{def:derivable}.  \qed
\end{proof}
Curiously, the converse of this lemma is false.

Consider a skeleton $\skel$ containing one regular strand, consisting
of a node that sends $+\enc{g^x}K$, after which $-\gen^{xw}$ is
received.  Here, $x\colon\sctrsc$ is assumed uniquely generated, and
the symmetric key is non-originating.  Since the variable
$w\colon\scfld$ is of the broader field sort, we may apply a
substitution $\sigma$ that sends $w\mapsto z/x$, meaning that, in
$\sigma(\skel)$ the adversary must supply the value
$\gen^{x(z/x)}=\gen^z$.  Since $z$ is unconstrained, this the
adversary can easily supply, with a creation followed by an
exponentiation.

The substitution $\tau$ in which $z\mapsto wx$ inverts $\sigma$,
i.e.~$\tau(\sigma(\skel))$ is the same theory as $\skel$.  Thus, the
two skeletons $\skel$ and $\sigma(\skel)$ are inter-interpretable.
Since $\sigma(\skel)$ is derivable---hence also realized---$\skel$ is
also realized.

{\cpsa} rewrites skeletons when possible to make receptions derivable,
after which we use derivability as a criterion of being realized.  The
rewriting must meet two constraints.  First, $x\colon\sctrsc$ must
occur in $\msg(n)$ if it is assumed to be chosen there in
$\mathtt{GenAt}(n,n_2,x)$ or $\mathtt{GenAt}(n_1,n,x)$.  Second, $x$
must not originate at $n$ if it is assumed non-originating.
\begin{lemma}
  \label{lemma:derivable:field:interpretation}
  Suppose that $\bnd$ $\mathcal{I}$-realizes $\skel$.  Let $T_0$ be
  the set of constants declared of sort $\sctrsc$ in $\skel$, and let
  $T\subseteq\trsc$ be the set of transcendentals that $\mathcal{I}$
  assigns to $T_0$.  Let the set $W_0$ of constants declared of sort
  $\scfld$ occurring in $\skel$, and let $W\subseteq\thefield$ be the
  image of $W_0$ under $\mathcal{I}$.

  There exists a bundle $\bnd'$ and an interpretation into $\bnd'$
  $\mathcal{I}'\models\skel$ such that $\mathcal{I}'$ maps constants
  declared with sort $\scfld$ in $\skel$ injectively to monomials of
  the form $(\prod_{x\in T} x^{d}) y$ where each $x\in T$, each
  $d\in\mathbb Z$, and each $y\in\trsc\setminus T$.
\end{lemma}
In the example above, we would like $\mathcal{I}'(w)=x^{-1}y$.  
\begin{proof}
  By assumption~\ref{Assum:monomial}, $W$ consists of monomials, and
  let $\thefield_0$ be the base field.

  By Def.~\ref{def:cover:realize},
  Clause~\ref{clause:realize:alg:ind}, the values in $W$ are
  algebraically independent.  The values in $W$ may involve the
  transcendentals $T$.  However, when we divide through by some
  $\prod_{x\in T} x^{d}$, we obtain a set $W'$, whose members are
  still algebraically independent but $T$-free.  Indeed, the map from
  $W'$ to $W$ is injective:  Otherwise, there are two members
  $w_1,w_2\in W$ such that $w_1/w_2$ is a monomial in members of $T$,
  contradicting their algebraic independence.

  The members of $W'$ are $T$-free independent monomials, rather than
  just transcendentals.  Choose a set of transcendentals
  $Y\subseteq\trsc$, disjoint from $T$, of the same cardinality as
  $W'$.  Letting $\thefield_1$ be the smallest field including
  $\thefield_0(T)$ and $W'$, Lemma~\ref{lemma:alg:ind} gives us an
  isomorphism $J\colon\thefield_1\rightarrow\thefield_0(T,Y)$.

  We let $\bnd'$ be the result of applying this $J$ to $\bnd$, and let
  $\mathcal{I}'=J\circ\mathcal{I}$.  It is clear that every regular
  strand in $\bnd'$ is an instance of the same role as its preimage in
  $\bnd$; adversary webs are preserved when we replace any derivation
  of $w\in W'$ by a one-step creation of $J(w)\in Y$.  \qed 
\end{proof}
\begin{lemma}
  \label{lemma:derivable:if:interpretation:simple}
  Suppose $\bnd$ $\mathcal{I}$-realizes $\skel$ and $\mathcal{I}$ maps
  constants declared with sort $\scfld$ in $\skel$ injectively to
  monomials of the form $(\prod_{x\in T} x^{d}) y$ where each
  $x\in T$, $d\in\mathbb Z$, and $y\in\trsc\setminus T$.

  There is a $\sigma\colon\scfld\rightarrow\term(\lang\Pi)$ that maps
  constants declared of sort $\scfld$ to terms of $\lang\Pi$ of sort
  $\scfld$, such that $\sigma(\skel)$ is derivable.

  Moreover, $\sigma$ is invertible.  
\end{lemma}
In the example above, $\sigma$ maps $w$ to $w'/x$, which cancels out
$x$ in the reception node, whose message now equals $w'$, which is
certainly derivable.  
\begin{proof}
  Suppose that $\mathcal{I}$ maps each $\syntax{x}\colon\sctrsc$ to
  $x\in\trsc$, $\syntax{w}\colon\scfld$ to
  $(\prod_{x\in T} x^{d}) y\in\thefield$, etc.  Define $\sigma$ to act
  on each $\syntax{w}$, sending it
  $\syntax{w}\mapsto(\prod_{\syntax{x}\in
    T}\syntax{x}^{d})\syntax{w'}$.  We are now interested in the
  interpretation $\mathcal{J}$ where $\mathcal{J}(\syntax{w'})=y$,
  etc.  That is, every constant of sort $\scfld$ is interpreted by a
  distinct transcendental, disjoint from those interpreting constants
  of sort $\sctrsc$.  

  Thus,
  $\mathcal{J}(\sigma(\syntax{w}))=\mathcal{J}((\prod_{\syntax{x}\in
    T}\syntax{x}^{d})\syntax{w'})=(\prod_{x\in T} x^{d}) y$, so that
  $\mathcal{J}$ applied to $\sigma(\skel)$ yields the same results as
  $\mathcal{I}$ applied to $\skel$.  Hence, $\bnd$
  $\mathcal{J}$-realizes $\sigma(\skel)$.

  Indeed, $\sigma(\skel)$ is derivable, because, for each node, we can
  use the preimage of the adversary web that derives in it $\bnd$.

  The inverse map $\tau$ divides by the product of the $\syntax{x}$s.
  The theory is unchanged to within equalities it proves.  
  \qed 
\end{proof}
\begin{theorem}
  \label{thm:derivable:realized:subst}
  $\skel$ is realized iff there exist $\sigma,\tau$ such that
  $\sigma(\skel)$ is derivable and $\tau(\sigma(\skel))$ is the same
  theory as $\skel$.
\end{theorem}
\begin{proof}
  By
  Lemmas~\ref{lemma:derivable:realized}--\ref{lemma:derivable:if:interpretation:simple}.
  \qed
\end{proof}


%% file: tests.tex


Our goal is to explain how to gradually enrich descriptions to
identify the minimal, essentially different forms that they can take,
which we call their \emph{shapes}.  To find the principles for
enrichments that lead to these shapes, we will isolate structural
characteristics that are present in all bundles.  When these
characteristics are absent, we will need to enrich a description to
add them, and the options for these enrichments provide the recipe 
driving {\cpsa}'s execution.

So let $\bnd$ be a bundle for a protocol $\Pi$.  Recall from
Def.~\ref{def:paths:1} that a \emph{unit} is a basic value or an
encryption, but not a tuple.  By an \emph{escape set}, we mean a set
$E$ of encryptions $\enc{t}K$.

The central technique of {\cpsa} is to consider how a unit $c$ that
has previously been protected by an escape set $E$ could have escaped
from its protection.

We will say that $c$ is \emph{protected by} $E$ in a particular
message $m$ iff every carried path $p=(m, \pi)$ such that
$c=\termat p$
traverses some member of $E$.  When there are no carried paths
$p=(m, \pi)$ such that $c=\termat p$ then this is true (vacuously).
We will write $\foundwithin c E m$ when $c$ is protected by $E$ in
$m$.

When $\foundwithin c E m$, any adversary that possesses $m$ and wants
to obtain $c$, or any message containing $c$ in a non-$E$ form, must
either break through the protection provided by $E$, or else build $c$
separately.

Since these notions involve only paths through the free structure of
the message algebra, they remain the same whether we are talking about
concrete messages $m\in\Algebraconstr$ or formal terms
$t\in\term(\lang{})$.  To formalize reasoning about the cases when a
value escapes, we define cuts. A \emph{cut} in a bundle is a
downward-closed set of nodes in a bundle, like a lower Dedekind cut:
\begin{definition}
  $S\subseteq\node(\bnd)$ is a \emph{cut in} $\bnd$ iff $\bnd$ is a
  bundle and, whenever $n_2\in S$ and $n_1\preceq_{\bnd}n_2$, then
  $n_1\in S$.

  When $S$ is a cut in $\bnd$, we call a node
  $n\in\node(\bnd)\setminus S$ an \emph{upper node of} $S$ in $\bnd$.
  A node $n$ is a \emph{minimal upper node of} $S$ in $\bnd$ iff it is
  an upper node of $S$, and if $n'$ is an upper node of $S$ in $\bnd$
  and $n'\preceq_\bnd n$, then $n'=n$.  
\end{definition}
\begin{lemma}
  \label{lemma:minimal:upper}
  Let $\bnd$ be a bundle, $c$ a unit, and $E$ a set of encryptions.
  \begin{enumerate}
    \item The set 
    $ S=\{n\in\node(\bnd)\colon \forall n'\preceq_\bnd n\qdot
      \foundwithin c E {\msg(n')} \} $ 
    is a cut in $\bnd$.
    \item If $S$ has upper nodes in $\bnd$, then $S$ has minimal upper
    nodes in $\bnd$.

    \item If $n_u$ is a minimal upper node of $S$ in $\bnd$, then
    $\direction(n_u)=+$ and there is a path $p=(\msg(n_u), \pi)$ such
    that $c=\msg(n_u)\termat \pi$ and $p$ traverses no member
    $e\in E$.  Moreover, if $n_S\Rightarrow^+n_u$, then $n_S\in S$, so
    $\foundwithin c E {\msg(n_S)}$.
  \end{enumerate}
\end{lemma}
We write $\esc(c,E,\bnd)$ for the minimal upper nodes of this cut,
since these nodes cause $c$ to escape from the protection of the
encryptions $E$.  We call $c$ the \emph{critical value} of the cut.
\begin{proof}
  \textbf{1.}  By the form of the definition, $S$ is downward-closed.
  \textbf{2.}  By Lemma~\ref{lemma:bundle:induction}.

  \textbf{3.}  If the direction is $-$, then the matching earlier
  transmission contradicts minimality.  The path $p$ exists by the
  definition of $\foundwithin c E {\msg(n')}$.  That $n_S\in S$ holds
  follows by the assumption of minimality.  \qed
\end{proof}
%
By examining the forms of the adversary strands, we find:
\begin{theorem}
  \label{thm:test} 
  Let: $\Pi$ be a compliant protocol; $\bnd$ be a $\Pi$-bundle;
  $n\in\node(\bnd)$; $E$ an escape set; $p=(\msg(n),\pi)$ a carried
  path that traverses no $e\in E$; and $c=\msg(n)\termat \pi$ be a
  unit.  Then there exist nodes $n_u\in\esc(c,E,\bnd)$.  Each $n_u$ is
  either an adversary creation node or satisfies at least one of these
  cases:
  \begin{description}
    \item[Regular escape:]  $n_u$ is a regular node, and for each
    $n_S\Rightarrow^+n_u$, $\foundwithin c E {\msg(n_S)}$;  
    \item[Breaking $E$:]  There is an encryption $\enc t K\in E$ and a
    transmission node $n_K$ such that $n_K\prec n_u$, and
    $\msg(n_K)=K^{-1}$;
    \item[Forging $c$:]  $c=\enc t K$ is an encryption, and there is a
    transmission node $n_K$ such that $n_K\prec n_u$, and
    $\msg(n_K)=K$;
    \item[Field element:]  $c\in\thefield$ is a field element, with
    $c$ visible in $\msg(n_u)$, and every $x\in\trsc$ present in $c$
    is visible before $n_u$;
    \item[Group element:]  $\msg(n_u)=c=\gen^\mu\in\thegroup$ is a
    group element, and $\mu=\xi\nu$ where $\nu$ is a product of field
    values visible before $n_u$, and either $\xi=\one$, or else
    $\gen^{\xi}$ is carried in a regular transmission node
    $n_r\prec n_u$.
  \end{description}
\end{theorem}
\begin{proof}
  By Lemma~\ref{lemma:minimal:upper}, an escape node $n_u$ exists.  If
  it lies on a regular strand, the first case obtains.  If $n_u$ is
  the last node of a decryption strand, the ciphertext node must be a
  member of $E$; thus, the key node receives the decryption key
  $K^{-1}$.  The matching transmission node $n_K$ satisfies the second
  case.

  If $n_u$ is the last node of an encryption strand, that strand must
  create the value $c=\enc t K$; the previous key node receives $K$,
  and the matching transmission node $n_K\prec n_u$ satisfies the
  third case.  Pairing and separation strands do not originate any
  unit, nor extract a unit from protection by $E$.

  If $c$ is a field or group element,
  Lemmas~\ref{lemma:simple:visible} or resp.~\ref{lemma:group:visible}
  establish the claim.
  \qed
\end{proof}
When $c$ is a field or group element originating on an adversary node,
then the last two cases apply, resp.  Otherwise, the first case or
second case applies, depending whether $n_u\in\esc(c,E,\bnd)$ is
regular or not.



%% file: cohorts.tex
%

When a user analyzes a protocol $\Pi$, she provides {\cpsa} with the
specification of $\Pi$, and starts the analysis with a particular
skeleton.  {\cpsa}'s job is to find the minimal, essentially different
derivable skeletons that enrich the starting
point~\cite{DoghmiGuttmanThayer07,Guttman10,cpsatheory11}.  

At each point in the analysis, {\cpsa} is operating with a
\emph{fringe}---a set of skeletons that are of interest, but have not
yet been explored---and some \emph{shapes}---which are derivable
skeletons already found.  
\begin{figure}[tb]
  \centering
  \begin{tabbing}
    procedure analyze(fringe, shapes) \\
    if \= fringe $=\emptyset$ \\
    \> then return shapes \\
    \> else \= begin choose $\skel$ from fringe in  \\
    \>\> if $\skel$ \= is derivable \\
    \>\>\> then analyze(fringe$\setminus\skel$, shapes $\cup\,\{\skel\}$) \\
    \>\>\> else analyze($($fringe $\cup$ cohort($\skel$)$)\setminus\skel$, shapes) \\
    \>\> end
  \end{tabbing}
  \caption{{\cpsa} Algorithm top level}
  \label{fig:cpsa:alg:top}
\end{figure}
The algorithm (Fig.~\ref{fig:cpsa:alg:top}) to do the analysis is
finished if the fringe is empty, in which case the answer is the set
of shapes found so far.  Otherwise, it chooses a skeleton $\skel$; if
$\skel$ is derivable it is added to the shapes.

Otherwise, {\cpsa} replaces $\skel_0$ with a set of skeletons that are
``closer'' to derivable skeletons.  In doing so, we would like to make
sure that no bundles are lost.  If $\skel_0$ covers a bundle $\bnd$,
we want to make sure that one of the ways we enrich $\skel_0$ will
continue to cover $\bnd$.  Thus, we seek a set of skeletons
$\{\skel_1,\ldots,\skel_k\}$ such that, for all $i,\bnd$:
\[  \skel_i\entails\skel_0 \qquad \skel_0\not\entails\skel_i \qquad
  \bnd\models\skel_0 \limp \exists j\qdot \bnd\models\skel_j 
\]
The first two assertions say that each $\skel_i$ adds information.
The last one says that every bundle that remains covered.  We call
such a set $\{\skel_1,\ldots,\skel_k\}$ a \emph{cohort}.

An important case is $k=0$, in which the cohort is the empty set.
Since the empty set covers no bundles, the last condition implies
$\skel_0$ also covers no bundles; we have learnt that it may be
discarded.

Thm.~\ref{thm:test} justifies {\cpsa}'s approach to generating
cohorts, which appears in Fig.~\ref{fig:cpsa:alg:cohort}.

When the cohort procedure is called, some node $n$ is not derivable;
so {\cpsa} selects one such, and moreover identifies a unit---not a
tuple---as the \emph{critical value} $t$, and an \emph{escape set}
$E$.

$t,E$ are chosen so that there is a carried path $p$ from $\msg(n)$ to
$t=\msg(n)\termat p$ such that $p$ traverses no member of $E$.
Moreover, $E$ is a set of encryptions such that, for earlier
transmissions $n_0$ preceding $n$, $t$ appears only protected within
the members of $E$ in nodes $\msg(n_0)$, which we wrote
$\foundwithin t E {\msg(n_0)}$ in Section~\ref{sec:tests}.  Moreover,
in $E$ we collect the topmost encryptions $\enc m K$ in $\msg(n_0)$
such that $t$ is carried within $m$ and the decryption key $K^{-1}$ is
not derivable by $n$ in $\skel_0$.  $E=\emptyset$ holds if $t$ was
carried nowhere in transmissions prior to $n$.  Thus, $E$ did protect
$t$ prior to $n$, but $t$ is found outside of the protection of $E$ in
$n$ itself.
\begin{figure}[tb]
  \centering
  \begin{tabbing} 
    pro\=cedure cohort($\skel$) \\
    choose $n$ from non-derivable($\node(\skel)$) in \\ 
    \> choose $t,E$ from critical-and-escape($n,\,\skel$) in \\
    \> let regular = regular-trans($t,E,n,\skel$) in \\
    \> let break = break-escape($E,n,\skel$) in \\
    \> let forge = forge-by-key($t,n,\skel$) in \\
    \> let contract = contract-away-test($t,E,n,\skel$) in \\
    \> let field = \= if $t\colon\scfld$  \\
    \>\> then \= choose $x\colon\sctrsc\in\kind{consts}(t)$ in \\
    \>\>\> $\{$ add-fld-deriv$(x,t,n,\skel)$, add-absence$(x,t,n,\skel)\}$ \\
    \>\> else $\emptyset$ in \\ 
    \> let group = \= if $t\colon\scgrp$  \\
    \>\> then \= choose $w\colon\scfld\not\in\kind{consts}(\skel)$ in \\
    \>\>\> $\{$ add-grp-deriv$(w, t,n,\skel) \}$ \\
    \>\> else $\emptyset$ in \\
    \> return $($regular $\cup$ break $\cup$ forge $\cup$ field $\cup$
    group $\cup$ contract$)$
  \end{tabbing}
  \caption{The {\cpsa} Cohort Algorithm}
  \label{fig:cpsa:alg:cohort}
\end{figure}

Prior versions of {\cpsa} computed skeletons in which an instance of a
\emph{regular} node transmits $t$ outside $E$; in which the adversary
\emph{breaks} the escape set $E$ by obtaining a decryption key
$K^{-1}$, possibly with the help of additional regular behavior; in
which the adversary may be able to \emph{forge} $t=\enc m K$ by
obtaining $K$; and in which a substitution applicable to $\skel_0$
\emph{contracts} $\skel_0$, equating encryptions on paths to $t$ with
members of $E$.

In this last case, the test in $\skel$ disappears before the covered
bundles are reached, so they do not need to offer any solution for it.

The remaining two clauses ensure that the cohort continues to cover
bundles in which field and group elements $c$ originate on adversary
nodes.  They rely on
Lemmas~\ref{lemma:simple:visible}--\ref{lemma:group:visible} to add
information to cover these {\DHKA} cases.  In each case, we express
this information by additional formulas in the resulting skeletons.

\begin{lemma}%
  \textbf{Cohort case,  Group member.} 
  \label{lemma:cohort:group}
  Let $t\colon\scgrp$ be a critical value in $n\in\node(\skel)$, and
  $w\colon\scfld$ be a constant not appearing in $\skel$.  Let $\Phi$
  be the sentence:
  \[ \mathtt{DerBy}(n,\expt(t,{1/w})) \land  \mathtt{DerBy}(n,w) . 
  \]
  The group cohort for $t,n,\skel$ is the singleton
  $\{\skel\cup\{\Phi\}\}$.  Moreover, $\skel\cup\{\Phi\}$ covers every
  $\bnd$ that $\skel$ covers.
\end{lemma}
\begin{proof}
  Suppose that $\skel\;\mathcal{I}$-covers a bundle $\bnd$, in which
  $\mathcal{I}(t)$ originates prior to $\mathcal{I}(n)$ on an
  adversary node.  Then, by Lemma~\ref{lemma:group:visible}, there is
  a monomial $\nu$ which is a product of transcendentals visible
  before $\mathcal{I}(n)$, and---assuming $(\mathcal{I}(t))^{1/\nu}$
  is distinct from $\gen$---the latter was previously visible in
  $\bnd$.  Thus, we may extend the interpretation $\mathcal{I}$ by
  assigning $w\mapsto\nu$, satisfying both conjuncts of $\Phi$.  \qed 
\end{proof}
\smallskip
\begin{lemma}
  \textbf{Cohort case,  Field member.} 
  \label{lemma:cohort:field}
  Let $t\colon\scfld$ be a critical value in $n\in\node(\skel)$; let
  $x\colon\sctrsc$ be a constant appearing in $t$.  Let $\Phi,\Psi$
  (resp.)~be the sentences:
  \[ \mathtt{DerBy}(n,x) \quad\mbox{and}\quad  \mathtt{Absent}(x,t) . 
  \]
  The field cohort $t,n,\skel$ is the pair
  $\{\skel\cup\{\Phi\},\;\skel\cup\{\Psi\}\}$.  If $\skel$ covers
  $\bnd$, then either $\skel\cup\{\Phi\}$ or $\skel\cup\{\Psi\}$
  covers $\bnd$.
\end{lemma}
\begin{proof}
  Suppose that $\skel\;\mathcal{I}$-covers a bundle $\bnd$, in which
  $\mathcal{I}(t)$ originates prior to $\mathcal{I}(n)$ on an
  adversary node.  If $\mathcal{I}(x)$ is present in the monomial
  $\mathcal{I}(t)$, then by Lemma~\ref{lemma:simple:visible}
  $\mathcal{I}(x)$ is visible in a node preceding $\mathcal{I}(n)$.
  Thus, formula $\Phi$ is satisfied in $\bnd$.  Otherwise, by
  Lemma~\ref{lemma:simple:visible}, formula $\Psi$ is satisfied in
  $\bnd$.  \qed
\end{proof}
Thus, the procedures add-fld-deriv$(x,t,n,\skel)$ and
add-absence$(x,t,n,\skel)$ add these two properties $\Phi$ and $\Psi$
respectively, and return the resulting skeletons.

%% file: ind.tex
We say that a number of values in a field are algebraically
independent when any polynomial in them that evaluates to 0 is
identically 0.
\begin{definition}
  \label{def:alg:ind}
  Suppose that $\thefield_0$ is any field.  Values
  $v_1,\ldots, v_n\in\thefield_0$ are \emph{algebraically independent}
  in $\thefield_0$ iff, for every polynomial $p$ in $n$ variables
  $x_1,\ldots, x_n$, if the value of $p$ for
  $[x_1\mapsto v_1,\ldots, x_n\mapsto v_n]$ is 0, then $p$ is the
  identically 0 polynomial.
\end{definition}
\begin{lemma}
  \label{lemma:alg:ind}
  Suppose that $\thefield_0$ is any field, and $W$ is a set of field
  extension elements.  Let $v_1,\ldots, v_n\in\thefield_0(W)$ be
  algebraically independent values in which extension elements from
  $W$ occur non-vacuously.  Let $Y=\{y_1,\ldots, y_n\}$ be field
  extension elements disjoint from $\thefield_0$ and $W$.  Let
  $\thefield_1$ be the minimal extension of $\thefield_0$ containing
  $v_1,\ldots, v_n$.

  There is a field homomorphism from $\thefield_1$ onto
  $\thefield_0(Y)$.  
\end{lemma}


%% file: unification.tex
\newtheorem{expl}{Example}

\newcommand{\typ}{\mathbin:}
\newcommand{\inc}[1]{\ensuremath{\bar{#1}}}
\newcommand{\eqq}{\overset{?}{=}}
\newcommand{\defn}{:=}
\newcommand{\eqqu}[1]{\underset{#1}{\eqq}}
\newcommand{\grp}{\ensuremath{A_+}}
\newcommand{\bel}{\ensuremath{B_b}}

\medskip\noindent This appendix describes how {\cpsa} performs
efficient unification.  The focus is solely on Diffie-Hellman
exponents, as the remainder of the algorithm is straightforward.  The
notation used is from the chapter on unification theory in Handbook of
Automated Reasoning~\cite[Chap.~5]{BaaderSnyder01}.  The reader should
have this text available while reading this section.

Terms in the exponents satisfy $\mathit{AG}$, the equations for a free
Abelian group.  To better match cited material, we use additive
notation for the group.  The binary operation is~$+$, constant~0 is
the identity element, and~$-$ is the unary inverse operation.  The
binary operation~$+$ is left associative so that $x+y+z=(x+y)+z$.  The
equational identities for an Abelian group are
\[\begin{array}{r@{{}\approx{}}llc}
x+0&x&\mbox{Unit element}&\mathcal{U}\\
x+-x&0&\mbox{Inverse element}&\mathcal{I}\\
x+y&y+x&\mbox{Commutativity}&\mathcal{C}\\
x+y+z&x+(y+z)&\mbox{Associativity}&\mathcal{A}
\end{array}\]

The normal form of term~$t$ is \(c_1 x_1 + c_2 x_2 + \cdots + c_n x_n\),
where~$c_i$ is a non-zero integer, and for $c>0$,~$cx$ abbreviates
\(x+x+\cdots+x\) summed~$c$ times while for $c<0$,~$cx$ abbreviates
$-((-c)x)$.  Every term has a normal form.

There are efficient algorithms for unification
modulo~$\mathit{AG}$~\cite[Section~5.1]{BaaderSnyder01}.  {\cpsa} uses
an algorithm that reduces the problem to finding integer solutions to
an inhomogeneous linear equation with integer coefficients.  The
equation solver used is from The Art of Computer
Programming~\cite[Pg.~327]{Knuth81}.

In the order sorted signature for our free Abelian group, there are
two sorts,~$G$ and~$B$ with $B<G$, and no operations or equations
specific to sort~$B$.

The unsorted algebra isomorphic to the sorted algebra adds a unary
inclusion operation~$b$ to the signature.  For term~$t$, we
write~$\inc{t}$ for $b(t)$.  A term~$t$ is \emph{constrained} if it
occurs as~$\inc{t}$.  Because problems in this algebra are derived
from the order-sorted algebra, terms have a restricted form.
Whenever~$\inc{t}$ occurs in a term,~$t$ is a variable.  Furthermore,
when variable~$x$ is constrained, it occurs everywhere in the term in
the form~$\inc{x}$.  For example, the term $x+\inc{x}$ is ill-formed.
Thus the normal form of a term is \(c_1 x_1 + c_2 x_2 + \cdots + c_m
x_m + c_{m+1}\inc{x}_{m+1} + c_{m+2}\inc{x}_{m+2} + \cdots
+c_n\inc{x}_n\).

\paragraph{Combination of Unification Algorithms.}  Section~6.1 of the
Handbook~\cite{BaaderSnyder01} describes a general method for
combining unification algorithms.  Let {\grp} be the theory associated
with the Abelian group, and ${\bel}\defn\{\inc{x}\approx\inc{x}\}$, so
that $E=\grp\cup\bel$.  Obviously, $=_{\bel}$ is just syntactic
equality.  The ``dummy'' identity $\inc{x}\approx\inc{x}$ ensures
that~$b$ belongs to the signature of~\bel.

An $E$-unification problem is written $t=^?_E t'$.  A \emph{reduced}
unification problem is one in which~$t$ contains no constrained
variables,~$t'$ contains no unconstrained variables, and~$t$ and~$t'$
are in normal form.  Any $E$-unification problem can converted into a
reduced $E$-unification problem.

To convert the reduced problem into decomposed form, a fresh
variable~$v_i$ is generated for each constrained variable~$\inc{x}_i$
in~$t'$.  When \(t'=c_1\inc{x}_1 + c_2\inc{x}_2 + \cdots
+c_n\inc{x}_n\), the decomposed form is
\[\{t-c_1v_1-c_2v_2-\cdots-c_nv_n\eqqu{\grp}0,
v_1\eqqu{\bel}\inc{x}_1,
v_2\eqqu{\bel}\inc{x}_2,\ldots,
v_n\eqqu{\bel}\inc{x}_n\},\]
where $x-cy$ abbreviates $x+(-c)y$.

For this problem, the shared variables are the fresh variables~$v_i$
introduced during decomposition.  Each shared variable is labeled with
theory~\bel.  The ordering of shared variables is irrelevant because
of the simple nature of the equations in theory~\bel.  The only source
of non-determinism is the partitioning of the shared variables.

\begin{expl}
  \begingroup\rm
  \[\begin{array}{ll}
  \{z\eqqu{E}\inc{x}-\inc{y}\}\\
  \{z-v_0+v_1\eqqu{\grp}0,v_0\eqqu{\bel}\inc{x},v_1\eqqu{\bel}\inc{y}\}
  &\mbox{Decomposed form}\\
  \{\{v_0\},\{v_1\}\}
  &\mbox{Partition}\\
  \{z\approx\inc{x}-\inc{y}, v_0\approx\inc{x},v_1\approx\inc{y}\}
  &\mbox{Solution}
  \end{array}\]
  \endgroup
\end{expl}

\begin{expl}
  \begingroup\rm
  \[\begin{array}{ll}
  \{0\eqqu{E}\inc{w}+\inc{x}-\inc{y}-\inc{z}\}\\
  \{v_0+v_1-v_2-v_3\eqqu{\grp}0,
  v_0\eqqu{\bel}\inc{w},
  v_1\eqqu{\bel}\inc{x},
  v_2\eqqu{\bel}\inc{y},
  v_3\eqqu{\bel}\inc{z}\}
  &\mbox{Decomposed form}\\
  \{\{v_0,v_2\},\{v_1,v_3\}\}
  &\mbox{Partition}\\
  \{w\approx y, x\approx z,\ldots\}
  &\mbox{Solution}
  \end{array}\]
  To find the other solution, $\{w\approx z, x\approx y\}$, use the
  partition $\{\{v_0,v_3\},\{v_1,v_2\}\}$.
  \endgroup
\end{expl}

\paragraph{Efficient Unification Algorithm.}
The issue surrounding an efficient implementation is the fast
generation of the fewest variable partitions required to generate a
complete set of unifiers.  A means to reduce the amount of
non-determinism in this algorithm is reported
in~\cite{KepserRichts99}.  The paper provides a method to intertwine
equation solving with making decisions about non-deterministic
choices.

The only kinds of decisions used here are equality/disequality
decisions for constrained variables.  For constrained variables~$x$
and~$y$, we write $x\doteq y$ for the decision to equate~$x$ and~$y$,
and $x\not\doteq y$ for the decision to keep them distinct.

Let $\sigma$ and $\theta$ be substitutions, and~$D$ be a set of
decisions.  To compute a complete set of unifiers for problem
$t=^?t'$, compute $\fn{unify}(t=^? t', \emptyset, \fn{Id})$ where $\fn{unify}(
t=^? t', D, \theta)$ is
\begin{enumerate}
\item if $\sigma$ is an $\grp$-unifier for $t=^? t'$ obtained by
  treating constrained variables as constants then return
  $\{\sigma\theta\}$;
\item else return $\fn{unify}'(t=^? t', D, \theta)$.
\end{enumerate}

\noindent Function $\fn{unify}'(t=^? t', D, \theta)$ is
\begin{enumerate}
\item if all distinct pairs of constrained variables have a decision
  in~$D$ then return~$\emptyset$;
\item else choose variables $x$ and $y$ without a decision in~$D$ and return
  \[\begin{array}{l}
  \fn{unify}(t=^? t'\{x\mapsto y\}, D\cup\{x\doteq y\},
  \theta\{x\mapsto y\})\cup{}\\
  \fn{unify}'(t=^? t', D\cup\{x\not\doteq y\},\theta).
  \end{array}\]
\end{enumerate}

The algorithm can further be improved by changing function \fn{unify}
so that when $\grp$-unification fails, the results of the steps up to
the failure are applied to the recursive call so as to avoid repeating
these steps.  This refinement is due to~\cite[Chapter~8]{Liu12}.  To
justify this refinement, one must show that constrained variable
identification does not invalidate any previous $\grp$-unification
steps.

To implement this refinement, we assume the equation is in reduced
form and replace \fn{unify} with $\fn{unify}_0$, which expands the
$\grp$-unification algorithm.  For equation, $t=^?t'$ let
\(t=\sum_{i=1}^n c_ix_i\) and \(t'=\sum_{j=1}^m d_j\inc{y}_j\).

\noindent Function $\fn{unify}_0(t=^? t', D, \theta)$ is
\begin{enumerate}
\item let $c_i$ be the smallest coefficient in absolute value;
\item if $c_i<0$ then return $\fn{unify}_0(-t=^? -t', D, \theta)$;
\item\label{item:solve} else if $c_i=1$ then return $\{\sigma\theta\}$
  where $\sigma=\{x_i\mapsto\sum_j^md_j\inc{y}_j -
  \sum_{k\neq i}^nc_kx_k\}$;
\item else if $c_i$ divides every coefficient in $c$ then
  \begin{enumerate}
  \item if $c_1$ divides every coefficient in $d$ then divide $c$ and
    $d$ by $c_i$ and goto step~\ref{item:solve};
  \item else return $\fn{unify}'(t=^? t', D, \theta)$;
  \end{enumerate}
\item else eliminate $x_i$ in favor of freshly created
  variable~$x_{n+1}$ using  $\sigma=\{x_i\mapsto x_{n+1}-\sum_{k\neq
    i} (c_k/c_i)x_k\}$ and then
  return $\fn{unify}_0(t''=^? t', D, \sigma\theta)$ where
  $t''=c_ix_{n+1}+\sum_k(c_k\bmod c_i)x_k$.
\end{enumerate}

Constrained variable identification does not invalidate any previous
$\grp$-unification steps because $\gcd(d_i,d_j)$ divides $d_i+d_j$.

There is another opportunity for performance improvement by refining
function $\fn{unify}'$.  It could inspect the coefficients of
constrained variables when deciding which constrained variables to
identify.  This refinement has yet to be explored.


%% file: dh_ebn.bbl
\begin{thebibliography}{10}

\bibitem{ankney1995unified}
R.~Ankney, D.~Johnson, and M.~Matyas.
\newblock {The Unified Model}. contribution to {ANSI X9F1}.
\newblock {\em Standards Projects (Financial Crypto Tools), ANSI X}, 42, 1995.

\bibitem{BaaderSnyder01}
Franz Baader and Wayne Snyder.
\newblock Unification theory.
\newblock In Alan Robinson and Andrei Voronkov, editors, {\em Handbook of
  Automated Reasoning}, volume~1, chapter~8. The {MIT} Press, 2001.
\newblock \url{http://www.cs.bu.edu/~snyder/publications/UnifChapter.pdf‎}.

\bibitem{BartheFFMSS14}
Gilles Barthe, Edvard Fagerholm, Dario Fiore, John~C. Mitchell, Andre Scedrov,
  and Benedikt Schmidt.
\newblock Automated analysis of cryptographic assumptions in generic group
  models.
\newblock In {\em {CRYPTO}}, volume 8616 of {\em LNCS}, pages 95--112.
  Springer, 2014.

\bibitem{BasinCM13}
David~A. Basin, Cas Cremers, and Simon Meier.
\newblock Provably repairing the {ISO/IEC} 9798 standard for entity
  authentication.
\newblock {\em Journal of Computer Security}, 21(6):817--846, 2013.

\bibitem{BellareRogaway93}
Mihir Bellare and Phillip Rogaway.
\newblock Entity authentication and key distribution.
\newblock In {\em Advances in Cryptology -- Crypto '93 Proceedings}, number 773
  in LNCS. Springer-Verlag, 1993.

\bibitem{Blanchet01}
Bruno Blanchet.
\newblock An efficient protocol verifier based on {Prolog} rules.
\newblock In {\em 14th Computer Security Foundations Workshop}, pages 82--96.
  IEEE CS Press, June 2001.

\bibitem{Cremers11keyexchange}
Cas Cremers.
\newblock Key exchange in {IPsec} revisited: Formal analysis of {IKEv1} and
  {IKEv2}.
\newblock In {\em Computer Security--ESORICS 2011}. Springer, 2011.

\bibitem{cremers2012operational}
Cas Cremers and Sjouke Mauw.
\newblock {\em Operational semantics and verification of security protocols}.
\newblock Springer, 2012.

\bibitem{DiffieHellman76}
W.~Diffie and M.~Hellman.
\newblock New directions in cryptography.
\newblock {\em IEEE Transactions on Information Theory}, 22(6):644--654,
  November 1976.

\bibitem{DoghmiGuttmanThayer07}
Shaddin~F. Doghmi, Joshua~D. Guttman, and F.~Javier Thayer.
\newblock Searching for shapes in cryptographic protocols.
\newblock In {\em Tools and Algorithms for Construction and Analysis of Systems
  {(TACAS)}}, number 4424 in LNCS, pages 523--538, 2007.

\bibitem{DoughertyGuttman12}
Daniel~J. Dougherty and Joshua~D. Guttman.
\newblock An algebra for symbolic {Diffie-Hellman} protocol analysis.
\newblock In {\em Trustworthy Global Computing}, volume 8191 of {\em LNCS},
  pages 164--181, 2012.

\bibitem{DoughertyGuttman2014}
Daniel~J. Dougherty and Joshua~D. Guttman.
\newblock Decidability for lightweight {Diffie-Hellman} protocols.
\newblock In {\em {IEEE} Symposium on Computer Security Foundations}, 2014.

\bibitem{DBLP:conf/fosad/EscobarMM07}
Santiago Escobar, Catherine Meadows, and Jos\'e Meseguer.
\newblock Maude-{NPA}: Cryptographic protocol analysis modulo equational
  properties.
\newblock In {\em Foundations of Security Analysis and Design {V}, {FOSAD}
  2007--2009 Tutorial Lectures}, volume 5705 of {\em Lecture Notes in Computer
  Science}, pages 1--50. Springer, 2009.

\bibitem{Gentzen35}
G.~Gentzen.
\newblock Investigations into logical deduction (1935).
\newblock In {\em The Collected Works of Gerhard Gentzen}. North Holland, 1969.

\bibitem{GoguenMeseguer92}
Joseph~A. Goguen and Jos\'{e} Meseguer.
\newblock Order-sorted algebra {I}: equational deduction for multiple
  inheritance, overloading, exceptions and partial operations.
\newblock {\em Theoretical Computer Science}, 105(2):217--273, 1992.

\bibitem{Guttman10}
Joshua~D. Guttman.
\newblock Shapes: Surveying crypto protocol runs.
\newblock In Veronique Cortier and Steve Kremer, editors, {\em Formal Models
  and Techniques for Analyzing Security Protocols}, Cryptology and Information
  Security Series. {IOS} Press, 2011.

\bibitem{Guttman12}
Joshua~D. Guttman.
\newblock State and progress in strand spaces: Proving fair exchange.
\newblock {\em Journal of Automated Reasoning}, 48(2):159--195, 2012.

\bibitem{Guttman14}
Joshua~D. Guttman.
\newblock Establishing and preserving protocol security goals.
\newblock {\em Journal of Computer Security}, 22(2):201--267, 2014.

\bibitem{GuttmanLRR15}
Joshua~D. Guttman, Moses~D. Liskov, John~D. Ramsdell, and Paul~D. Rowe.
\newblock Formal support for standardizing protocols with state.
\newblock In {\em Security Standardisation Research {SSR}}, pages 246--265,
  2015.
\newblock Extended version at \url{http://arxiv.org/abs/1509.07552}.

\bibitem{GuttmanThayer02}
Joshua~D. Guttman and F.~Javier {{Thayer}}.
\newblock Authentication tests and the structure of bundles.
\newblock {\em Theoretical Computer Science}, 283(2):333--380, June 2002.

\bibitem{KepserRichts99}
Stephan Kepser and J\"{o}rn Richts.
\newblock Optimisation techniques for combining constraint solvers.
\newblock In Maarten de~Rijke and Dov Gabbay, editors, {\em Frontiers of
  Combining Systems 2}, pages 193--210, Asmsterdam, 1999. Research Studies
  Press/Wiley.
\newblock
  \url{http://tcl.sfs.uni-tuebingen.de/~kepser/papers/optimisation.ps.gz}.

\bibitem{Knuth81}
Donald~E. Knuth.
\newblock {\em The Art of Computer Programming, Volume {II:} Seminumerical
  Algorithms, 2nd Edition}.
\newblock Addison-Wesley, 1981.

\bibitem{kuesters2009using}
Ralf K{\"u}sters and Tomasz Truderung.
\newblock Using {ProVerif} to analyze protocols with {Diffie-Hellman}
  exponentiation.
\newblock In {\em {IEEE} Computer Security Foundations Symposium}, pages
  157--171. IEEE, 2009.

\bibitem{LiskovThayer14}
Moses Liskov and F.~Javier Thayer.
\newblock Modeling {Diffie-Hellman} derivability for automated analysis.
\newblock In {\em {IEEE} Computer Security Foundations}, pages 232--243, 2014.

\bibitem{cpsatheory11}
Moses~D. Liskov, Paul~D. Rowe, and F.~Javier Thayer.
\newblock Completeness of {CPSA}.
\newblock Technical Report MTR110479, The MITRE Corporation, March 2011.
\newblock
  \url{http://www.mitre.org/publications/technical-papers/completeness-of-cpsa}.

\bibitem{Liu12}
Zhiqiang Liu.
\newblock {\em Dealing Efficiently with Exclusive OR, Abelian Groups and
  Homomorphism in Cryptographic Protocol Analysis}.
\newblock PhD thesis, Clarkson University, September 2012.
\newblock \url{http://people.clarkson.edu/~clynch/papers/}.

\bibitem{Lowe97}
Gavin Lowe.
\newblock A hierarchy of authentication specifications.
\newblock In {\em 10th Computer Security Foundations Workshop Proceedings},
  pages 31--43. {IEEE} CS Press, 1997.

\bibitem{MarreroEtAl97}
Will Marrero, Edmund Clarke, and Somesh Jha.
\newblock A model checker for authentication protocols.
\newblock In Cathy Meadows and Hilary Orman, editors, {\em Proceedings of the
  {DIMACS} Workshop on Design and Verification of Security Protocols}.
  {DIMACS}, Rutgers University, September 1997.

\bibitem{Maurer05}
Ueli~M. Maurer.
\newblock Abstract models of computation in cryptography.
\newblock In {\em Cryptography and Coding}, volume 3796 of {\em LNCS}, pages
  1--12. Springer, 2005.

\bibitem{MeierSCB13}
Simon Meier, Benedikt Schmidt, Cas Cremers, and David~A. Basin.
\newblock The {tamarin} prover for the symbolic analysis of security protocols.
\newblock In {\em Computer Aided Verification ({CAV})}, pages 696--701, 2013.

\bibitem{MoedersheimKatsoris14}
Sebastian M{\"{o}}dersheim and Georgios Katsoris.
\newblock A sound abstraction of the parsing problem.
\newblock In {\em {IEEE} {CSF}}, pages 259--273, 2014.

\bibitem{Paulson98}
Lawrence~C. Paulson.
\newblock The inductive approach to verifying cryptographic protocols.
\newblock {\em Journal of Computer Security}, 1998.
\newblock Also Report 443, Cambridge University Computer Lab.

\bibitem{agum09}
John~D. Ramsdell.
\newblock {AGUM}: Unification and matching in an abelian group, 2009.
\newblock \url{http://hackage.haskell.org/package/agum}.

\bibitem{Ramsdell12}
John~D. Ramsdell.
\newblock Deducing security goals from shape analysis sentences.
\newblock The MITRE Corporation, April 2012.
\newblock \url{http://arxiv.org/abs/1204.0480}.

\bibitem{cpsa09}
John~D. Ramsdell and Joshua~D. Guttman.
\newblock {CPSA}: A cryptographic protocol shapes analyzer, 2009.
\newblock \url{http://hackage.haskell.org/package/cpsa}.

\bibitem{RoweEtAl2015}
Paul~D. Rowe, Joshua~D. Guttman, and Moses~D. Liskov.
\newblock Measuring protocol strength with security goals.
\newblock Submitted to \emph{IJIS} in the SSR 2014 special issue. Available at
  \url{http://web.cs.wpi.edu/~guttman/pubs/ijis_measuring-security.pdf}, April
  2015.

\bibitem{RoweEtAl2016}
Paul~D. Rowe, Joshua~D. Guttman, and Moses~D. Liskov.
\newblock Measuring protocol strength with security goals.
\newblock {\em International Journal of Information Security}, 15(6):575--596,
  November 2016.
\newblock DOI 10.1007/s10207-016-0319-z,
  \url{http://web.cs.wpi.edu/~guttman/pubs/ijis_measuring-security.pdf}.

\bibitem{Cremers12}
Benedikt Schmidt, Simon Meier, Cas Cremers, and David~A. Basin.
\newblock Automated analysis of {Diffie-Hellman} protocols and advanced
  security properties.
\newblock {\em Computer Security Foundations, ({CSF})}, pages 25--27, 2012.

\bibitem{Shoup97}
Victor Shoup.
\newblock Lower bounds for discrete logarithms and related problems.
\newblock In {\em {EUROCRYPT}}, volume 1233 of {\em Lecture Notes in Computer
  Science}, pages 256--266. Springer, 1997.

\bibitem{ThayerHerzogGuttman99}
F.~Javier Thayer, Jonathan~C. Herzog, and Joshua~D. Guttman.
\newblock Strand spaces: Proving security protocols correct.
\newblock {\em Journal of Computer Security}, 7(2/3):191--230, 1999.

\end{thebibliography}
